\documentclass[12pt,reqno,hidelinks]{amsart}
\usepackage{amssymb}
\usepackage{amsmath, amssymb, verbatim, url, mathrsfs}
\usepackage{mathrsfs}
\usepackage{dutchcal}
\usepackage{mathtools}
\usepackage{verbatim}
\usepackage{amsthm}
\usepackage{framed}
\usepackage{cite}
\usepackage{wasysym}
\usepackage{upgreek}
\usepackage{color}
\usepackage[dvipsnames]{xcolor}
\usepackage{tensor}
\usepackage{accents}
\usepackage{dsfont}
\usepackage[colorlinks,linkcolor=blue,citecolor=blue]{hyperref}
\usepackage{enumerate}
\usepackage{enumitem}
\usepackage[normalem]{ulem}
\usepackage{longtable}
\usepackage{mathtools}

\numberwithin{equation}{section}

\newtheorem{proposition}{Proposition}[section]
\newtheorem{lemma}[proposition]{Lemma}
\newtheorem{corollary}{Corollary}[section]
\newtheorem{theorem}{Theorem}[section]

\newtheorem{remark}{Remark}[section]

\newtheorem{definition}[proposition]{Definition}

\newtheorem*{theorem*}{Theorem}
\newtheorem*{mquestion*}{Main Question}
\newtheorem*{claim*}{Claim}

\newcommand{\mg}{\gamma}
\newcommand{\uu}{\underline{u}}

\newcommand{\uuo}{\underline{u_0}}

\newcommand{\ca}{\mathcal{a}}
\newcommand{\ba}{\bar{\mathcal{a}}}
\newcommand{\cb}{\mathcal{b}}
\newcommand{\cg}{\mathcal{g}}

\newcommand{\cc}{\mathcal{c}}

\newcommand{\bc}{\bar{\mathcal{c}}}

\newcommand{\mf}{\beta}

\newcommand{\vertiiii}[1]{{\left\vert\kern-0.25ex\left\vert\kern-0.25ex\left\vert\kern-0.25ex\left\vert #1 \right\vert\kern-0.25ex\right\vert\kern-0.25ex\right\vert\kern-0.25ex\right\vert}}
\newcommand{\vertiii}[1]{{\left\vert\kern-0.25ex\left\vert\kern-0.25ex\left\vert #1 \right\vert\kern-0.25ex\right\vert\kern-0.25ex\right\vert}}

\newcommand{\tv}{\tilde{v}}
\newcommand{\hv}{\hat{v}}

\newcommand{\cv}{\check{v}}
\newcommand{\cuv}{\check{\dwave{v}}}
\newcommand{\trho}{\tilde{\rho}}
\newcommand{\hrho}{\hat{\varrho}}
\newcommand{\tphi}{\tilde{\phi}}
\newcommand{\hphi}{\hat{\phi}}

\newcommand{\rphi}{\mathring{\phi}}
\newcommand{\rp}{\mathring{p}}
\newcommand{\rrho}{\mathring{\rho}}
\newcommand{\rv}{\mathring{v}}

\newcommand{\uf}{\underline{f}}

\newcommand{\Rbb}{\mathbb{R}}
\newcommand{\Zbb}{\mathbb{Z}}
\newcommand{\Tbb}{\mathbb{T}}
\newcommand{\Pbb}{\mathbb{P}}

\newcommand{\del}[1]{{\partial_{#1}}}

\newcommand{\AND}{{\quad\text{and}\quad}}

\newcommand{\Li}{L^\infty}

\newcommand{\bx}{\mathbf{x}}

\newcommand{\p}[1]{
	\begin{pmatrix}
		#1
	\end{pmatrix}
}

\newcommand{\dwave}[1]{\mbox{\uwave{$#1$}}}

\newcommand{\tr}{\text{tr}}

\newcommand{\Rho}{\mbox{\uwave{$\varrho$}}}
\newcommand{\V}{\dwave{v}}

\newcommand{\B}{\mathcal{B}}

\newcommand{\U}{\mathcal{U}}

\newcommand{\Pbp}{\mathbb{P}^{\perp}}

\newcommand{\be}{\begin{equation}}
	\newcommand{\ee}{\end{equation}}

\allowdisplaybreaks

\begin{document}

\title[Nonlinear gravitational instabilities for Newtonian universes]{Fully nonlinear gravitational instabilities for expanding spherical symmetric Newtonian universes with inhomogeneous density and pressure}

\author{Chao Liu}

\address[Chao Liu]{Center for Mathematical Sciences and School of Mathematics and Statistics, Huazhong University of Science and Technology, Wuhan 430074, Hubei Province, China.}
\email{chao.liu.math@foxmail.com}

\begin{abstract}
Nobel Prize laureate P.J.E. Peebles \cite{Peebles2020} has emphasized the importance and difficulties of studying the \textit{large scale clustering of matter} in cosmology.
Nonlinear gravitational instability plays a central role in understanding the clustering of matter and the formation of nonlinear structures in the universe and stellar systems.  However, there is no rigorous result on the nonlinear analysis of this instability except for some particular exact solutions without pressure, and numerical and phenomenological approaches.   Both Rendall \cite{Rendall2005}  and Mukhanov \cite{ViatcehslavMukhanov2013}  have highlighted the challenge posed by nonlinear gravitational instability with effective pressure. This has been a longstanding open problem in astrophysics for over a century since the occurrence of linearized Jeans instabilities in Newtonian universes in 1902. This article contributes to a fully nonlinear analysis of the gravitational instability for \textit{the Euler--Possion system} which models expanding Newtonian universes with inhomogeneous density and pressure.
The exponential or finite-time increasing blowups of the density contrast $\varrho$  can be determined, which may account for the considerably faster growth rate of nonlinear structures observed in astrophysics than that suggested by the classical Jeans instability. We believe this is the first rigorous result for the nonlinear Jeans instability with effective pressure and the method is concise and robust.

 \vspace{2mm}


\vspace{2mm}

\end{abstract}

	\maketitle
	
	\setcounter{tocdepth}{2}
	
	\pagenumbering{roman} \pagenumbering{arabic}

\section{Introduction}
The gravitational instability, which Nobel laureate P.J.E. Peebles extensively discusses and highlights in his renowned book \cite{Peebles2020}, is of tremendous significance in astrophysics. It characterizes the \textit{clustering of matter at large scales} (as opposed to an isolated singularity), the mass accumulation of self-gravitating systems, and enhances the understanding of the formation of stellar systems and nonlinear structures in the universe\footnote{Additionally, it is worth noting the tenth inquiry in John Baez's compilation of open questions in cosmology and astrophysics \cite{Baez2020}, which asks, ``10. Why are the galaxies distributed in clumps and filaments?
	'' }. However, up to now, the gravitational instability was only rigorously studied in the \textit{linear regime} since the first linearized gravitational instability was obtained by Jeans \cite{Jeans1902} for Newtonian gravity in $1902$ (thus it is also known as \textit{Jeans instability}), and then generalized to general relativity by Lifshitz \cite{Lifshitz1946} in $1946$. The Jeans instability was extended to the expanding background universe in Bonnor \cite{Bonnor1957} (also see \cite{Zeldovich1971,ViatcehslavMukhanov2013}). Revised: The linearized Jeans instability derivations become substantially distorted by the accretions of the mass, as the increasing density results in significant deviations from the linear regime. Additionally, the growth rate predicted by the classical linearized Jeans instability fails to account for the considerable inhomogeneities observed in the present-day universe and the formation of galaxies. This growth rate is deemed too slow and, as a result, less efficient. 
Previous studies by \cite{Bonnor1957,Zeldovich1971,ViatcehslavMukhanov2013}, and our own work \cite{Liu2022} demonstrate that the growth rate of density for certain partially nonlinear or linearized Jeans instability models in expanding Newtonian universe follows a power law of $\sim t^\frac{2}{3}$, where $t$ denotes time.  Due to the incompleteness of current astrophysical theories and the presence of numerous mysterious phenomena that may rely on nonlinear Jeans instability, it is imperative to conduct a thorough study of the nonlinear analysis of Jeans instability. Furthermore, it is worth noting that Rendall has pointed out that there are no existing results on Jeans instability for the fully nonlinear case \cite{Rendall2005}. Similarly, Mukhanov has emphasized the challenge of studying nonlinear gravitational instability in Section $6.4$ of their work \cite{ViatcehslavMukhanov2013}. It becomes a long-standing open problem presented by astrophysics and more backgrounds can be found in, for instance, \cite{Bonnor1957,Zeldovich1971,Peebles2020,ViatcehslavMukhanov2013} and our prior articles \cite{Liu2022,Liu2022b,Liu2023a}.

Several references (see \cite{Zeldovich1971,ViatcehslavMukhanov2013,Arbuzova2014,Sciama1955}) have discussed nonlinear strategies that involve approximations and numerical methods, such as the well-known \textit{Zel'dovich solutions}. Additionally, an \textit{exact} solution known as the \textit{Tolman solution} (see \cite[\S$6.4.1$]{ViatcehslavMukhanov2013}, \cite{LANDAU1975}) describes the evolution and collapse of inhomogeneities in a spherical, \textit{dust}-filled fluid, however, it is unable to incorporate non-vanishing pressure effects. Our method is the first of its kind, addressing effective pressure cases and allowing for generalization to more complex scenarios.

This article builds upon our previous works \cite{Liu2022b} and \cite{Liu2023a}. On one hand, in \cite{Liu2022b}, we presented a second-order hyperbolic equation as a simplified model of the nonlinear Jeans instability. We demonstrated that its solution experiences an exponentially self-increase blowup, which is a stable ODE-type blowup according to Alihnac \cite{Alinhac1995}. This was achieved by using a Fuchsian formulation constructed through proper time transforms. However, the current article faces additional challenges that cannot be addressed directly using this result.
On the other hand,  our previous work described in \cite{Liu2023a} involved the construction of a family of particular solutions to the Euler-Poisson system that exhibit nonlinear gravitational instability of matter in an expanding Newtonian universe, with inhomogeneous pressure and entropy – known as the cold center and hot rim (for a more detailed explanation, refer to \cite{Liu2023a}). These unstable solutions, specifically \cite[eqs. $(3.2)$--$(3.5)$]{Liu2023a}, serve as reference solutions in this article.

It is well known that the expanding Newtonian universe can be characterized by the Euler--Poisson system and the equation of state in \cite[eqs. $(2.1)$--$(2.5)$]{Liu2023a}. We work on the following \textit{dimensionless and normalized Euler--Poisson system}\footnote{We have used notation $\partial^i:=\delta^{ij}\del{j}$ and see \S\ref{iandc}. } (we use
  the Einstein summation convention),
\begin{align}
	\del{t}\rho+\del{i}(\rho v^i) =& 0 ,\label{e:EP1}\\
	\del{t}v^i +v^j\del{j} v^i+\frac{\del{}^i p}{\rho}+\del{}^i\phi = & \mathcal{D}^i(t,x^j,\rho,v^k,s,\phi) ,\label{e:EP2}\\
	\del{t} s+v^i\del{i}s=& \mathcal{S}(t,x^j,\rho,v^k,s,\phi),\label{e:EP2b}\\
	\Delta\phi=\delta^{ij}\del{i}\del{j} \phi =&  4\pi \rho.     \label{e:EP3}
\end{align}
where $\rho$, $ v^i$, $p$, $\phi$ and $s$ are the dimensionless density, velocities, pressure of the fluids, gravitational potential and specific entropy, $\mathcal{D}^i$ and $\mathcal{S}$ are sources of specific momentum and entropy.
We point out that the non-dimensionalization of the Euler--Poisson system can be found in \cite[Appendix A]{Liu2023a} and the arbitrary initial time $t=t_0$ can be \textit{normalized} to $t=1$.
In addition, \textit{the equation of state} of the fluids (see \cite[Appendix A]{Liu2023a}) is assumed by
\begin{equation}
	p =   K e^{s} \rho^{\frac{4}{3} } \quad \text{for} \quad  K>0.   \label{e:eos}
\end{equation}
The \textit{initial data}, in this article, is given by \textit{spherical symmetric functions} $\beta\mathcal{d}(|\bx|)$ and $\gamma \mathcal{v}(|\bx|)$,
\begin{gather}
	\rho|_{t=1}(x^i)=\rho_0(|\bx|):= \frac{\iota^3}{6\pi}(1+\beta \mathcal{d}(|\bx|)) ,  \quad v^i|_{t=1}(x^i)=v^i_0(|\bx|):= \frac{2}{3}x^i +\gamma \mathcal{v}(|\bx|) x^i , \label{e:data0}\\ s|_{t=1}(x^i)=s_0(|\bx|):=\ln \Bigl(  \frac{(1+\beta \mathcal{d}(|\bx|))^{\frac{2}{3}+\omega}}{(1+\beta)^{\omega}} |\bx|^2 \Bigr). \label{e:data0b}
\end{gather}
where   $|\bx|^2:=\delta_{kl}x^k x^l$, $\beta$,  $\gamma$ and $\omega$ are given constants, and $\iota$ is a constant determined by
\begin{equation*}\label{e:iota}
	\iota:=	\iota(\tilde{K})=
	\Bigl( \frac{1}{2}\sqrt{1+	18\tilde{K}} +\frac{1}{2}\Bigr)^{\frac{1}{3}} - \Bigl(\frac{1}{2}\sqrt{1+18	\tilde{K}}-\frac{1}{2}\Bigr)^{\frac {1}{3}}\in (0,1)  \AND \tilde{K}:=\frac{K^3 }{\pi  } .
\end{equation*}
Essentially, $\iota^3$ and $\tilde{K}$ depend on the mater itself (see \cite{Liu2023a} for details).
In this article, we only consider $\iota^3\in (0, 1/5]$.

\subsection{Reviews on the homogeneous models}
Before illustrating the main question of this article, as preparations, let us first review two crucial solutions to the Euler--Poisson system \eqref{e:EP1}--\eqref{e:eos} for
\begin{equation*}
	\mathcal{D}^i(t,x^j,\rho,v^k,s,\phi)=\mathcal{S}(t,x^j,\rho,v^k,s,\phi)=0, \quad \text{(conservations of momentum and energy)}
\end{equation*}
and some notable facts from \cite{Liu2023a}:
\begin{enumerate}[leftmargin=*,label={(R\arabic*)}]
	\item  \label{c:Ntsl} \underline{Newtonian universes (backgroud solutions):} If $\beta=\gamma=0$, then the initial data \eqref{e:data0}--\eqref{e:data0b} reduce to
	\begin{equation}\label{e:data1}
		\rho|_{t=1}(x^i)= \frac{\iota^3}{6\pi}, \quad v^i|_{t=1}(x^i)= \frac{2}{3}x^i \AND s|_{t=1}(x^i) =\ln |\bx|^2  .
	\end{equation}
	Then according to \cite[\S$2$]{Liu2023a}, the solution to the Euler--Possion system \eqref{e:EP1}--\eqref{e:eos} is
	\begin{gather}
		\rrho(t)= \frac{\iota^3}{6\pi  t^2}, \quad \rp(t) = K t^{-\frac{4}{3}} \delta_{kl} x^k x^l  \rrho^{\frac{4}{3}} ,  \quad
		\rv^i(t,x^k)=\frac{2}{3t}x^i, \label{e:exsol1} \\ \rphi(t,x^k)=\frac{2}{3}\pi  \rrho \delta_{ij} x^i x^j =\frac{\iota^3}{9t^2}  \delta_{ij} x^i x^j  \AND \mathring{s}(t,x^k)=\ln (t^{-\frac{4}{3}} \delta_{kl} x^k x^l) .    \label{e:exsol2}
	\end{gather} 
This solution \eqref{e:exsol1}--\eqref{e:exsol2} gives an expanding Newtonian universe with homogeneous density and a  paraboloid temperature field (see \cite[Appendix A]{Liu2023a}).

To facilitate our analysis of the behavior of perturbed variables that deviate from the background Newtonian universe \eqref{e:exsol1}--\eqref{e:exsol2}, we first need to decompose the variables $(\rho,v^i,p,s,\phi)$ into two parts: the exact background solution $(\rrho,\rv^i,\rp,\mathring{s}, \rphi)$ determined by \eqref{e:exsol1}--\eqref{e:exsol2} and the perturbed parts $(\trho,\tv^i,\tilde{p},\tilde{s}, \tphi)$. Then, we can introduce a density contrast $\varrho$ defined as follows:
\begin{gather}\label{e:perv}
	\rho =\rrho + \trho,\quad
	v^i= \rv^i +\tv^i,\quad \phi=\rphi + \tphi, \quad
	p =\rp +  \tilde{p}, \quad 	s =\mathring{s}+  \tilde{s}  \AND \varrho:= \frac{\trho}{\rrho}  .  
\end{gather}

	\item  \label{c:hmsl} \underline{Homogeneous blowup solutions (reference solutions):} If constants $\beta>0$, $\gamma>0$ and $\mathcal{d}(|\bx|)=1$,  $\mathcal{v}(|\bx|)=-1$, then the initial data becomes
	\begin{equation}\label{e:data2}
		\rho|_{t=1}(x^i)= \frac{\iota^3}{6\pi}(1+\beta), \quad v^i|_{t=1}(x^i)= \Bigl( \frac{2}{3}  -\gamma \Bigr) x^i \AND s|_{t=1}(x^i)= \ln \bigl(  (1+\beta  )^{\frac{2}{3} }  |\bx|^2 \bigr).
	\end{equation}
	According to \cite[\S$3$]{Liu2023a}, there is a solution to the Euler--Poisson system \eqref{e:EP1}--\eqref{e:eos} given by
	\begin{gather}
		\rho_r(t) =  \frac{\iota^3(1+f(t) )}{6\pi  t^2} ,\quad 	v^i_r  (t,x^i) =\frac{2}{3t}x^i - \frac{ f^\prime(t)}{3 (1+f(t))} x^i  ,   \label{e:sl1}\\
		\phi_r(t,x^i)  = \frac{\iota^3(1+f(t)) |\boldsymbol{x}|^2}{9  t^{2}}  \AND
		s_r(t, x^k)=\ln \bigl( t^{-\frac{4}{3}} (1+f)^{\frac{2}{3}} \delta_{kl} x^k x^l\bigr) ,   \label{e:sl2}
	\end{gather}
and the density contrast
$\varrho_r(t)=f(t)$,
where the density contrast
$f(t)$ is a given function by \cite[$(3.6)$--$(3.7)$]{Liu2023a}, i.e., $f(t)$ solves an ordinary differential equation (ODE),
\begin{gather}
	f^{\prime\prime}(t)+\frac{4}{3t} f^\prime(t)-\frac{2}{3 t^2} f(t)(1+f(t)) -\frac{4(  f^\prime(t))^2}{3(1+f(t))}= 0 ,\label{e:feq1b} \\
	f|_{t=1}=\mf \AND f^\prime|_{t=1}=3(1+\mf)\mg. \label{e:feq2b}
\end{gather}
This equation \eqref{e:feq1b}--\eqref{e:feq2b} has been elaborated in our previous article \cite[Theorem $1.1$]{Liu2022b} and we have estimates of $f(t)$. For readers' convenience, we have listed the results in Appendix \ref{s:ODE}.
This solution \eqref{e:sl1}--\eqref{e:sl2} serves as a \textit{reference solution} in this article.

	\item  Note $\iota$ and $\tilde{K}$ are both \textit{dimensionless quantity} depending on the molar mass (see \cite[Appendix A]{Liu2023a} for proofs) and one can verify that $\iota$ satisfies an important identity (also see \cite[\S$2$]{Liu2023a}) that
\begin{align}\label{e:ioeq}
	\iota^3+ 9 \Bigl(\frac{\tilde{K}}{6}\Bigr)^{\frac{1}{3}}   \iota -1=0 ,
\end{align}
and $\iota(\tilde{K})$ is a decreasing function, and $
\lim_{\tilde{K}\rightarrow 0} \iota(\tilde{K})=1$ and $ \lim_{\tilde{K}\rightarrow +\infty} \iota(\tilde{K})=0$.

\item  We only consider the case for $K>0$ since for the dust ($K=0$), it is reduced to the \textit{Tolman solutions} (see \cite[\S$6.4$]{ViatcehslavMukhanov2013}) which have been well studied already. However, for $K>0$, there is no further result except \cite{Liu2023a} for homogeneous density contrast.
\end{enumerate}

\subsection{Questions} 
We begin by noting that according to \cite{Liu2023a}, the Newtonian universe \eqref{e:exsol1}--\eqref{e:exsol2} exhibits a density decay with the time that behaves like $\sim 1/t^2$. However, when we introduce a homogeneous perturbation characterized by two positive parameters $\beta$ and $\gamma$ in the initial data for density and velocity, as described in equation \eqref{e:data2}, the solution becomes unstable. This means that as time goes on, the solution  \eqref{e:sl1}--\eqref{e:sl2} can never become close to the Newtonian universe \eqref{e:exsol1}--\eqref{e:exsol2}. In fact, as we can see from the expression \eqref{e:sl1}--\eqref{e:sl2}, both the density $\rho_r(t)$ and the density contrast $\varrho_r(t)$ will eventually increase with time since $f(t)$ grows faster than an exponential function. The estimates for $f(t)$ can be found in \cite[eqs. $(3.8)$ and $(3.9)$]{Liu2023a}, or in Theorem \ref{t:mainthm0} in Appendix \ref{s:ODE}.

While we have successfully constructed a family of unstable solutions in \eqref{e:sl1}--\eqref{e:sl2} in our previous work \cite{Liu2023a}, it should be noted that these solutions require homogeneous initial perturbations of the density, which are characterized by two parameters $\beta$ and $\gamma$, as described in that work. This raises the question of what happens if the initial perturbations of the density are not homogeneous, as in \eqref{e:data0}. To be more specific, we need to ask:
\begin{mquestion*}\
	\begin{enumerate}[leftmargin=*]
		\item Given an initially inhomogeneous but spherically symmetric perturbation \eqref{e:data0}, deviating from the data \eqref{e:data1}, we need to investigate whether the Newtonian universe solution to the Euler-Poisson system \eqref{e:EP1}--\eqref{e:eos} is stable or unstable. 
		\item If the Newtonian universe solution is unstable, we need to examine how the perturbation solution behaves in response to the initial perturbation. 
		\item In the case of a perturbation solution blow-up, we also need to investigate how fast the perturbation solution grows.
		
	\end{enumerate}   
\end{mquestion*}

In this article, we utilize the \textit{reference solution} \eqref{e:sl1}--\eqref{e:sl2} to develop a family of inhomogeneous but spherically symmetric solutions to the Euler--Poisson system \eqref{e:EP1}--\eqref{e:eos} that demonstrate exponential \textit{self-increasing density} on $\mathbb{R}^3$ until blowups, as a means of answering the question at hand. Furthermore, if the parameter $\gamma>1/3$ in the data \eqref{e:data0}, the density contrast and its derivative will experience blowups at a finite time. This implies that the Newtonian universe is unstable when subjected to initially inhomogeneous but spherically symmetrical perturbations and that these perturbations are close to their corresponding homogeneous reference solutions. This outcome represents the phenomena of \textit{mass accretions} and \textit{clusterings on large scales}. In the following sections, we will propose suitable assumptions and present the precise findings in the main theorem.

\subsection{Assumptions of the model}\label{s:model} 
Due to the difficulty of directly answering the previous Main Question, several assumptions and simplifications have been made. The fundamental physical implications of this expanding Newtonian universe are located in \cite[\S$2$]{Liu2023a}. Furthermore, to streamline the computations, we impose the following assumptions:
\begin{enumerate}[leftmargin=*,label={(A\arabic*)}]
\item\label{A:dpen}  \underline{Dampings   and entropy productions:}
In this article, we consider the Euler--Poisson system \eqref{e:EP1}--\eqref{e:data0b} with the following specific terms $\mathcal{D}^i$ and $\mathcal{S}$:
\begin{align}
	\mathcal{D}^i(t,x^j,\rho,v^k,s,\phi):= & -\frac{\kappa f_0}{1+f} \Bigl[v^i- \Bigl(\frac{2}{3t}-\frac{f_0}{3(1+f)}\Bigr)x^i\Bigr] ,   &&\kappa>\frac{7}{6}, \label{e:D1} \\
	\mathcal{S}(t,x^j,\rho,v^k,s,\phi):= &	-\Bigl(\frac{2}{3}+\omega\Bigr)\del{i}   v^i  +\frac{2   v^i  x_i}{|\bx|^2}+3\omega \Bigl(\frac{2}{3t}-\frac{f_0}{3(1+f)}\Bigr),  &&\omega=-\frac{8}{5}. \label{e:S1}
\end{align}
where we denote $f_0(t):=f^\prime(t)$ throughout this article. The further physical meanings of $ \mathcal{D}^i$ and $ \mathcal{S} $ are discussed in Appendix \ref{s:DS}. 
As shown in Appendix \ref{s:DS}, $    \mathcal{D}^i$ serves as a \textit{damping term} that arises directly from the inhomogeneous densities, while $\mathcal{S}$ accentuates the \textit{growth of entropy} caused by the aforementioned inhomogeneities through a flow of matter that depends on temperature.
\begin{remark} 
	Both $\mathcal{D}^i$ and $\mathcal{S}$ are simplifications that aid in the proofs. Specifically, $\mathcal{D}^i$ implies that the solutions blow up at a fixed time for every spatial point, resulting in stable blow-ups. On the other hand, $\mathcal{S}$ simplifies system \eqref{e:EP1}--\eqref{e:EP3} by allowing us to reduce equation \eqref{e:EP2b} to a specific relationship, \eqref{e:s2}, which is demonstrated in Lemma \ref{t:sepr}.
\end{remark}

\item\label{A4} \underline{Periodicity and spherical symmetry of data:}    
To simplify the analysis, we assume that the initial data $\mathcal{d}(|\bx|)$ and $\mathcal{v}(|\bx|)$ are both \textit{$1$-log-periodic} functions, as defined below.
\end{enumerate}
\begin{definition}\label{e:lgprd}		
	A function $F(|\bx|)$ is called \textbf{$t$-log-periodic} if there is a $t$-parameterized exp-log transform $\mathcal{y}_t$ satisfying
		\begin{equation}\label{e:xztsf}
			|\bx|=\mathcal{y}_t(\zeta):=t^{\frac{2}{3}}  \bigl(1+f(t)\bigr)^{-\frac{1}{3}}\exp\zeta  \; \Leftrightarrow  \;   \zeta=\mathcal{y}_t^{-1}(|\bx|)=\ln   \bigl(t^{-\frac{2}{3}}  \bigl(1+f(t)\bigr)^{\frac{1}{3}} |\bx| \bigr),
		\end{equation}
		 such that  $
		\acute{F}(\zeta):=F\circ\mathcal{y}_t(\zeta)  
		$
		is a periodic function with the unit period, that is, $	\acute{F}(\zeta+m)=	\acute{F}(\zeta)$ for any $m\in \Zbb$ and $\zeta \in \Rbb$.
\end{definition}
\begin{remark}
	In other words, Definition \ref{e:lgprd} and \ref{A4} imply $\mathcal{d}\circ\mathcal{y}_1$ and  $\mathcal{v}\circ\mathcal{y}_1$ (where $\mathcal{y}_1(\zeta) =  \bigl(1+\beta\bigr)^{-\frac{1}{3}}\exp\zeta$) are both defined on $\Tbb$.	
\end{remark}

Under these assumptions, there is a direct conclusion:
\begin{claim*}
	 The Newtonian universe solution \eqref{e:exsol1}--\eqref{e:exsol2} (see \ref{c:Ntsl}) and the homogeneously perturbed solutions \eqref{e:sl1}--\eqref{e:sl2} (see \ref{c:hmsl}) both solve the Euler--Poisson system \eqref{e:EP1}--\eqref{e:data0b} with the specific terms   $\mathcal{D}^i$ and $\mathcal{S}$ defined by \eqref{e:D1} and \eqref{e:S1}.
\end{claim*}
\begin{proof}
Because these two solutions \eqref{e:exsol1}--\eqref{e:exsol2} and \eqref{e:sl1}--\eqref{e:sl2} both have a homogeneous density which yields the relative velocity defined by
\begin{equation}\label{e:relv1a}
	\check{v}^i:=v^i-\Bigl(\frac{2}{3t}-\frac{f_0}{3(1+f)}\Bigr) x^i
\end{equation}
vanishes, i.e., $\check{v}^i=0$.
 Further noting
\begin{equation*}
	\mathcal{S}(t,x^j,\rho,v^k,s,\phi) = 	-\Bigl(\frac{2}{3}+\omega\Bigr)\del{i}   \check{v}^i  +\frac{2  \check{v}^i  x_i}{|\bx|^2} \AND \mathcal{D}^i(t,x^j,\rho,v^k,s,\phi)=   -\frac{\kappa f_0}{1+f} 	\check{v}^i ,
\end{equation*}
 $\mathcal{D}^i$ and $\mathcal{S}$ vanish.  	
\end{proof}

\subsection{Main theorems}\label{s:mthm}
Let us present the main theorem of this article which answers the Main Question to some extend. In fact, from \cite{Liu2023a} we have already known:
\begin{claim*}
A homogeneous initial perturbation of \eqref{e:data2} around data \eqref{e:data1} results in a blowup solution $(\rho_r,v_r^i,\phi_r,s_r)$ given by \eqref{e:sl1}--\eqref{e:sl2}, meaning that the Newtonian universe $(\rrho,\rv^i,\rphi,\mathring{s})$ (as given in \eqref{e:exsol1}--\eqref{e:exsol2}) is gravitationally unstable.
\end{claim*} 
In order to address the issue of \textit{inhomogeneous density perturbations}, this article demonstrates that the blowup solutions $(\rho_r,v_r^i,\phi_r,s_r)$ are \textit{stable} in the presence of \textit{inhomogeneous density perturbations}. This result can be seen as a nonlinear version of the Jeans instability, with the stable blowup serving as a nonlinear counterpart to the Jeans criterion.

\begin{theorem}\label{t:mainthm1}
	Under assumptions \ref{A:dpen}--\ref{A4}, suppose $s\in \Zbb_{> \frac{7}{2}}$,  $\iota^3\in (0, 1/5]$ and $f\in C^2([1,t_m))$ given by Theorem \ref{t:mainthm0} solves equation \eqref{e:feq1b}--\eqref{e:feq2b} where $\mf>0$ and $\mg>0$, and assume $t_m>1$ such that $[1,t_m)$ is the maximal interval of existence of $f$ given by Theorem \ref{t:mainthm0}. Then there are small constants $\sigma_\star,\sigma>0$, such that if the initial data \eqref{e:data0} satisfies
	\begin{equation}\label{e:mtdata}
		\bigl\| \mathcal{d}\circ\mathcal{y}_1 -1 \bigr\|_{H^{s+1}(\Tbb)}+  	\bigl\|  \mathcal{v}\circ\mathcal{y}_1+1  \bigr\|_{H^{s+1}(\Tbb)} \leq  \sigma_\star\sigma ,
	\end{equation}
	then
	\begin{enumerate}[leftmargin=*]
		\item 	there is a  solution $(\rho, v^i,s,\phi)\in C^2([1,t_m)\times \Rbb^3)$ to the system \eqref{e:EP1}--\eqref{e:data0b} and  $\rho(t,|\bx|)$,  $v^i(t,|\bx|) x_i/|\bx|^2$ are $t$-log-periodic and spherical symmetric;
		\item there is a constant $C_1\in(0,1/\sigma)$,  such that $\varrho$ (defined by \eqref{e:perv}) and $v^i$ satisfy the estimates
		\begin{gather}
			0<\bigl(1-C_1 \sigma\bigr)f(t)< \varrho(t,x^i)  <  \bigl(1+C_1\sigma\bigr) f(t)  ,  \label{e:mainest1}\\
			0<(1-C_1\sigma) f_0(t) \leq 	\del{t}\varrho (t,x^i) \leq (1+C_1\sigma) f_0(t) , \label{e:mainest1b} \\
			- C\sigma (1+f(t))\leq 	x^i	\partial_{i}\varrho(t, x^i) \leq  C\sigma (1+f(t)) \label{e:mainest1c}
			\intertext{and}
			\Bigl(\frac{2}{3t}-\frac{ (1+C_1\sigma) f_0(t) }{3(1+f(t))}  \Bigr)x^i
			< v^i(t,x^i) <\Bigl(\frac{2}{3t}-\frac{(1-C_1\sigma) f_0(t)  }{3(1+f(t))}  \Bigr)x^i   \label{e:mainest2}
		\end{gather}
		for $(t,x^i)\in[1,t_m) \times \Rbb^3$;
		\item the entropy $s$ can be expressed by \begin{equation}\label{e:s2}
			s=\ln \Bigl(	t^{-\frac{4}{3}} \frac{(1+\varrho)^{-\frac{14}{15}}}{(1+f)^{-\frac{8}{5}}} \delta_{kl} x^k x^l \Bigr) ;
		\end{equation}
		\item $\varrho$ and $\del{t}\varrho$ both blowup at $t=t_m$, i.e.,
		\begin{equation}\label{e:blup0}
			\lim_{t\rightarrow t_m} \varrho(t,x^i)=+\infty \AND \lim_{t\rightarrow t_m} \del{t}\varrho(t,x^i)=+\infty ;
		\end{equation}
	\item if the parameter $\gamma$ of the data from \eqref{e:data0} satisfies $\gamma>1/3$,
	then there is a finite time $t_m<\infty$, such that the density contrast $\varrho$ and its derivative $\del{t}\varrho$ blow up at a finite time $t_m$.
	\end{enumerate}
\end{theorem}
The main theorem will be proved in the final section \S\ref{s:stp7}, following thorough preparation.

\subsection{Overviews and outlines}\label{s:ol} 
The \textit{primary tool} of this article is a Cauchy problem for a Fuchsian system, originally introduced by \cite{Oliynyk2016a}. This method is useful for determining the long-term behaviors of hyperbolic systems. The Fuchsian system is singular at $t=0$ and is represented by the equation \begin{equation}\label{e:fucprt} B^{\mu}(t,x,u)\partial_{\mu}u = \frac{1}{t}\textbf{B}(t,x,u)\textbf{P}u+H(t,x,u), \end{equation} 
as described in more detail in Appendix \ref{s:fuc}. Conditions \ref{c:2}--\ref{c:7} in Appendix \ref{s:fuc} outline the necessary requirements for the coefficients and remainders to ensure proper nonlinearities. Using this system, Theorem \ref{t:fuc} can be employed to determine the global behaviors of the Fuchsian fields $u$.

The Fuchsian formulations utilize the energy method but with normalization to a specific class of singular hyperbolic systems. The main advantage of this approach is that it streamlines the proof process, as it eliminates the need for exhaustive energy estimates. By transforming the Euler--Poisson system into Fuchsian formulations, we can use global existence theorems for Fuchsian systems, which have been developed by several authors, including Oliynyk, Beyer, J. Arturo, and the author of this work (see, for example, \cite{Oliynyk2016a, Liu2018b, Liu2018, Oliynyk2021, Beyer2021, Beyer2020}). The global existence theorem of the Fuchsian system already incorporates the necessary energy estimates.

While the primary idea is to transform to the Fuchsian system, there are several \textit{difficulties}:
\begin{enumerate}[leftmargin=*,label={(Q\arabic*)}]
	\item\label{c:htt} How to transform?
	\item\label{c:ffd}  
	How to select the \textit{appropriate Fuchsian fields}, as the proper fields typically encode the correct behaviors of solutions? This requires making predictions about the behaviors of solutions, which can be difficult. This challenge is similar to the issue in the energy method of selecting the correct energy to use.
	\item\label{c:cot} How to \textit{compactify the time} because the Fuchsian formulations require the time variables lie in $[-1,0)$ and the system is singular at $t=0$.
\end{enumerate}
The answers to these questions will be expounded throughout this article. In fact, . To address question \ref{c:cot}, we introduce a useful function, $g(t)$,  and  a \textit{compactified time transformation}:
\begin{align}\label{e:gdef}
	g(t):=\exp\Bigl(-A\int^t_{1
	} \frac{f(s)(f(s)+1)}{s^2 f_0(s)} ds \Bigr)>0  \AND \tau := -g(t)\in[-1,0) ,
\end{align}
respectively, for $t\in[1,\infty)$ where $A\in(0,2)$ is a constant and $f$ is a solution to \eqref{e:feq1b}--\eqref{e:feq2b}. This compactified time $\tau$ with the function $g(t)$ is well compatible with the Fuchsian transform.

Next, we will provide an overview of the structure of this article, detailing the steps involved in answering question \ref{c:htt}.

\S\ref{s:2} serves to transform the Euler--Poisson system into a system consisting of a second-order quasilinear hyperbolic equation for the density contrast $\varrho$ (similar to the one in \cite[eq. $(1.1)$]{Liu2022b}), a transport equation for the rescaled speed $\nu$ \eqref{e:chv}, and an algebraic equation that expresses $s$. The nonlinearities in this system are more complicated than those in the aforementioned equation.

In \S\ref{s:2.1}, we use a near-comoving coordinate system $(t, X^k)$ to transform the Euler--Poisson system into a second-order quasilinear hyperbolic equation for the density contrast, analogous to the one in \cite[eq. $(1.1)$]{Liu2022b}. In the following, we will use formal notations to briefly indicate how to proceed with the transformation, using notation such as $\eqref{e:EP2}_{(t,X^k)}$ to represent equation \eqref{e:EP2} in terms of the coordinate system $(t,X^k)$. 
	\begin{equation*}
		 \left. 	
		 \begin{aligned}
			\tr	\partial_{X^k} \eqref{e:EP2}_{(t,X^i)} \Rightarrow & \text{the eq. of }  \del{t}\Theta , (\text{Lemma \ref{t:dtth}})\\
			  \text{the continuity eq. \eqref{e:EP1}}  \Rightarrow & \text{the expression  \eqref{e:th2} of } \Theta
		\end{aligned}
	\right\}
	\Rightarrow
\begin{aligned}
&\text{$2$nd order hyperbolic eq. analogous}  \\
&\text{to the one in \cite[eq. $(1.1)$]{Liu2022b}}\\
&\text{(see Lemma \ref{t:Dttrho})}.
\end{aligned}
	\end{equation*}
On the other hand,  we can compare \eqref{e:EP1} and \eqref{e:EP2b} to re-express $s$ in terms of $\varrho$ via an algebraic expression. This allows us to simplify the system by removing the equation for $s$ and the variable $s$ itself (see Lemma \ref{t:sepr}).

In \S\ref{s:SSd}, we employ \textit{spherical symmetries} to rewrite the second order hyperbolic equation of $\varrho$ and the transport equation of the rescaled speed $\nu$ (defined in \eqref{e:chv}) in terms of spherical coordinates. To obtain the Fuchsian system, we introduce a \textit{spatial log-coordinate} in \S\ref{s:logfml} and use the new coordinate $\zeta :=\ln R$ to represent the $(\varrho,\nu)$ system (see Lemma \ref{t:lgfml}). In addition, in \S \ref{s:lggr}, we need to verify that the rescaled gravity defined by \eqref{e:Psi0}, i.e., \begin{equation*}     \Psi(t,\zeta):=\frac{1}{f(t) e^{3\zeta}}   \int^{\zeta}_{-\infty} [\hrho(t,z) -f(t)] e^{3z} d z 
\end{equation*} 
is periodic in its variable $\zeta$ (see Lemma \ref{t:Psprd}). Finally, we use Proposition \ref{t:prdsl} to show that the solution $(\varrho,\nu)$ is periodic in $\zeta$.

We follow a similar methodology as \cite{Liu2022b} in our analysis of the \textit{first step} (\S\ref{s:prefuc}) in \S\ref{s:fuchian}. However, we have obtained a second-order nonlinear hyperbolic equation with more complicated $\mathcal{g}^{\zeta\zeta}$ and remainders $F$. Similar to \cite{Liu2022b}, we select Fuchsian fields of density contrast $\varrho$ and its derivatives $\del{\mu}\varrho$. By compactifying time and using the Fuchsian fields, we formulate a Fuchsian system for the density contrast fields (see Lemma \ref{t:pref2}). However, to fully complete the system, we need a singular formulation of the equation for the rescaled speed $\nu$. With the help of the compactified time $\tau$, we are able to rewrite the equation for the Fuchsian field $\nu$ in a singular form (see Lemma \ref{t:pref3}). Furthermore, using a variant of the continuity equation \eqref{e:keyid3a} (see Lemma \ref{t:vph2}), we simplify the equation by removing the spatial derivative term.

After deriving the singular equation of $\nu$, it initially appeared that the suspected Fuchsian system was complete. However, upon closer examination, there were singular terms such as $\Psi/\tau$ present in the remainders that proved difficult to remove. As a solution, we decided to treat $\Psi$ as a new Fuchsian field and establish a singular equation for it. In \S \ref{s:evo}, we were able to derive this singular equation \eqref{e:dtpsi6} for $\Psi$ by differentiating $e^{3\zeta}\Psi$ with respect to $t$ and using a form of the continuity equation \eqref{e:keyid3b} (see Lemma \ref{t:pseq}). However, this equation alone did not guarantee that the complete singular system constituted a Fuchsian system due to the failure of condition \ref{c:5} (see Appendix \ref{s:fuc}). To rectify this, we differentiated $\Psi$ with respect to $\zeta$ to obtain equation \eqref{e:dipsi} (i.e., $\del{\zeta}  \Psi =u-3 \Psi $). By incorporating this equation into the previous singular equation of $\Psi$ in a suitable manner, we were able to satisfy condition \ref{c:5}. The resulting equation is presented in Lemma \ref{t:vph2}.

Upon gathering all the singular equations, we find that the system is a Fuchsian system, which is verified in \S\ref{s:stp4}. However, these verifications are not trivial. First, estimates on the key quantities $\chi$ and $\xi$ (which we introduced and estimated in our previous article \cite[\S$2.5$]{Liu2022b}) are crucial, as they help distinguish the singular terms in $\tau$ from the regular terms in the right-hand side of the Fuchsian system \eqref{e:fucprt}. Nevertheless, these estimates are insufficient to distinguish the order of singular terms in this article, and further estimates are required. 
In our previous article \cite{Liu2022b}, we established that $\chi(t)=4B+\mathfrak{G}(t)$ and $\mathfrak{G}(t)$ satisfies $\lim_{t\rightarrow t_m}\mathfrak{G}(t)=0$. In this article, we need to show that $\mathfrak{G}(t) \lesssim (-\tau)^{\frac{1}{2}}$, where $\tau$ is the compactified time defined in \eqref{e:gdef}. The \textit{idea} to achieve this estimate is to first calculate $\del{t} \chi=\del{t}\mathfrak{G}$ using the definition of $\chi$ to obtain an evolution equation of $\chi$ or $\mathfrak{G}$ (see Lemma \ref{t:dtchi}). We then use the compactified time $\tau$ to re-express this equation as a singular equation of $\mathfrak{G}$ (see Lemma \ref{t:Geq2}). Using a Gronwall type inequality, we can estimate $\mathfrak{G}(t) \lesssim  (-\tau)^{\frac{1}{2}}$ (see Lemma \ref{t:Gest2}). 
In addition to establishing the previously proven result that $\lim_{t\rightarrow t_m} \xi=0 $ in \cite{Liu2022b}, we must also show that 	
\begin{equation*}
	\lim_{t\rightarrow t_m} \biggl(\frac{1}{g f^{\frac{1}{2}}}\biggr)= 0  . 
\end{equation*}
This conclusion is established in Proposition \ref{t:fginv2} and  Corollary \ref{s:gf1/2} if $A\in(0,2)$. With the help of these estimates and by restricting $\iota^3\in(0,1/5]$, we are able to verify Condition \ref{c:5} (see \S\ref{s:F4}).

After verifying the Fuchsian system in \S\ref{s:stp4}, we can apply the global existence and estimate theorem (Theorem \ref{t:fuc}) to the Fuchsian fields in \S\ref{s:stp5}. In order to prove the Main Theorem \ref{t:mainthm1} in \S\ref{s:stp7}, we must transform the smallness condition \eqref{e:mtdata} of the data into a smallness condition for the Fuchsian fields. This requires an estimate $\|\Psi(t)\|_{H^s(\Tbb)}  \leq C   \|u(t)\|_{H^{s-1}(\Tbb)}$, which can be proven using $\del{\zeta}  \Psi =u-3 \Psi$ (Lemma \ref{t:psest}). Finally, by transforming the estimates of the Fuchsian fields back to $\varrho$ and its derivatives and $v^i$, we can conclude the Main Theorem \ref{t:mainthm1}.

\subsection{Related works} 
The blowup phenomenon has been extensively researched in relation to different hyperbolic equations. Notably, \cite{Alinhac1995,Kichenassamy2021} provide a comprehensive overview of blowups in various hyperbolic systems. In addition, Speck \cite{Speck2020,Speck2017} has investigated stable ODE-type blowup outcomes for quasilinear wave equations with a Riccati-type derivative-quadratic semilinear term and a form of stable Tricomi-type degeneracy formation for a specific wave equation. Furthermore, the author of this article \cite{Liu2022b} has studied self-increasing blowup solutions of a group of second-order hyperbolic equations.

The fundamental technique employed in this article is the Cauchy problem for Fuchsian systems, initially proposed by Oliynyk \cite{Oliynyk2016a}. Subsequently, it has been developed and utilized in various works, such as \cite{Beyer2021,Beyer2020,Liu2018b,Liu2018,Liu2022a,Liu2022,Liu2018a,Liu2022b} for a range of systems, including the Einstein--scalar field, Einstein--Euler, and Einstein--Yang--Mills systems, etc.

The gravitational collapse and its associated mass accretion have been investigated in previous works, such as \cite{Guo2018,Guo2021}. In contrast to \cite{Guo2018}, where the focus is on blowups for compact isolated bodies, our work concentrates on blowups on $\Rbb^3$, enabling us to model the clustering of matter on a larger scale. Furthermore, our proof methodology is distinct, and the connections between the two approaches warrant further exploration.

\subsection{Notations}\label{s:AIN}
Unless stated otherwise, we will apply the following conventions of notations throughout this article without recalling their meanings in the following sections.
\subsubsection{Indices and vectors}\label{iandc}
Unless stated otherwise, our indexing convention will be as follows: we use lowercase Latin letters, e.g. $i, j,k$, for spatial indices that run from $1$ to $3$, and lowercase Greek letters, e.g. $\alpha, \beta, \gamma$, for spacetime indices
that run from $0$ to $3$.  We will follow the \textit{Einstein summation convention}, that is, repeated lower and upper indices are implicitly summed over. We use $x^{i}$ ($i=1, 2, 3$) to denote the standard coordinates on the $\mathbb{R}^3$ and $t = x^{0}$ a time coordinate on the interval $[1,\infty)$.  We use the \textit{bold fonds} to represent a vector field, for example, $\mathbf{v}$ to represent a velocity field and $\bx$ the displacement. In this article, we will use the Euclidean metric $\delta^{ij}$ and $\delta_{ij}$ to raise and lower the indices, for instance, $x_i:=\delta_{ij}x^j$,  $v_i:=\delta_{ij}v^j$ and $\partial^i:=\delta^{ij}\del{j}$.

\subsubsection{Lagrangian descriptions}\label{s:lag}
In this article, we denote $(t,x^i)$ the Eulerian coordinate and $(t,X^j)$ the comoving coordinate with the reference solutions \eqref{e:sl1}--\eqref{e:sl2} (i.e, regard the reference solutions as \textit{observers}) defined by
\begin{equation}\label{e:xvartrf}
X^k := a^{-1}(t)(1+f(t))^{\frac{1}{3}} x^k  \quad	\Leftrightarrow   \quad  x^k=x^k(t,X^j)= a(t)(1+f(t))^{-\frac{1}{3}} X^k
\end{equation}
where $a(t)=t^{\frac{2}{3}}$ is the scale factor of the Newtonian universe. In this case, we have
\begin{equation}\label{e:vrepr}
	v^k_r=	\frac{\partial x^k}{\partial t} =   \frac{2}{3t}  x^k
	-\frac{f_0}{3(1+f) }  x^k =   \frac{2}{3t}   a(1+f)^{-\frac{1}{3}} X^k
	-\frac{a f_0}{3(1+f)^{\frac{4}{3}} }  X^k .
\end{equation}
where we denote $f_0:=f^\prime$.
Under these two coordinates, a vector $\mathbf{v}$ can be represented, with respect to $(t,x^i)$ and $(t,X^j)$ respectively, by $(v^i(t,x^j))$ and the \textit{wavy underlined} notations $(\V^i(t,X^j))$ (i.e., $\mathbf{v}=v^i\frac{\partial}{\partial x^i}=\V^i \frac{\partial}{\partial X^i}$), and a scalar function $F$ be represented by $F(t,x^j)$ and $	\dwave{F}(t,X^j)$. We have the relations
\begin{equation}\label{e:undf}
	v^i(t,x^k(t,X^j))=\frac{\partial x^k}{\partial X^i} \V^i(t,X^k)=a  (1+f)^{-\frac{1}{3}}  \V^i(t,X^k) \AND 	\dwave{F}(t,X^j)=F(t,x^i(t,X^j)).
\end{equation}
For simplicity of the notations, we introduce $D_i:=\frac{\partial }{\partial X^i} $, then
\begin{equation}\label{e:direl2}
	\del{i}F(t,x^k(t,X^j)) =a^{-1} (1+f)^{\frac{1}{3}} D_i \dwave{F}(t,  X^j)  .
\end{equation}

This article involves two different time coordinates: the open time $t\in[t_0,\infty)$ and the compactified time $\tau=-g(t)\in [-1,0)$ given by \eqref{e:ttf}. For any scalar function $F(t,X^i)$, we always use the \textit{underlined} notation,
\begin{equation} \label{e:udl}
	\underline{F}(\tau,X^i):=  F(g^{-1}(-\tau), X^i)
\end{equation}
to denote the representation of $F$ in compactified time coordinate $\tau$ throughout this article.

\subsubsection{Function spaces, inner-products and matrix inequalities}\label{s:funsp}
Partial derivatives with respect to the  coordinates $(x^\mu)=(t,x^i)$ will be denoted by $\partial_\mu = \partial/\partial x^\mu$.
We also use Greek letters to denote multi-indices, e.g.
$\alpha = (\alpha_1,\alpha_2,\ldots,\alpha_n)\in \mathbb{Z}_{\geq 0}^n$, and employ the standard notation $D^\alpha = \partial_{1}^{\alpha_1} \partial_{2}^{\alpha_2}\cdots
\partial_{n}^{\alpha_n}$ for spatial partial derivatives.

Given a finite-dimensional vector space $V$, we let
$H^s(\mathbb{T}^n,V)$, $s\in \mathbb{Z}_{\geq 0}$,
denote the space of maps from $\mathbb{T}^n$ to $V$ with $s$ derivatives in $L^2(\Tbb^n)$. When the
vector space $V$ is clear from context, for example, $V=\Rbb^N$, we write $H^s(\mathbb{T}^n)$ instead of $H^s(\mathbb{T}^n,V)$.
Letting
\begin{equation*}
	\langle{u,v\rangle} = \int_{\mathbb{T}^n} (u(x),v(x))\, d^n x,
\end{equation*}
where $(\cdot,\cdot)$
is the Euclidean inner product on $\Rbb^N$ (i.e., $(\xi,\zeta)=\xi^T\zeta$ for any $\xi, \zeta\in \Rbb^N$), denote the standard $L^2$ inner product, the $H^s$ norm is defined by
\begin{equation*}
	\|u\|_{H^s}^2 = \sum_{0\leq |\alpha|\leq s} \langle D^\alpha u, D^\alpha u \rangle.
\end{equation*}

For matrices $A,B\in \mathbb{M}_{N\times N}$, we define
\begin{equation*}
	A\leq B \quad \Leftrightarrow \quad  (\zeta,A\zeta)\leq (\zeta,B\zeta),  \quad \forall \zeta\in \Rbb^N.
\end{equation*}
We also employ the notation
\begin{equation*}
	a\lesssim b \quad \Leftrightarrow \quad a\leq C b
\end{equation*}
where $C$ denotes a generic constant.


\section{Reference solutions-based formulations}\label{s:2}
\subsection{A second order equation of the density contrast}\label{s:2.1}
The \textit{objectives} of this section are twofold: Firstly, to rewrite the Euler--Poisson system in a \textit{Lagrangian coordinate} $(t,X^k)$ (as discussed in \S\ref{s:lag}). Secondly, to obtain a second order hyperbolic equation for the density contrast $\varrho$, similar to the one studied in \cite{Liu2022b}. However, the current case involves more complicated nonlinear terms $\mathcal{g}^{ij}$ and $F$, with more unknown variables. To estimate the perturbations of the reference solutions, we further decompose the variables by introducing the following definitions.
\begin{equation}\label{e:ckvph}
	\check{\V}^k (t,X^i) : =  \tilde{\V}^k(t,X^i)+\frac{  f_0}{3(1+f) }  X^k \AND  \dwave{\check{\phi}} (t,X^k) =
		\dwave{\tphi} (t,X^k)-
		a^2(1+f)^{-\frac{2}{3}} \frac{ \iota^3 f X^i X^j \delta_{ij} }{9t^{2}}
\end{equation}

Before proceeding, let us first point out the entropy $s$ can be expressed in terms of the density contrast $\varrho$.
\begin{lemma}\label{t:sepr}
Suppose $(\rho,v^i)$ solves \eqref{e:EP1} with data \eqref{e:data0}, $\varrho$ is defined by \eqref{e:perv},  $(\rrho,\rv^i,\mathring{s},\rphi)$ and $(\rho_r,v^i_r,s_r,\phi_r)$ are given by \eqref{e:exsol1}--\eqref{e:exsol2} and \eqref{e:sl1}--\eqref{e:sl2}, respectively. If $s$ solves \eqref{e:EP2b} ($\mathcal{S}$ is given by \eqref{e:S1}) with the data \eqref{e:data0b}, i.e., $s$ solves
\begin{equation}\label{e:seq2}
	\del{t} s+v^i\del{i}s=  -\Bigl(\frac{2}{3}+\omega\Bigr)\del{i}   \check{v}^i  +\frac{2  \check{v}^i  x_i}{|\bx|^2}
\end{equation}
with the data \eqref{e:data0b},  then
\begin{equation}\label{e:s2}
	s=\ln \Bigl(	t^{-\frac{4}{3}} \frac{(1+\varrho)^{\frac{2}{3}+\omega}}{(1+f)^{\omega}} \delta_{kl} x^k x^l \Bigr) .
\end{equation}

Moreover, according to the equation of state \eqref{e:eos}, the pressure becomes
\begin{equation}\label{e:eos2}
	p =K e^s \rho^{\frac{4}{3} } =K  t^{-\frac{4}{3}} \frac{(1+\varrho)^{\frac{2}{3}+\omega}}{(1+f)^{\omega}} \delta_{kl} x^k x^l  \rho^{\frac{4}{3} }  .
\end{equation}
\end{lemma}
\begin{proof}
	Firstly, let us recall  \eqref{e:exsol1}, \eqref{e:perv},  \eqref{e:sl1} and \eqref{e:relv1a}, and they imply
	\begin{equation}\label{e:v0}
		\tilde{v}^i=v^i-\rv^i=v_r^i-\rv^i+\check{v}^i = \check{v}^i -\frac{f_0}{3(1+f)} x^i \AND   \check{v}^i:=v^i-v_r^i  .
	\end{equation}
	Recalling the definition \eqref{e:perv} of density contrast $\varrho$, we first note the equation \eqref{e:EP1} and the solutions \eqref{e:exsol1}--\eqref{e:exsol2} imply \begin{equation}\label{e:eq5}
		\del{t}\varrho  +v^i \del{i}\varrho   =  -(1+\varrho ) \del{i} \tv^i
		\quad  \overset{\eqref{e:v0}}{\Leftrightarrow}  \quad
	 \del{t}\ln(1+\varrho ) + v^i \del{i} \ln(1+\varrho )  - \frac{f_0}{ 1+f } =   -  \del{i}\check{v}^i .
	\end{equation}
Noting
\begin{equation*}
 \del{t} s_r+v^i_r\del{i}s_r=  0, \quad 	\check{v}^i\del{i}s_r
	\overset{\eqref{e:sl2}}{=}  \frac{2\check{v}^i x_i}{|\bx|^2}  \AND \frac{f_0}{1+f}= \del{t} \ln (1+f) ,
\end{equation*}
and substituting \eqref{e:eq5} into \eqref{e:seq2} yield
\begin{equation*}
		\del{t}\check{s} + v^i\del{i}\check{s}+\check{v}^i\del{i}s_r =  \del{t}\ln(1+\varrho )^{\frac{2}{3}+\omega} + v^i \del{i} \ln(1+\varrho )^{\frac{2}{3}+\omega}  -  \del{t} \ln(1+f)^{\frac{2}{3}+\omega} +  \check{v}^i\del{i}s_r ,
\end{equation*}
which implies
\begin{equation}\label{e:dtseq1}
	\del{t}\biggl[\check{s} - \ln\Bigl(\frac{1+\varrho }{1+f}\Bigr)^{\frac{2}{3}+\omega} \biggr] + v^i\del{i} \biggl[\check{s} -  \ln\Bigl(\frac{1+\varrho }{1+f}\Bigr)^{\frac{2}{3}+\omega} \biggr] =   0  .
\end{equation}
By data \eqref{e:data0}, \eqref{e:data0b} and the reference solution \eqref{e:sl2}, we obtain
\begin{equation}\label{e:dtsdt1}
	\biggl(\check{s}- \ln\Bigl(\frac{1+\varrho }{1+f}\Bigr)^{\frac{2}{3}+\omega}\biggr)\bigg|_{t=1} =0 .
\end{equation}
The equation \eqref{e:dtseq1} and \eqref{e:dtsdt1} mean  along every velocity field $v^i$ generated characteristics, 	$\check{s}- \ln\bigl(\frac{1+\varrho }{1+f}\bigr)^{\frac{2}{3}+\omega}$ vanishes, which means solving \eqref{e:dtseq1} and \eqref{e:dtsdt1}, there is a unique solution
\begin{equation*}
	\check{s}- \ln\Bigl(\frac{1+\varrho }{1+f}\Bigr)^{\frac{2}{3}+\omega} =0 \quad \Rightarrow \quad 	\check{s} =\Bigl(\frac{2}{3}+\omega\Bigr) \ln \Bigl(\frac{1+\varrho}{1+f}\Bigr) \quad \overset{\eqref{e:sl2}}{\Rightarrow} \quad \eqref{e:s2} .
\end{equation*}
We then complete the proof.
\end{proof}

\begin{remark}
	This lemma enable us to use the algebraic expression \eqref{e:s2} to replace the transport PDE \eqref{e:EP2b} and use \eqref{e:eos2} to replace \eqref{e:eos} in the Euler--Poisson system \eqref{e:EP1}--\eqref{e:data0b}.
\end{remark}

\begin{lemma}\label{t:EPrep}
	Suppose $\Rho$, $\check{\V}^k$ and $\dwave{\tphi} $ (recalling the notations in \S\ref{s:lag}) are defined by \eqref{e:perv} and \eqref{e:ckvph}, then the Euler--Poisson system \eqref{e:EP1}--\eqref{e:EP3} can be represented in terms of the perturbation variables $(\Rho, \check{\V}^i, \dwave{\check{\phi}} )$ in the coordinates $(t,X^j)$ in \S\ref{s:lag},
	\begin{align}
&	\del{t}\Rho  + \bigl(1+\Rho  \bigr)\Bigl(D_i\check{\V}^i-\frac{f_0}{1+f}\Bigr)+ \check{\V}^i  D_{i} \Rho  =  0, \label{e:EP.a.1}\\
 &\del{t}  \check{\V}^k  +	\Bigl(\frac{4}{3t}     -\frac{2 f_0}{3(1+f)} \Bigr)  \check{\V}^k
	+     \dwave{\cv}^j D_{j}   \check{\V}^k     + (2+\omega)  K t^{-\frac{4}{3}} \delta_{ij} X^i X^j \rrho^{\frac{1}{3}}  (1+f)^{-\omega}	 (1 +\Rho)^{\omega} \delta^{kl}  D_l  \Rho  \notag   \\
 &\hspace{2.5cm} + \frac{ 2 K\iota (1+f)}{(6\pi  )^{\frac{1}{3}}t^2} X^k    \Bigl[   \Bigl(\frac{1+\Rho}{1+f}\Bigr)^{1+\omega}  -1 \Bigr]   +a^{-2}(1+f)^{\frac{2}{3}} \delta^{kl}  D_l \dwave{\check{\phi}} =  -\frac{\kappa f_0}{1+f} \check{\V}^k ,
 \label{e:EP.a.2}  \\
&	\delta^{kl}D_k D_l  \dwave{\check{\phi}}    =    \frac{2 \iota^3}{3t^{\frac{2}{3}}} (1+f)^{-\frac{2}{3}}  ( \dwave{\varrho} -f) .  \label{e:EP.a.3}
	\end{align}
\end{lemma}
\begin{proof}
The proof is from the tedious and direct calculations by \eqref{e:exsol1}--\eqref{e:exsol2}, \eqref{e:perv}, \eqref{e:undf} and \eqref{e:direl2}. We only remark some key steps but leave the details to the readers.
Firstly, by \eqref{e:eq5}, direct calculations imply
\begin{equation}\label{e:eq1}
	\del{t}\varrho  +\frac{2}{3t} x^i \del{i}\varrho  +   \del{i} \tv^i + \varrho  \del{i} \tv^i  + \tv^i   \del{i} \varrho =  0.
\end{equation}
Note, by the chain rule, we have
\begin{equation}\label{e:eq0}
		\del{t}\Rho(t,X^k)
	=  \del{t}\varrho(t,x^i(t,X^k))+ \frac{\partial x^i (t,X^j)}{\partial t}\del{i} \varrho (t,x^i(t,X^k))   .
\end{equation}
By \eqref{e:xvartrf} and \eqref{e:vrepr}, this \eqref{e:eq0} implies
\begin{equation*}
	\del{t}\Rho(t,X^k) +\frac{  f_0}{3(1+f)  }  X^i   D_i \Rho (t,X^k)
	=  \del{t}\varrho(t,x^i(t,X^k))+ \frac{2}{3t} x^i
	\del{i} \varrho (t,x^i(t,X^k))   .
\end{equation*}

If we define the expansion of   perturbations
\begin{equation}\label{e:th0}
	\Theta(t,X^k)=D_i\check{\V}^i (t,X^k),
\end{equation}
then by \eqref{e:ckvph} and since a scalar is independent of the coordinates, we reach the relations,
\begin{equation}\label{e:eq2}
	\del{i}\tilde{v}^i=D_{i} \tilde{\V}^i   = \Theta-\frac{f_0}{1+f}.
\end{equation}
Then, with the help of \eqref{e:eq1}--\eqref{e:eq2}, we conclude \eqref{e:EP.a.1}.

Let us turn to \eqref{e:EP.a.2}, straightforward calculations, by using \eqref{e:exsol1}--\eqref{e:exsol2}, \eqref{e:perv} yield
\begin{equation*}
	\del{t}  \tv^i + Hx^j   \del{j}  \tv^i +  H \tv^i+   \tv^j \del{j}  \tv^i +\frac{\del{}^i  (\tilde{p} / \rrho)}{1 +\varrho} - 2 K t^{-\frac{4}{3}} \rrho^{\frac{1}{3}} x^i \frac{   \varrho}{1+\varrho}  + \del{}^i \tphi =  -\frac{\kappa f_0}{1+f} \tv^i .
\end{equation*}
Then we note, by \eqref{e:eos2} and \eqref{e:exsol1}, we obtain
\begin{equation*}
	\delta^{ij} 	\del{j} \Bigl( \frac{\tilde{p}}{\rrho} \Bigr)
	=  2 K t^{-\frac{4}{3}}   x^i \rrho^{\frac{1}{3}} \bigl((1+f)^{-\omega}(1+\varrho)^{2+\omega}-1\bigr) +  (2+\omega) K t^{-\frac{4}{3}} \delta_{kl} x^k x^l \rrho^{\frac{1}{3}}  (1+f)^{-\omega}	(1+\varrho)^{1+\omega}\delta^{ij}  \del{j} \varrho  .
\end{equation*}
In addition, by \eqref{e:exsol1} and the identity of $\iota$ \eqref{e:ioeq}, we have
\begin{align*}
 & \frac{2 K t^{-\frac{4}{3}}   X^k \rrho^{\frac{1}{3}} \bigl((1+f)^{-\omega}(1+\varrho)^{2+\omega}-1\bigr) }{1+\Rho} - \frac{2 K t^{-\frac{4}{3}} \rrho^{\frac{1}{3} } X^k \Rho }{1+\Rho} +\frac{2f(\iota^3-1) X^k }{9t^2}  \notag  \\
 = &  \frac{2 K \iota X^k}{(6\pi)^{\frac{1}{3}} t^2} \bigl[(1+f)^{-\omega}(1+\Rho)^{1+\omega} -1 \bigr]  - \frac{2 K \iota X^k f }{(6\pi )^{\frac{1}{3}}t^2}
 =    \frac{2 K \iota X^k (1 +f) }{(6\pi)^{\frac{1}{3}} t^2} \Bigl[ \Bigl(\frac{1+\Rho}{1+f}\Bigr)^{1+\omega} - 1 \Bigr]  .
\end{align*}
Using above two identities and changing coordinates with the help of \eqref{e:xvartrf}--\eqref{e:direl2} and \eqref{e:ckvph}, lengthy but straightforward computations, with the help of \eqref{e:feq1b} due to the presence of $f^{\prime\prime}$,  conclude \eqref{e:EP.a.2}.

In order to derive \eqref{e:EP.a.3}, we only note $\delta^{kl}D_k D_l \bigl(   X^i X^j \delta_{ij}\bigr)=6$, changing coordinates and using \eqref{e:ckvph}, we can arrive at it. these then complete the proof.
\end{proof}

In the next, we intend to drive a second-order equation analogous to the one in \cite[eq. $(1.1)$]{Liu2022b}. In order to do so, by noticing that \eqref{e:EP.a.1} implies
\begin{equation}\label{e:th2}
	\Theta= \frac{f_0}{1+f} -\frac{\del{t}\Rho}{1+\Rho} -\frac{\check{\dwave{v}}^i D_i\Rho}{1+\Rho} ,
\end{equation}
we first try to find an equation of the expansion $\Theta$. Then \eqref{e:th2} helps achieve the goal.

Let us define the \textit{rotations} $\Omega_{ij}$ and the \textit{shears} $\Xi_{ij}$ of the perturbed velocity fields $\dwave{\cv}^k$ in the following.
\begin{align}\label{e:Th1}
	\Theta_{ij}=\frac{1}{2}(D_i \check{\V}_j+ D_j \check{\V}_i), \quad \Omega_{ij}=\frac{1}{2}(D_i \check{\V}_j-D_j \check{\V}_i)  \AND \Xi_{ij} :=\Theta_{ij} -\frac{1}{3}\Theta\delta_{ij}
\end{align}
where we denote $\check{\V}_j:= \delta_{jk}\dwave{\cv}^k$ (recall \S\ref{iandc}), and we see $\Theta:=\delta^{ij} \Theta_{ij}=D_i\cuv^i$ consisting with the \eqref{e:th0}. We note $\delta^{ij}\Xi_{ij}=0$ and $\delta^{ij}\Omega_{ij}=0$, and we can verify that $\Xi_{ij}$ is symmetric and trace-free, while $\Omega_{ij}$ is anti-symmetric.
Then there is a decomposition
\begin{equation}\label{e:dcpdv}
	D_i \check{\V}_j=\Xi_{ij}+\frac{1}{3}\Theta \delta_{ij}+\Omega_{ij} .
\end{equation}

\begin{lemma}\label{t:dtth}
	If $(\dwave{\varrho},\dwave{\check{v}}^i,\dwave{\check{\phi}})$ solves the system \eqref{e:EP.a.1}--\eqref{e:EP.a.3}, then  $(\Theta,\Xi_{il}, \Omega_{ij},\dwave{\varrho},\dwave{\check{v}}^i)$ satisfy an equation
	\begin{align}\label{e:Dtth0a}
		& \del{t}  \Theta   +	\Bigl(\frac{4}{3t}     -\frac{2 f_0}{3(1+f)} \Bigr) \Theta
		+  \delta^{ij} \Xi_{il} \delta^{lk}\Xi_{kj}  +\frac{1}{3}\Theta^2     +\delta^{ij} \Omega_{kj} \Omega_{il}  \delta^{lk}  +     \dwave{\cv}^k D_{k} \Theta    \notag  \\
		&    \hspace{1cm} +   \frac{ (2+\omega )  K \iota}{(6\pi)^{\frac{1}{3}} t^{2}} \delta_{rs} X^r X^s \Bigl(\frac{1+\Rho}{1+f}\Bigr)^\omega    \delta^{ij}  D_j  D_i \Rho +\frac{6K\iota(1+f)}{(6\pi)^{\frac{1}{3}} t^2}   \Bigl[\Bigl(\frac{1+\Rho}{1+f}\Bigr)^{1+\omega}-1\Bigr] \notag   \\
		&   \hspace{1cm}   + \omega(2+\omega) K t^{-\frac{4}{3}} \delta_{sr} X^s X^r \rrho^{\frac{1}{3}}  (1+f)^{-\omega}	  (1 +\Rho)^{\omega-1} \delta^{ij}  D_i \Rho  D_j \Rho +  \frac{2 \iota^3}{3t^{2}}   ( \dwave{\varrho} -f) \notag  \\
		&    \hspace{1cm}  +    \frac{2(3+2\omega)  K  \iota}{(6\pi)^{\frac{1}{3}}t^2} \Bigl(\frac{1+\Rho}{1+f}\Bigr)^\omega    \delta^{ij} X_i   D_j \Rho    =  -\frac{\kappa f_0}{1+f} \Theta .
\end{align}
\end{lemma}
\begin{proof}
Differentiating \eqref{e:EP.a.2} with respect to $X^j$,
	\begin{align*}
		& \del{t}  D_j \check{\V}^k  +	\Bigl(\frac{4}{3t}     -\frac{2 f_0}{3(1+f)} \Bigr) D_j  \check{\V}^k
		   +     \dwave{\cv}^l D_j D_{l}   \check{\V}^k     +   \frac{ (2+\omega )  K \iota}{(6\pi)^{\frac{1}{3}} t^{2}} \delta_{is} X^i X^s  \Bigl(\frac{1+\Rho}{1+f}\Bigr)^\omega \delta^{kl} D_j  D_l  \Rho  \notag   \\
		&   \hspace{0.5cm}  +    \frac{2(2+\omega)  K  \iota}{(6\pi)^{\frac{1}{3}}t^2} \Bigl(\frac{1+\Rho}{1+f}\Bigr)^\omega \delta_{jr}  X^r \delta^{kl}  D_l  \Rho   + \omega(2+\omega) K t^{-\frac{4}{3}} \delta_{ir} X^i X^r \rrho^{\frac{1}{3}}  	 \frac{ (1 +\Rho)^{\omega-1}}{(1+f)^{\omega}} \delta^{kl}  D_l  \Rho  D_j \Rho  \notag  \\
		& \hspace{0.5cm} +   D_j  \dwave{\cv}^l D_{l}   \check{\V}^k  +\frac{2K\iota(1+f)}{(6\pi)^{\frac{1}{3}} t^2} \delta^k_j \Bigl[\Bigl(\frac{1+\Rho}{1+f}\Bigr)^{1+\omega}-1\Bigr]  + \frac{2(1+\omega)K\iota}{(6\pi)^{\frac{1}{3}}t^2} \Bigl(\frac{1+\Rho}{1+f}\Bigr)^\omega X^k  D_j \Rho  \notag  \\
		&  \hspace{0.5cm}  +a^{-2}(1+f)^{\frac{2}{3}} \delta^{kl}  D_j D_l \dwave{\check{\phi}} =  -\frac{\kappa f_0}{1+f} D_j \check{\V}^k .
\end{align*}
Then, denoting $ X_i=X^k \delta_{ik}$ and using $\delta_{ik}$ to lower the indices yield
	\begin{align}\label{e:dtdv1}
		& \del{t}  D_j \check{\V}_i  +	\Bigl(\frac{4}{3t}     -\frac{2 f_0}{3(1+f)} \Bigr) D_j  \check{\V}_i
		 +     \dwave{\cv}^l D_j D_{l}   \check{\V}_i      +   \frac{ (2+\omega )  K \iota}{(6\pi)^{\frac{1}{3}} t^{2}} \delta_{rs} X^r X^s  	\Bigl(\frac{1+\Rho}{1+f}\Bigr)^\omega  D_j  D_i \Rho  \notag   \\
		& \hspace{0.5cm}   +\frac{2K\iota(1+f)}{(6\pi)^{\frac{1}{3}} t^2} \delta_{ij} \Bigl[\Bigl(\frac{1+\Rho}{1+f}\Bigr)^{1+\omega}-1\Bigr]  + \omega(2+\omega) K t^{-\frac{4}{3}} \delta_{sr} X^s X^r \rrho^{\frac{1}{3}}  	  \frac{(1 +\Rho)^{\omega-1}}{(1+f)^{\omega}}   D_i \Rho  D_j \Rho  \notag  \\
		& \hspace{0.5cm} +    \frac{2(2+\omega)  K  \iota}{(6\pi)^{\frac{1}{3}}t^2} \Bigl(\frac{1+\Rho}{1+f}\Bigr)^\omega   X_j   D_i  \Rho  + \frac{2(1+\omega)K\iota}{(6\pi)^{\frac{1}{3}}t^2} \Bigl(\frac{1+\Rho}{1+f}\Bigr)^\omega   X_i  D_j \Rho    +a^{-2}(1+f)^{\frac{2}{3}}   D_j D_i \dwave{\check{\phi}} \notag  \\
		&\hspace{0.5cm}  +   D_j  \dwave{\cv}^l D_{l}   \check{\V}_i   =  -\frac{\kappa f_0}{1+f} D_j \check{\V}_i .
\end{align}
Switching the indices leads to
	\begin{align}\label{e:dtdv2}
		& \del{t}  D_i \check{\V}_j  +	\Bigl(\frac{4}{3t}     -\frac{2 f_0}{3(1+f)} \Bigr) D_i  \check{\V}_j
	 +     \dwave{\cv}^l D_i D_{l}   \check{\V}_j     +   \frac{ (2+\omega )  K \iota}{(6\pi)^{\frac{1}{3}} t^{2}} \delta_{rs} X^r X^s  \Bigl(\frac{1+\Rho}{1+f}\Bigr)^\omega    D_j  D_i \Rho  \notag   \\
		&  \hspace{0.5cm} +\frac{2K\iota(1+f)}{(6\pi)^{\frac{1}{3}} t^2} \delta_{ij} \Bigl[\Bigl(\frac{1+\Rho}{1+f}\Bigr)^{1+\omega}-1\Bigr]    + \omega(2+\omega) K t^{-\frac{4}{3}} \delta_{sr} X^s X^r \rrho^{\frac{1}{3}}  	\frac{  (1 +\Rho)^{\omega-1} }{(1+f)^{\omega}}  D_i \Rho  D_j \Rho  \notag  \\
		& \hspace{0.5cm} +    \frac{2(2+\omega)  K  \iota}{(6\pi)^{\frac{1}{3}}t^2}  \Bigl(\frac{1+\Rho}{1+f}\Bigr)^\omega   X_i   D_j \Rho   + \frac{2(1+\omega)K\iota}{(6\pi)^{\frac{1}{3}}t^2} \Bigl(\frac{1+\Rho}{1+f}\Bigr)^\omega   X_j  D_i \Rho    +a^{-2}(1+f)^{\frac{2}{3}}   D_j D_i \dwave{\check{\phi}} \notag  \\
		& \hspace{0.5cm} 	+   D_i  \dwave{\cv}^l D_{l}   \check{\V}_j   =  -\frac{\kappa f_0}{1+f} D_i \check{\V}_j .
\end{align}

Next, with the help of \eqref{e:dcpdv},  we firstly note
\begin{equation}\label{e:dvdv}
	 \delta^{kl} D_i  \dwave{\check{v}_l}    D_k  \dwave{\check{v}_j}
	 =  \delta^{kl} \Xi_{il} \Xi_{kj}+\frac{2}{3}\Theta  \Xi_{ij} +\Omega_{il}\Xi_{kj}\delta^{kl}   +\frac{1}{9}\Theta^2   \delta_{ji}+\frac{2}{3}\Theta  \Omega_{ij} + \Omega_{kj} \delta^{lk}\Xi_{il}  +\Omega_{kj} \Omega_{il}  \delta^{lk}.
\end{equation}
Using  $\Xi_{ij}=\Xi_{ji}$ and $\Omega_{ij}=-\Omega_{ji}$, then  $\Omega_{il}\Xi_{kj}\delta^{kl} +\Omega_{ki}\Xi_{jl}\delta^{kl}=\Omega_{il}\Xi_{kj}\delta^{kl} +\Omega_{li}\Xi_{jk}\delta^{kl} =0$ and $\Omega_{jl}\Xi_{ki}\delta^{kl} +\Omega_{kj}\Xi_{il}\delta^{kl}=\Omega_{jl}\Xi_{ki}\delta^{kl} +\Omega_{lj}\Xi_{ik}\delta^{kl} =0$. By using \eqref{e:dvdv} adding the one derived by switching the indices $i$ and $j$ of \eqref{e:dvdv}, we then arrive at
\begin{equation}\label{e:dvdv2}
	\frac{1}{2}(D_i \tv^k D_k \tv_j + D_j \tv^k  D_k \tv_i) =   \Xi_{il} \delta^{lk}\Xi_{kj}+\frac{2}{3}\Theta  \Xi_{ij}    +\frac{1}{9}\Theta^2   \delta_{ji}   +\Omega_{kj} \Omega_{il}  \delta^{lk}.
\end{equation}
Then by using \eqref{e:Th1} to calculate [\eqref{e:dtdv1}$+$\eqref{e:dtdv2}]$/2$, with the help of \eqref{e:dvdv2}, we obtain
\begin{align}\label{e:Dttij}
		& \del{t}  \Theta_{ij}  +	\Bigl(\frac{4}{3t}     -\frac{2 f_0}{3(1+f)} \Bigr) \Theta_{ij}
		+   \Xi_{il} \delta^{lk}\Xi_{kj}+\frac{2}{3}\Theta  \Xi_{ij}    +\frac{1}{9}\Theta^2   \delta_{ji}   +\Omega_{kj} \Omega_{il}  \delta^{lk}  +     \dwave{\cv}^k D_{k} \Theta_{ij}     \notag  \\
		&  \hspace{0.5cm}  +   \frac{ (2+\omega )  K \iota}{(6\pi)^{\frac{1}{3}} t^{2}} \delta_{rs} X^r X^s  (1+f)^{-\omega}	  (1 +\Rho)^{\omega}    D_j  D_i \Rho +\frac{2K\iota(1+f)}{(6\pi)^{\frac{1}{3}} t^2} \delta_{ij} \Bigl[\Bigl(\frac{1+\Rho}{1+f}\Bigr)^{1+\omega}-1\Bigr] \notag   \\
		&   \hspace{0.5cm}  + \omega(2+\omega) K t^{-\frac{4}{3}} \delta_{sr} X^s X^r \rrho^{\frac{1}{3}}  (1+f)^{-\omega}	  (1 +\Rho)^{\omega-1}  D_i \Rho  D_j \Rho +a^{-2}(1+f)^{\frac{2}{3}}   D_j D_i \dwave{\check{\phi}} \notag  \\
		& \hspace{0.5cm}   +    \frac{(3+2\omega)  K  \iota}{(6\pi)^{\frac{1}{3}}t^2} \Bigl(\frac{1+\Rho}{1+f}\Bigr)^\omega     X_i   D_j \Rho   + \frac{(3+2\omega)  K\iota}{(6\pi)^{\frac{1}{3}}t^2} \Bigl(\frac{1+\Rho}{1+f}\Bigr)^\omega   X_j  D_i \Rho     =  -\frac{\kappa f_0}{1+f} \Theta_{ij} .
\end{align}
Taking the trace of \eqref{e:Dttij} (i.e., using $\delta^{ij}$ acting on the both sides of \eqref{e:Dttij}), using \eqref{e:EP.a.3} and noting $\delta^{ij}\Xi_{ij}=0$, we arrive at \eqref{e:Dtth0a}.
It completes the proof.
\end{proof}

Now we are in the position to derive a second-order equation analogous to the one in \cite[eq. $(1.1)$]{Liu2022b}. We present it in the next lemma.

\begin{lemma}\label{t:Dttrho}
	If $(\dwave{\varrho},\dwave{\check{v}}^i,\dwave{\check{\phi}})$ solves the system \eqref{e:EP.a.1}--\eqref{e:EP.a.3}, then  $(\dwave{\varrho},\Xi_{il}, \Omega_{ij},\dwave{\check{\phi}} ,\dwave{\check{v}}^i)$ satisfy an equation
	\begin{align}\label{e:Dtth0}
		&      \partial_t^2\Rho    +\frac{4}{3t} \del{t}\Rho  -   \frac{2  }{3t^{2}}   \dwave{\varrho}(1+\Rho)  - \frac{4(\del{t}\Rho)^2}{3(1+\Rho) }  -   \frac{ (\omega+1)   (2+\omega)  K \iota}{(6\pi)^{\frac{1}{3}}t^2} \delta_{kj} X^k X^j  	\frac{ (1 +\Rho)^{\omega} }{(1+f)^{\omega}} \delta^{il}  D_l  \Rho   D_i\Rho   \notag  \\
		& \hspace{0.5cm}  +\frac{2 f_0}{3(1+f)} \check{\dwave{v}}^i D_i\Rho   - \frac{ 2 K\iota (1+f)}{(6\pi  )^{\frac{1}{3}}t^2}  \Bigl[   \Bigl(\frac{1+\Rho}{1+f}\Bigr)^{1+\omega}  -1 \Bigr]    X^i  D_i\Rho  - a^{-2}(1+f)^{\frac{2}{3} }  \delta^{il}  D_l \uwave{\check{\phi}} D_i\Rho    \notag  \\
		&   \hspace{0.5cm}  -\frac{8 \dwave{\cv}^i D_i\Rho\del{t}\Rho }{3(1+\Rho) }
		- (1+\Rho) \delta^{ij} \Xi_{il} \delta^{lk}\Xi_{kj}      - \frac{4(\check{\dwave{v}}^i D_i\Rho)^2}{3(1+\Rho) }    -(1+\Rho)\delta^{ij} \Omega_{kj} \Omega_{il}  \delta^{lk}     + \check{\dwave{v}}^i \dwave{\check{v}}^k D_k D_i\Rho    \notag  \\
		&  \hspace{0.5cm} -   \frac{ (2+\omega )  K \iota}{(6\pi)^{\frac{1}{3}} t^{2}} \delta_{rs} X^r X^s   \frac{(1+\Rho)^{\omega+1} }{(1+f)^\omega }  \delta^{ij}  D_j  D_i \Rho -\frac{6K\iota}{(6\pi)^{\frac{1}{3}} t^2}   \Bigl[  \frac{(1+\Rho)^{\omega}}{(1+f)^{\omega}} - 1 \Bigr](1+\Rho)^2   + 2\check{\dwave{v}}^i \del{t} D_i\Rho   \notag   \\
		&    \hspace{0.5cm}    - \frac{2 (3+2\omega)  K  \iota  }{(6\pi)^{\frac{1}{3}}t^2}   \frac{(1+\Rho)^{\omega+1 }    }{(1+f)^{\omega }    }  X^i D_i \Rho    =    \frac{\kappa f_0}{1+f}\Bigl(\frac{1+\Rho}{1+f}f_0 -\del{t}\Rho\Bigr)  .
\end{align}
\end{lemma}	
\begin{proof}
	This result again can be proven by straightforward calculations by substituting \eqref{e:th2} into \eqref{e:Dtth0a} (see Lemma \ref{t:dtth}).
Firstly, by \eqref{e:th2}, direct calculations lead to
	\begin{gather}\label{e:thsq}
			\frac{1}{3} \Theta^2
			=   \frac{f_0^2}{3(1+f)^2} + \frac{(\del{t}\Rho)^2}{3(1+\Rho)^2} + \frac{(\check{\dwave{v}}^i D_i\Rho)^2}{3(1+\Rho)^2}- \frac{2f_0}{3(1+f)} \frac{\del{t}\Rho}{1+\Rho}-\frac{2 f_0}{3(1+f)}\frac{\check{\dwave{v}}^i D_i\Rho}{1+\Rho}+\frac{2 \dwave{\cv}^i D_i\Rho\del{t}\Rho }{3(1+\Rho)^2}  ,   \\
			\intertext{and}
	\dwave{\check{v}}^k D_k \Theta
	= 	  -\frac{\dwave{\check{v}}^k D_k \del{t}\Rho}{1+\Rho}+\frac{\dwave{\check{v}}^k D_k\Rho \del{t}\Rho}{(1+\Rho)^2} -\frac{\dwave{\check{v}}^k D_k\check{\dwave{v}}^i D_i\Rho}{1+\Rho}  -\frac{\check{\dwave{v}}^i \dwave{\check{v}}^k D_k D_i\Rho}{1+\Rho} + \frac{\dwave{\check{v}}^k  \check{\dwave{v}}^i D_k\Rho  D_i\Rho }{(1+\Rho)^2} .   \label{e:vdtth}
\end{gather}

Next, let us calculate $D_t \Theta$. By using \eqref{e:feq1b} to replace $f^{\prime\prime}$, we obtain
\begin{align*}
 &	 \del{t} \Bigl( \frac{f_0}{1+f}\Bigr)=  \frac{f^{\prime\prime}}{1+f}-\frac{f_0^2}{(1+f)^2}
	 =   \frac{1}{1+f}\Bigl(-\frac{4}{3t} f_0+\frac{2}{3 t^2} f (1+f ) + \frac{4}{3}  \frac{ f_0^2}{(1+f )}
	 \Bigr)-\frac{f_0^2}{(1+f)^2}  \notag  \\
& \hspace{3cm} 	 =   -\frac{4}{3t}  \frac{f_0 }{1+f} +\frac{2}{3 t^2} f  + \frac{1}{3}  \frac{ f_0^2}{(1+f )^2}  .
\end{align*}
By using \eqref{e:th2} to replace $\Theta$ and substituting $\del{t}  \check{\V}^k $ by the equation \eqref{e:EP.a.2}, we arrive at
\begin{align}\label{e:Dtth5}
		\del{t} \Theta
		= &  -\frac{4}{3t}  \frac{f_0 }{1+f} +\frac{2}{3 t^2} f  + \frac{1}{3}  \frac{ f_0^2}{(1+f )^2}   -\frac{\partial_t^2\Rho}{1+\Rho} + \frac{(\del{t}\Rho)^2}{(1+\Rho)^2} +   \Bigl(\frac{4}{3t}   -  \frac{2 f_0}{3(1+f)} +\frac{\kappa f_0}{1+f} \Bigr)   \frac{  \check{\V}^i D_i\Rho}{1+\Rho} \notag  \\
		&  +  \frac{ \dwave{\check{v}}^k D_k  \check{\V}^i D_i\Rho}{1+\Rho}    +    (2+\omega)  K t^{-\frac{4}{3}} \delta_{kj} X^k X^j \rrho^{\frac{1}{3}}  (1+f)^{-\omega}	 (1 +\Rho)^{\omega-1} \delta^{il}  D_l  \Rho   D_i\Rho  + \frac{\check{\dwave{v}}^i D_i\Rho \del{t}  \Rho}{(1+\Rho)^2}   \notag  \\
		&  + \frac{ 2 K\iota (1+f)}{(6\pi  )^{\frac{1}{3}}t^2}  \Bigl[   \Bigl(\frac{1+\Rho}{1+f}\Bigr)^{1+\omega}  -1 \Bigr] \frac{   X^i  D_i\Rho}{1+\Rho}    + a^{-2}(1+f)^{\frac{2}{3} }\frac{ \delta^{il}  D_l \dwave{\check{\phi}} D_i\Rho}{1+\Rho}  -\frac{\check{\dwave{v}}^i \del{t} D_i\Rho}{1+\Rho}    .
\end{align}

Then using \eqref{e:th2}, \eqref{e:thsq}--\eqref{e:vdtth} and \eqref{e:Dtth5} to replace the corresponding terms in \eqref{e:Dtth0a}, noting the cancellations of terms $ H \tv^i D_i  \varrho/(1+ \varrho)   $, $ a^{-1}    \tv^k   D_k\tv^i  D_i  \varrho /(1+ \varrho) $ and $ K a^{-1}  D^i  \varrho D_i \varrho /(1 +\varrho)^2 $, with straightforward computations, we arrive at
	\begin{align}\label{e:Dtth0.a}
		&    -\frac{\partial_t^2\Rho}{1+\Rho} + \frac{4(\del{t}\Rho)^2}{3(1+\Rho)^2}-\frac{2 f_0}{3(1+f)}\frac{\check{\dwave{v}}^i D_i\Rho}{1+\Rho}    +   \frac{ (\omega+1)   (2+\omega)  K \iota}{(6\pi)^{\frac{1}{3}}t^2} \delta_{kj} X^k X^j  	\frac{ (1 +\Rho)^{\omega-1} }{(1+f)^{\omega}} \delta^{il}  D_l  \Rho   D_i\Rho   \notag  \\
		&   \hspace{0.5cm}    + \frac{ 2 K\iota (1+f)}{(6\pi  )^{\frac{1}{3}}t^2}  \Bigl[   \Bigl(\frac{1+\Rho}{1+f}\Bigr)^{1+\omega}  -1 \Bigr] \frac{   X^i  D_i\Rho}{1+\Rho} + a^{-2}(1+f)^{\frac{2}{3} }\frac{ \delta^{il}  D_l \dwave{\check{\phi}} D_i\Rho}{1+\Rho}   -\frac{2\check{\dwave{v}}^i \del{t} D_i\Rho}{1+\Rho}  \notag  \\
		&   \hspace{0.5cm}    +\frac{8 \dwave{\cv}^i D_i\Rho\del{t}\Rho }{3(1+\Rho)^2}        -\frac{4}{3t} \frac{\del{t}\Rho}{1+\Rho}
		+  \delta^{ij} \Xi_{il} \delta^{lk}\Xi_{kj}   + \frac{4(\check{\dwave{v}}^i D_i\Rho)^2}{3(1+\Rho)^2}  +\delta^{ij} \Omega_{kj} \Omega_{il}  \delta^{lk}     -\frac{\check{\dwave{v}}^i \dwave{\check{v}}^k D_k D_i\Rho}{1+\Rho}   \notag  \\
		& \hspace{0.5cm}  +   \frac{ (2+\omega )  K \iota}{(6\pi)^{\frac{1}{3}} t^{2}} \delta_{rs} X^r X^s  \Bigl(\frac{1+\Rho}{1+f}\Bigr)^\omega   \delta^{ij}  D_j  D_i \Rho +\frac{6K\iota}{(6\pi)^{\frac{1}{3}} t^2}   \Bigl[(1+f)\Bigl(\frac{1+\Rho}{1+f}\Bigr)^{1+\omega}-1  \Bigr] \notag   \\
		&   \hspace{0.5cm}  +   \frac{2 \iota^3 }{3t^{2}}   \dwave{\varrho}      + \frac{2 (3+2\omega) K  \iota  }{(6\pi)^{\frac{1}{3}}t^2} \Bigl(\frac{1+\Rho}{1+f}\Bigr)^{\omega }    X^i D_i \Rho    =  -\frac{\kappa f_0}{1+f}\Bigl(\frac{f_0}{1+f}-\frac{\del{t}\Rho}{1+\Rho}\Bigr) .
\end{align}
In the end, by using \eqref{e:ioeq}, we have an identity
\begin{equation*}
	 \frac{2 \iota^3}{3t^{2}}   \dwave{\varrho}   =    \frac{2 (\iota^3-1)}{3t^{2}}   \dwave{\varrho}   +   \frac{2  }{3t^{2}}   \dwave{\varrho}   =-\frac{6K\iota}{(6\pi  )^{\frac{1}{3}}t^{2}}   \dwave{\varrho}   +   \frac{2  }{3t^{2}}   \dwave{\varrho} .
\end{equation*}
 With the help of this identity into \eqref{e:Dtth0.a}, we conclude \eqref{e:Dtth0}.
\end{proof}

\subsection{Spherical symmetric formulations}\label{s:SSd}
If we confine ourselves to \textit{spherically symmetric solutions}, we can simplify the equation \eqref{e:Dtth0} and the Euler--Poisson system \eqref{e:EP.a.1}--\eqref{e:EP.a.3}. In this section, we rewrite the Euler--Poisson system \eqref{e:EP.a.1}--\eqref{e:EP.a.3} and \eqref{e:Dtth0} under the assumption of spherical symmetry.

If we consider \textit{spherically symmetric solutions}, the velocity field can be written in terms of the rescaled speed $\nu$ as follows:
\begin{equation}\label{e:chv}
	\dwave{\check{v}}^i(t,X^k)=Z(t) \nu(t,R) X^i \quad \text{i.e.,} \; \nu(t,R):=Z^{-1}(t)\dwave{\check{v}}^i(t,X^k)X_i  
\end{equation}
where we denote $Z(t):=\frac{f_0}{3(1+f)} $ and  $R=|\mathbf{X}|=\sqrt{\delta_{ij}X^iX^j}$.

We remark in the following, we slightly \textit{abuse the notations} of functions $\nu$. We use $\nu$ to represent $\nu(t,R)$ and $\nu(t,X^k)$ interchangeably to simplify the notations and it will be clear from context which one we are using.  This convention also applies to for $\Theta$ and $\del{R}\nu$.
\begin{lemma}\label{t:dtvkid}
	For the spherical symmetric solutions, $\Rho$ and $\nu$ (defined by \eqref{e:chv}) satisfy equations,
\begin{align}\label{e:dtnveq}
	&  \del{t}   \nu     + \Bigl(\frac{2}{3t^2} \frac{(1+f) f }{f_0}   -  \frac{ f_0}{3 (1+f ) } \Bigr) \nu     	
	+    \frac{f_0}{3(1+f)} \nu X^j D_{j}  \nu     +   \frac{f_0}{3 (1+f)} \nu^2     \notag \\
	&\hspace{1cm}  +  \frac{3 (2+\omega)  K   \iota}{(6\pi)^{\frac{1}{3}}t^2}\frac{(1+f)^{1-\omega}  (1 +\Rho)^{\omega} }{f_0}  X_k \delta^{kl}  D_l  \Rho  +  \frac{ 6 K\iota (1+f)^2}{(6\pi  )^{\frac{1}{3}}t^2 f_0}    \Bigl[   \Bigl(\frac{1+\Rho}{1+f}\Bigr)^{1+\omega}  -1 \Bigr]   \notag \\
	& \hspace{1cm}+\frac{3a^{-2}(1+f)^{\frac{5}{3}}   }{f_0} \frac{X_k}{R^2} \delta^{kl}  D_l \dwave{\check{\phi}} =        -\frac{\kappa f_0}{1+f}  \nu   ,
\end{align}
and
\begin{align}\label{e:Dtth2.a}
	&      \partial_t^2\Rho   +  \frac{2 f_0}{3(1+f)} \nu  X^i \del{t} D_i\Rho   +\frac{4}{3t} \del{t}\Rho  -   \frac{2  }{3t^{2}}   \dwave{\varrho}(1+\Rho)  - \frac{4(\del{t}\Rho)^2}{3(1+\Rho) } +\frac{2 f_0^2}{9(1+f)^2}  \nu X^i D_i\Rho      \notag  \\
	& \hspace{0.5cm} -   \frac{ (\omega+1)   (2+\omega)  K \iota}{(6\pi)^{\frac{1}{3}}t^2} \delta_{kj} X^k X^j  	\frac{ (1 +\Rho)^{\omega} }{(1+f)^{\omega}} \delta^{il}  D_l  \Rho   D_i\Rho        + \frac{f_0^2}{9(1+f)^2} \nu^2 X^i   X^k D_k D_i\Rho  \notag  \\
	&  \hspace{0.5cm}  - \frac{ 2 K\iota (1+f)}{(6\pi  )^{\frac{1}{3}}t^2}  \Bigl[   \Bigl(\frac{1+\Rho}{1+f}\Bigr)^{1+\omega}  -1 \Bigr]    X^i  D_i\Rho  - a^{-2}(1+f)^{\frac{2}{3} }  \delta^{il}  D_l \uwave{\check{\phi}} D_i\Rho      \notag  \\
	&  \hspace{0.5cm}  -\frac{8  f_0  \nu  X^i D_i\Rho\del{t}\Rho }{9(1+f)(1+\Rho) }
	- \frac{2}{3} (1+\Rho)   \Bigl(\Theta-\frac{f_0}{(1+f)} \nu\Bigr)^2  - \frac{  4f_0^2 \nu^2 ( X^i D_i\Rho)^2}{27 (1+f)^2 (1+\Rho) }     \notag  \\
	& \hspace{0.5cm}  -   \frac{ (2+\omega )  K \iota}{(6\pi)^{\frac{1}{3}} t^{2}} \delta_{rs} X^r X^s   \frac{(1+\Rho)^{\omega+1} }{(1+f)^\omega }  \delta^{ij}  D_j  D_i \Rho -\frac{6K\iota}{(6\pi)^{\frac{1}{3}} t^2}   \Bigl[  \frac{(1+\Rho)^{\omega}}{(1+f)^{\omega}} - 1 \Bigr](1+\Rho)^2 \notag   \\
	&   \hspace{0.5cm}   - \frac{2 (3+2\omega)  K  \iota  }{(6\pi)^{\frac{1}{3}}t^2}   \frac{(1+\Rho)^{\omega+1 }    }{(1+f)^{\omega }    }  X^i D_i \Rho    =    \frac{\kappa f_0}{1+f}\Bigl(\frac{1+\Rho}{1+f}f_0 -\del{t}\Rho\Bigr)   .
\end{align}
Moreover, there is a crucial identity (an alternative form of the continuity equation),
\begin{align}\label{e:Th6}
	\Theta = \frac{f_0}{1+f}  \nu  + \frac{f_0}{3(1+f)}  X^i D_i \nu = \frac{f_0}{1+f} -\frac{\del{t}\Rho}{1+\Rho} - \frac{f_0}{3(1+f)} \frac{  \nu X^i D_i \Rho}{1+\Rho}  .
\end{align}
\end{lemma}

\begin{proof}
	To prove \eqref{e:Dtth2.a}, we need to calculate $	\delta^{ij}	\Xi_{il} \delta^{lk} \Xi_{kj} $ and $\delta^{ij} \Omega_{kj} \Omega_{il}  \delta^{lk}   $ and then substitute them into \eqref{e:Dtth0}. By \eqref{e:chv}, we have
	\begin{equation}\label{e:divj}
		D_i \dwave{\cv}_j= Z(t) \nu \delta_{ij} + Z(t) X_j D_i\nu
	\end{equation}
We denote $\del{R}:=\frac{\partial}{\partial R}$, then
\begin{equation}\label{e:dxdrv}
  \frac{\partial R}{\partial X^i}=\frac{\delta_{ij}X^j}{R}  \quad \Rightarrow   \quad   X^iD_i\nu=R\del{R}\nu \AND	D_i \nu = \frac{X_i}{R} \del{R}\nu.
\end{equation}
Thus, by \eqref{e:Th1}, \eqref{e:divj} and \eqref{e:dxdrv}, we obtain
\begin{equation}\label{e:omg}
	\Omega_{ij}
	= \frac{1}{2}Z(t)(   X_j D_i\nu  -  X_i D_j\nu)= \frac{1}{2}Z(t) \Bigl(  \frac{  X_jX_i}{R} \del{R} \nu  -  \frac{  X_i X_j}{R}\del{R}\nu\Bigr)=0.
\end{equation}

To calculate $\delta^{ij}	\Xi_{il} \delta^{lk} \Xi_{kj}$, by \eqref{e:Th1}, we note	
\begin{equation}\label{e:xith}
	\delta^{ij}	\Xi_{il} \delta^{lk} \Xi_{kj}=  	\delta^{ij} \Theta_{il}\delta^{lk}\Theta_{kj} -\frac{1}{3} \Theta^2  .
\end{equation}
Hence, we calculate $	\delta^{ij} \Theta_{il}\delta^{lk}\Theta_{kj} $ (by \eqref{e:Th1}) and $\Theta$ (by \eqref{e:th0}), with the help of \eqref{e:divj}:
\begin{align}\label{e:Th5}
	\Theta=&\delta^{ij}D_i\dwave{\cv}_j
	= 3Z(t) \nu  + Z(t)   \frac{\delta^{ij} X_iX_j}{R}\del{R}\nu	= 3Z(t) \nu  + Z(t)   R \del{R}\nu
\end{align}
and
\begin{align}\label{e:ThTh}
	\delta^{ij} \Theta_{il}\delta^{lk}\Theta_{kj} = &\frac{1}{2} \delta^{ij} D_i \dwave{\cv}_l\delta^{lk}(D_j\dwave{\cv}_k+D_k\dwave{\cv}_j)  \notag  \\
	= &Z^2(t) \delta^{ij} \Bigl(  \nu \delta_{il} +   \frac{X_lX_i}{R}\del{R}\nu \Bigr)\delta^{lk}\Bigl(  \nu \delta_{jk} +    \frac{X_jX_k}{R}\del{R} \nu  \Bigr)
	\notag  \\
	= &Z^2(t)\bigl( 3 \nu^2  + 2\nu  R\del{R} \nu  + R^2(\del{R} \nu)^2  \bigr) \notag  \\
	\overset{\text{by}\; \eqref{e:Th5}}{=} &  3 Z^2(t) \nu^2  + 2 \nu  Z(t)  (\Theta-3 Z(t) \nu)  +  (\Theta-3 Z(t) \nu)^2    \notag  \\
	= &    \Theta^2- 4 Z(t) \nu \Theta+ 6Z^2(t) \nu^2   .
\end{align}
Then by putting \eqref{e:Th5} and \eqref{e:ThTh} into \eqref{e:xith}, we arrive at
\begin{equation}\label{e:XiXi}
	\delta^{ij}	\Xi_{il} \delta^{lk} \Xi_{kj}
	=     \frac{2}{3} \Theta^2- 4 Z(t) \nu \Theta+ 6Z^2(t) \nu^2
 	=  \frac{2}{3} \bigl(\Theta-3Z(t)\nu\bigr)^2
=\frac{2}{3} \Bigl(\Theta-\frac{f_0}{(1+f)} \nu\Bigr)^2 .
\end{equation}

Substituting \eqref{e:chv}, \eqref{e:omg} and \eqref{e:XiXi} to \eqref{e:Dtth0}, we conclude \eqref{e:Dtth2.a}, and combining \eqref{e:dxdrv}, \eqref{e:Th5} and \eqref{e:th2} yields \eqref{e:Th6}.

 In the end, let us calculate \eqref{e:dtnveq}. Substituting \eqref{e:chv} into \eqref{e:EP.a.2}, with the help of \eqref{e:feq1b} to replace $f^{\prime\prime}$, we arrive at
\begin{align*}
	&X^k  \del{t}   \nu     + \frac{2}{3t^2} \frac{(1+f) f }{f_0}X^k \nu -  \frac{ f_0}{3 (1+f ) }  X^k \nu     	
	+    \frac{f_0}{3(1+f)} \nu  X^k X^j D_{j}  \nu       \notag \\
	& \hspace{1cm}  +   \frac{f_0}{3 (1+f)} \nu^2  X^k  + 3 (2+\omega)  K t^{-\frac{4}{3}} \delta_{ij} X^i X^j \rrho^{\frac{1}{3}} \frac{(1+f)^{1-\omega}  (1 +\Rho)^{\omega} }{f_0}  \delta^{kl}  D_l  \Rho  \notag   \\
	&\hspace{1cm} +  \frac{ 6 K\iota (1+f)^2}{(6\pi  )^{\frac{1}{3}}t^2 f_0} X^k    \Bigl[   \Bigl(\frac{1+\Rho}{1+f}\Bigr)^{1+\omega}  -1 \Bigr]   +\frac{3a^{-2}(1+f)^{\frac{5}{3}}   }{f_0} \delta^{kl}  D_l \dwave{\check{\phi}} =  -\frac{\kappa f_0}{1+f}   \nu X^k   .
\end{align*}
 Multiplying $X_k/R^2$ on the both sides of the above equation, we can conclude \eqref{e:dtnveq}, which completes the proof.
\end{proof}

The next step is to rewrite Lemma \ref{t:dtvkid} in terms of the \textit{spherical coordinate} which is given in Appendix \ref{s:SScd}.
Since all the variables are spherical symmetric,  let us transform $(\Rho, \dwave{\cv}^i,\dwave{\check{\phi}})$ from the coordinate $\{t,X^i\}$ to the spherical coordinate $\{t,R,\theta,\varphi\}$ and denote
\begin{equation}\label{e:rdvar}
	\hrho(t,R):= \Rho(t,X^k),\quad \hv(t,R):= \frac{\delta_{ij} X^j}{R} \dwave{\cv}^i  (t,X^k) \AND 	\hphi(t,R): = \dwave{\check{\phi}}(t,X^k) ,
\end{equation}
and, by \eqref{e:chv} and \eqref{e:rdvar}, we have
\begin{equation}\label{e:vv}
	\dwave{\cv}^i(t,X^k)=\frac{X^i}{R}\hv(t,R) \AND
		\nu(t,R)= \frac{3(1+f(t))} {f_0(t) R }\hv(t,R) .
\end{equation}
Direct computations, by \eqref{e:th0},  imply the expansion can be represented by
\begin{equation*}
\Theta  =\frac{1}{R^2} \del{R} (R^2 \hv )=\del{R}\hv +\frac{2}{R}\hv  .
\end{equation*}

Under spherical coordinate \eqref{e:Spcrd2}, Lemma \ref{t:dtvkid} becomes:
\begin{lemma}\label{t:dtvss}
	Under spherical coordinate \eqref{e:Spcrd2}, the spherical symmetric solutions $(\hrho,\nu)$ satisfy equations
\begin{align}\label{e:Dtth2c}
	&      \partial_t^2\hrho   +  \frac{2 f_0}{3(1+f)} \nu  R\del{R} \del{t} \hrho   +\frac{4}{3t} \del{t}\hrho  -   \frac{2  }{3t^{2}}   \hrho(1+\hrho)  - \frac{4(\del{t}\hrho)^2}{3(1+\hrho) } +\frac{2 f_0^2}{9(1+f)^2}  \nu R\del{R} \hrho      \notag  \\
	& \hspace{0.5cm} -   \frac{ (\omega+1)   (2+\omega)  K \iota}{(6\pi)^{\frac{1}{3}}t^2} 	\frac{ (1 +\hrho)^{\omega} }{(1+f)^{\omega}} (R  \del{R}   \hrho )^2       + \frac{f_0^2}{9(1+f)^2} \nu^2 R^2\partial_R^2 \hrho   - \frac{  4f_0^2 \nu^2 ( R \del{R} \hrho)^2}{27 (1+f)^2 (1+\hrho) }  \notag  \\
	&  \hspace{0.5cm}  - \frac{ 2 K\iota (1+f)}{(6\pi  )^{\frac{1}{3}}t^2}  \Bigl[   \Bigl(\frac{1+\hrho}{1+f}\Bigr)^{1+\omega}  -1 \Bigr]    R\del{R} \hrho  -    \frac{2 \iota^3   }{3t^{2}R^2}  \del{R}\hrho   \int^R_0 ( \hrho(t,y) -f(t)) y^2 dy    \notag  \\
	&  \hspace{0.5cm}  -\frac{8  f_0  \nu  R \del{R} \hrho\del{t}\hrho }{9(1+f)(1+\hrho) }
	- \frac{2}{3} (1+\hrho)   \Bigl(  \frac{f_0}{1+f} -\frac{\del{t}\hrho}{1+\hrho} -\frac{f_0}{3(1+f)}\frac{  \nu   R  \del{R} \hrho}{1+\hrho} -\frac{f_0}{(1+f)} \nu\Bigr)^2  \notag \\
	& \hspace{0.5cm}  -   \frac{ (2+\omega )  K \iota}{(6\pi)^{\frac{1}{3}} t^{2}}    \frac{(1+\hrho)^{\omega+1} }{(1+f)^\omega }  R^2 \partial_R^2 \hrho  -   \frac{ 2 (2+\omega )  K \iota}{(6\pi)^{\frac{1}{3}} t^{2}}    \frac{(1+\hrho)^{\omega+1} }{(1+f)^\omega }  R \del{R} \hrho \notag\\
	&\hspace{0.5cm} -\frac{6K\iota}{(6\pi)^{\frac{1}{3}} t^2}   \Bigl[  \frac{(1+\hrho)^{\omega}}{(1+f)^{\omega}} - 1 \Bigr](1+\hrho)^2    - \frac{2 (3+2\omega)  K  \iota  }{(6\pi)^{\frac{1}{3}}t^2}   \frac{(1+\hrho)^{\omega+1 }    }{(1+f)^{\omega }    }  R \del{R} \hrho   \notag\\
	&\hspace{0.5cm} =   \frac{\kappa f_0}{1+f}\Bigl(\frac{1+\hrho}{1+f}f_0 -\del{t}\hrho\Bigr)    .
\end{align}
	and
\begin{align}\label{e:Dttv2c}
	&  \del{t}   \nu     + \Bigl( \frac{2}{3t^2} \frac{(1+f) f }{f_0}  -  \frac{ f_0}{3 (1+f ) } \Bigr) \nu     	
	+    \frac{f_0}{3(1+f)} \nu R\del{R}  \nu     +   \frac{f_0}{3 (1+f)} \nu^2      \notag \\
	&\hspace{1cm}  +  \frac{3 (2+\omega)  K   \iota}{(6\pi)^{\frac{1}{3}}t^2}\frac{(1+f)^{1-\omega}  (1 +\hrho)^{\omega} }{f_0}  R\del{R}   \hrho  +  \frac{ 6 K\iota (1+f)^2}{(6\pi  )^{\frac{1}{3}}t^2 f_0}    \Bigl[   \Bigl(\frac{1+\hrho}{1+f}\Bigr)^{1+\omega}  -1 \Bigr]    \notag \\
	&\hspace{1cm} +   \frac{2 \iota^3 (1+f)  }{t^{2}  f_0 R^3} \int^R_0 ( \hrho(t,y) -f(t)) y^2 dy =  -\frac{\kappa f_0}{1+f} \nu ,
\end{align}
Moreover, there is a crucial identity (an alternative form of the continuity equation),
\begin{align}\label{e:keyid}
	\Theta =  \frac{f_0}{1+f} -\frac{\del{t}\hrho}{1+\hrho} -\frac{f_0}{3(1+f)}\frac{  \nu   R  \del{R} \hrho}{1+\hrho} =  \frac{ f_0}{ 1+f} \nu  + \frac{f_0}{3(1+f)} R \del{R}\nu
\end{align}
\end{lemma}
\begin{proof}
	This proof again relies on straightforward calculations, and we still only point out some key steps but leave the rest to the readers. We note in the spherical coordinate, using Appendix \ref{s:SScd} and in view of the spherical symmetric solution, $	\dwave{\cv}^i  \del{t} D_i   \Rho $ becomes $\frac{f_0}{3(1+f)} R \nu \del{R}\del{t}\hrho$, $\delta^{ij}D_iD_j   \Rho$ becomes  $\frac{1}{R^2} \del{R}(R^2 \del{R} \hrho)$, $\dwave{\cv}^k \dwave{\cv}^i  D_k D_i  \Rho$  becomes $ \frac{f_0^2}{9(1+f)^2}  \nu^2     R^2 \partial_R^2 \hrho   $  and $X^iX^kD_kD_i \Rho$ becomes $R^2\partial_R^2 \hrho$.  If we note that using the Poisson equation \eqref{e:EP.a.3}, $D_l \dwave{\check{\phi}} $ can be explicitly expressed by the following \eqref{e:dphiq}, then we are able to use these to rewrite \eqref{e:dtnveq}--\eqref{e:Th6} to \eqref{e:Dtth2c}--\eqref{e:keyid}, respectively.

	Now let us express $D_l \dwave{\check{\phi}} $ by solving the Poisson equation \eqref{e:EP.a.3}.
	In terms of the spherical coordinate, the Poisson equation \eqref{e:EP.a.3} becomes
	\begin{equation}\label{e:ddphi2}
		\del{R}(R^2 \del{R} \hphi) =    \frac{2 \iota^3}{3t^{\frac{2}{3}}} (1+f)^{-\frac{2}{3}}  ( \hrho -f) R^2.
	\end{equation}
	Then integrating \eqref{e:ddphi2}  yields
	\begin{equation}\label{e:dphiq}
		\del{R} \hphi =    \frac{2 \iota^3 (1+f)^{-\frac{2}{3}} }{3t^{\frac{2}{3}}R^2} \int^R_0 ( \hrho(t,y) -f(t)) y^2 dy.
	\end{equation}
It completes the proof.
\end{proof}

\subsection{Log-periodic formulations}\label{s:logfml}
Since we only consider the \textit{$t$-log-periodic solutions} (as defined in Definition \ref{e:lgprd}) in this article, we need to introduce the \textit{logarithmic coordinate} by using the following transformation for simplicity:
\begin{equation}\label{e:logcd}
	\Rbb \ni \zeta :=\ln R \quad \text{i.e.,}\; R=e^{\zeta} \quad\text{where} \; R\in(0,\infty) .
\end{equation}

We remark that we will slightly \textit{abuse the notation} for the functions $\hrho$ and $\nu$ in the following discussion. Specifically, we will use $\nu$ to refer to both $\nu(t,R)$ and $\nu(t,\zeta)$ interchangeably in order to simplify the notation. It will be clear from the context which one we are referring to. This convention also applies to $\hrho$ and its derivatives.

Then the direct calculations yield
\begin{gather}
	\del{\zeta}=\frac{d R}{d\zeta}\del{R}=R \del{R},  \label{e:dzdr}   \\
	\frac{1}{R^3} \int^R_0 ( \hrho(t,y) -f(t)) y^2 dy =\frac{1}{e^{3\zeta}} \int^{\zeta}_{-\infty} ( \hrho(t,z) -f(t)) e^{3z} dz,  \label{e:dph2}  \\
		R^2\partial_R^2 \hrho =R\del{R}(R\del{R} \hrho) -R\del{R} \hrho =\partial_\zeta^2 \hrho -\del{\zeta} \hrho   . \label{e:R2dr2}
\end{gather}

By using these basic identities, we obtain the following system where \eqref{e:main1} resembles the model equation presented in \cite[eq. $(1.1)$]{Liu2022b}. However, we would like to emphasize that the remaining terms $F_1$ and $\mathcal{g}^{\zeta\zeta}$ are more complex when compared to the corresponding terms in \cite[eq. $(1.1)$]{Liu2022b}.

\begin{lemma}\label{t:lgfml}
	The spherical symmetric solutions $(\hrho,\nu)$ of the Euler--Poisson system satisfy
\begin{align}
	\Box_{\mathcal{g}} \hrho      +   \Bigl( \frac{4}{3t}    + \frac{\kappa f_0}{1+f} \Bigr) \del{t}\hrho   -   \frac{2}{3t^2} \hrho  (1+ \hrho)   -     \frac{4(\partial_t \hrho )^2 }{3 (1+ \hrho)} = & F_1 , \label{e:main1} \\
	\del{t} \nu
	+  \frac{f_0}{3 (1+f)}\nu  	\del{\zeta} \nu = & G_1  ,    \label{e:main2}
\end{align}
	where the wave operator  is
\begin{gather}
	\Box_{\mathcal{g}}:= \partial_t^2 - \mathcal{g}^{\zeta\zeta}\partial_\zeta^2 +2 \cg^{0\zeta} \del{\zeta}\del{t},   \label{e:g1} \\
	\mathcal{g}^{\zeta\zeta}:=
	 \frac{ (2+\omega )  (1-\iota^3) }{ 9t^{2}}    \frac{(1+\hrho)^{\omega+1} }{(1+f)^\omega }  -\frac{f_0^2}{9(1+f)^2} \nu^2  ,   \quad
	 \mathcal{g}^{0\zeta} :=  \frac{f_0}{3(1+f)} \nu ,   \label{e:g2}
\end{gather}
and the remainders are
	\begin{align}
	F_1
	:=&     - \frac{2 f_0^2}{9(1+f)^2}  \nu \del{\zeta} \hrho    +   \frac{ (\omega+1)   (2+\omega) (1 - 	\iota^3 ) }{9 t^2}  	\frac{ (1 +\hrho)^{\omega} }{(1+f)^{\omega}} (\del{\zeta} \hrho )^2      + \frac{f_0^2}{9(1+f)^2} \nu^2  \del{\zeta} \hrho  \notag  \\
	&     + \frac{ 2  (1 - 	\iota^3 )   (1+f)}{9 t^2}  \Bigl[   \Bigl(\frac{1+\hrho}{1+f}\Bigr)^{1+\omega}  -1 \Bigr]    \del{\zeta} \hrho  +   \frac{2 \iota^3 f  }{3t^{2} }  \del{\zeta}\hrho 	\Psi   + \frac{  4f_0^2 \nu^2 ( \del{\zeta} \hrho)^2}{27 (1+f)^2 (1+\hrho) }   \notag  \\
	&     + \frac{8  f_0  \nu  \del{\zeta} \hrho\del{t}\hrho }{9(1+f)(1+\hrho) }
	+ \frac{2}{3} (1+\hrho)   \Bigl(  \frac{f_0}{1+f} -\frac{\del{t}\hrho}{1+\hrho} -\frac{f_0}{3(1+f)}\frac{  \nu   \del{\zeta} \hrho}{1+\hrho} -\frac{f_0}{(1+f)} \nu\Bigr)^2  \notag \\
	& + \frac{2 (1 - 	\iota^3 )  }{ 3 t^2}  \Bigl[  \frac{(1+\hrho)^{\omega}}{(1+f)^{\omega}} - 1 \Bigr](1+\hrho)^2    +  \frac{  (8+5\omega)  (1 - 	\iota^3 )   (1+\hrho)^{\omega+1 } }{  9t^2(1+f)^{\omega }    }       \del{\zeta} \hrho  +  \frac{\kappa f_0^2(1+\hrho)}{(1+f)^2}    ,    \label{e:F1}\\
	G_1
	:= &  -   \frac{2(1+f) f }{3t^2f_0}  \nu + \Bigl(\frac{1}{3}-\kappa\Bigr) \frac{ f_0}{ 1+f  }  \nu       -   \frac{f_0}{3 (1+f)} \nu^2      -    \frac{(2+\omega)  (1 - 	\iota^3 )(1+f)^{1-\omega}  (1 +\hrho)^{\omega} }{3 t^2 f_0}  \del{\zeta}   \hrho  \notag \\
	&  -  \frac{ 2     (1 - 	\iota^3 )   (1+f)^2}{3 t^2 f_0}   \Bigl[   \Bigl(\frac{1+\hrho}{1+f}\Bigr)^{1+\omega}  -1 \Bigr]  - \frac{2 \iota^3 (1+f) f }{t^{2}  f_0  } \Psi  ,  \label{e:G1}  \\	
	\Psi := & \Psi(t,\zeta) = \frac{1}{f(t) e^{3\zeta}}   \int^{\zeta}_{-\infty} [\hrho(t,z) -f(t)] e^{3z} d z  .    \label{e:Psi0}
	\end{align}

Moreover, the initial data \eqref{e:data0} becomes
\begin{equation}\label{e:data3}
	\hrho|_{t=1}(\zeta)= \beta \mathcal{d} \bigl( (1+\beta)^{-\frac{1}{3}} \exp(\zeta)\bigr) \AND \nu|_{t=1}(\zeta)= 1+\mathcal{v} \bigl((1+\beta)^{-\frac{1}{3}} \exp(\zeta)\bigr).
\end{equation}

In addition, the continuity equation \eqref{e:keyid} becomes
\begin{align}\label{e:keyid5}
	\Theta =  \frac{f_0}{1+f} -\frac{\del{t}\hrho}{1+\hrho} -\frac{f_0}{3(1+f)}\frac{  \nu     \del{\zeta} \hrho}{1+\hrho} =  \frac{ f_0}{ 1+f} \nu  + \frac{f_0}{3(1+f)}   \del{\zeta}\nu
\end{align}
\end{lemma}
\begin{proof}
	The proof still relies on straightforward computations with the help of Lemma \ref{t:dtvss} and  \eqref{e:dzdr}--\eqref{e:R2dr2}. The key point is to use identity \eqref{e:ioeq} to replace terms $K\iota/(6\pi )^{\frac{1}{3}}$, i.e.,
	\begin{equation*}
		\frac{K\iota}{(6\pi )^{\frac{1}{3}}}  = \frac{1}{9} (1 - 	\iota^3 ) .
	\end{equation*}
The data can be directly derived with regular computations by using  \eqref{e:data0}, \eqref{e:data1},  \eqref{e:perv}, \eqref{e:feq2b}, \eqref{e:undf}, \eqref{e:ckvph},  \eqref{e:rdvar}--\eqref{e:vv} and \eqref{e:logcd}.
We omit the details but leave them to the readers.
\end{proof}
\begin{remark}\label{r:hyper}
	$\Box_{\mathcal{g}}$ is a nonlinear wave operator since $\mathcal{g}^{\zeta\zeta}$ is positive if we consider a small enough perturbation of the solution \eqref{e:sl1}--\eqref{e:sl2}. In fact, by later \eqref{e:Gdef} to define $\chi(t)$ and using \eqref{e:gbA}, we can prove
	\begin{equation*}
		\frac{f_0^2}{(1+f)^2} =\frac{f\chi}{Bt^2}   \;  \Rightarrow   \;  \mathcal{g}^{\zeta\zeta}=   \biggl[
		(2+\omega )  (1-\iota^3)     \biggl(\frac{1+\hrho }{1+f} \biggr)^{1+\omega}  -\frac{\nu^2   }{ B}  \frac{f \chi}{1+f}  \biggr] \frac{1+f}{ 9 t^2}.
	\end{equation*}
This wave operator will also be clear later in the Fuchsian formulations.
\end{remark}

\subsection{Periodicity of solutions and log-periodicity of the gravity $\Psi$}\label{s:lggr}
In this section, let us prove if $(\hrho(t,\zeta),\nu(t,\zeta))$ solves \eqref{e:main1}--\eqref{e:main2} with the \textit{periodic data} \eqref{e:data3} and if the solution is unique, then $(\hrho(t,\zeta),\nu(t,\zeta))$ must be a \textit{periodic solution}. Consequently, given the periodicity of the solution, we can modulo the integers $\Zbb$ from $\Rbb$. That is, we can transform the domain from $\Rbb$ to a torus $\Tbb=\Rbb/ \Zbb$ by a quotient of $\Rbb$ under integer shifts.  Additionally, we introduce a lemma that addresses the periodicity of the function $\Psi$ as defined in \eqref{e:Psi0}.

\begin{lemma}\label{t:Psprd}
	Suppose $\hrho(t,\zeta)$ is a periodic function with the period $1$, i.e.,
	\begin{equation}\label{e:rsft}
		\hrho(t,\zeta+n)=\hrho(t,\zeta), \quad \text{for any}\; n\in \Zbb.
	\end{equation}
and $\Psi(t,\zeta)$ is defined by \eqref{e:Psi0}.
	Then $\Psi(t,\zeta)$ is also a periodic function with the period $1$, i.e.,
	\begin{equation*}
		\Psi(t,\zeta+n)=	\Psi(t,\zeta), \quad \text{for any}\; n\in \Zbb.
	\end{equation*}
\end{lemma}
\begin{proof}
	Let $\bar{z}=z-n$, by \eqref{e:Psi0}, we calculate
	\begin{equation}\label{e:Pstr}
		\Psi(t,\zeta+n)=   \frac{1}{f e^{3\zeta}e^{3 n}}   \int^{\zeta+n}_{-\infty} [\hrho(t,z) -f(t)] e^{3z} d z
	=   \frac{1}{f e^{3\zeta} }   \int^{\zeta}_{-\infty} [\hrho(t,\bar{z}+n) -f(t)] e^{3\bar{z}} d \bar{z}    .
	\end{equation}
Further, using the periodicity of $\hrho$ \eqref{e:rsft},  we arrive at
\begin{equation*}
		\Psi(t,\zeta+n)=   \frac{1}{f e^{3\zeta} }   \int^{\zeta}_{-\infty} [\hrho(t,\bar{z}) -f(t)] e^{3\bar{z}} d \bar{z}  =\Psi(t,\zeta) .
\end{equation*}
It completes the proof of this lemma.
\end{proof}

Now we can present the periodicity of solutions which is guaranteed by the uniqueness of the solution.
\begin{proposition}\label{t:prdsl}
	Suppose  $(\hrho(t,\zeta),\nu(t,\zeta))$ solves \eqref{e:main1}--\eqref{e:main2} with the data \eqref{e:data3} where $\mathcal{d}$ and $\mathcal{v}$ are both $1$-log-periodic as the assumption \ref{A4} (see \S\ref{s:model}), and assume the solution is unique, then $(\hrho(t,\zeta),\nu(t,\zeta))$ is a periodic solution satisfying
	\begin{equation}\label{e:prdsl}
		(\hrho(t,\zeta),\nu(t,\zeta))=(\hrho(t,\zeta+n),\nu(t,\zeta+n)),  \quad \text{for any} \; n\in \Zbb.
	\end{equation}
\end{proposition}
\begin{proof}
	\underline{Step $1$ (periodic data):} Since $(\hrho(t,\zeta),\nu(t,\zeta))$ satisfies \eqref{e:data3}, by assumption  \ref{A4} (see \S\ref{s:model}) and Definition \ref{e:lgprd}, we then have, for any $n\in \Zbb$,
	\begin{align*}
			\hrho|_{t=1}(\zeta)= &\beta \mathcal{d} \bigl( (1+\beta)^{-\frac{1}{3}} \exp(\zeta)\bigr)=\beta \mathcal{d} \bigl( (1+\beta)^{-\frac{1}{3}} \exp(\zeta+n)\bigr) =	\hrho|_{t=1}(\zeta+n) , \\  \nu|_{t=1}(\zeta)=& 1+\mathcal{v} \bigl((1+\beta)^{-\frac{1}{3}} \exp(\zeta)\bigr)=1+\mathcal{v} \bigl((1+\beta)^{-\frac{1}{3}} \exp(\zeta+n)\bigr) =\nu|_{t=1}(\zeta+n) .
	\end{align*}
This means the data functions $\hrho|_{t=1}(\zeta)$ and $\nu|_{t=1}(\zeta)$ are periodic.

\underline{Step $2$ (translational invariance):} We only need to substitute $(\hrho(t,\zeta+n),\nu(t,\zeta+n))$ into \eqref{e:main1}--\eqref{e:main2} to verify these equations still hold. By letting $\bar{\zeta}=\zeta+n$, noting   $\del{\zeta}=\del{\bar{\zeta}}$ (since $d\bar{\zeta}/d\zeta =1$) and with the help of \eqref{e:Pstr},  i.e.,
\begin{equation*}
	\frac{1}{f e^{3\zeta} }   \int^{\zeta}_{-\infty} [\hrho(t,\bar{z}+n) -f(t)] e^{3\bar{z}} d \bar{z}  = \Psi(t,\bar{\zeta}),
\end{equation*}
direct calculations lead to   \eqref{e:main1}--\eqref{e:main2} in terms of the new variable $\bar{\zeta}$. Then the obtained equations are exactly \eqref{e:main1}--\eqref{e:main2} with $\bar{\zeta}$ replacing $\zeta$. This means $(\hrho(t,\zeta+n),\nu(t,\zeta+n))$, for all $n\in \Zbb$, are also solutions to \eqref{e:main1}--\eqref{e:main2} with the same data (see Step $1$).  Therefore, the uniqueness of the solution concludes \eqref{e:prdsl}, which completes this proposition.
\end{proof}


\section{Fuchsian formulations and proof of Main Theorem}\label{s:fuchian}	

In \cite{Liu2022b}, we successfully transformed a second-order hyperbolic equation (similar to \eqref{e:main1}) into a Fuchsian system and then applied the global theorem for Fuchsian systems to study the stable blowups of the solutions. In this article, we aim to use a similar approach, however, the complexities of the remaining terms $F_1$ and the nonlinear term $g^{\zeta\zeta}$ (see \eqref{e:g2} and \eqref{e:F1}) require further investigation before reaching a Fuchsian system. Despite this challenge, we still endeavor to transform \eqref{e:main1}--\eqref{e:main2} into a Fuchsian formulation, since it has proven advantageous for long-term nonlinear analysis. As always, our goal is to provide a more comprehensive understanding of the behavior of the solutions to this problem.

In order to achieve this, two questions arise, recalling \ref{c:ffd} and \ref{c:cot} in \S\ref{s:ol}. That is, how to define the Fuchsian fields and the compactified time since the Fuchsian system has a time singularity at $\tau=0$ which implies we have to transform the blowup time to the singular time $\tau=0$.   In order to take advantage of the reference solution \eqref{e:sl1}--\eqref{e:sl2} which involve a crucial function $f(t)$ with a blowup time $t=t_m$, we still try to use the compactified time introduced by \cite{Liu2022b} which is given, in this setting, by
\begin{align}\label{e:ttf}
	\tau := -g(t)=& -\exp\Bigl(-A\int^t_{1} \frac{f(s)(f(s)+1)}{s^2f_0(s)} ds\Bigr)  \notag  \\
	\overset{(\star)}{=} & -\Bigl(1+ \frac{2}{3} B \int^t_{1} s^{-\frac{2}{3}} f(s)(1+f(s))^{-\frac{1}{3}}  ds \Bigr)^{-\frac{3A}{2}}\in[-1,0),
\end{align}
where $B:= (1+\mf)^{\frac{1}{3}}/(  3\gamma )>0 $ is a constant depending on the data (recall $\beta$ and $\gamma$ given by \eqref{e:feq2b}), $f(t)$ is given by the solution to \eqref{e:feq1b}--\eqref{e:feq2b}, $g(t)$ is a function defined by \eqref{e:gdef} (equivalently, by \eqref{e:gdef2}) and $A\in(0,2)$ is an arbitrary constant.
\begin{remark}
	 Let us remark on some facts from \cite{Liu2022b} about this compactified time:
	\begin{enumerate}[leftmargin=*,label={(T\arabic*)}]
		\item Note $(\star)$ in \eqref{e:ttf} holds due to Lemma \ref{t:f0fg} in Appendix \ref{t:refsol} (see \cite[Lemma $2.1$]{Liu2022b} for the detailed proof).
		\item This time transform $\tau$ maps the initial time $t=1$ to $\tau=-1$ and the maximal time of existence $t=t_m$ to $\tau=0$ due to Lemma \ref{t:gmap}.
	\end{enumerate}
\end{remark}

In the following, in order to use calculations and results in \cite{Liu2022b}, we start with the variables therein.

\subsection{Pre-Fuchsian system}\label{s:prefuc}

The transforms and calculations presented in this section are similar to those in \cite[\S$3.1$--\S$3.2$]{Liu2022b}, with more involved computations. For that reason, we will not go into the details of the calculations here; readers can refer to \cite{Liu2022b} for more information. 
Assuming $f(t)$ solves equation \eqref{e:feq1b}-\eqref{e:feq2b}, we denote:
\begin{align}
	w(t,\zeta)& := \hrho(t,\zeta)-f(t), \label{e:ww} \\
	w_{0}(t,\zeta)&:=\partial_{t}w(t,\zeta)=\partial_{t}\hrho(t,\zeta)- f_0(t), \label{e:w0}
	\\
	w_{\zeta}(t,\zeta)&:=\partial_{\zeta}w(t,\zeta)=\partial_{\zeta}\hrho(t,\zeta) .\label{e:wi}
\end{align}

In terms of these variables, we have
\begin{lemma}\label{t:pref1}
	If $(\hrho,\nu)$ solves \eqref{e:main1}--\eqref{e:main2} and $(w,w_0,w_\zeta)$ is defined by \eqref{e:ww}--\eqref{e:wi}, then \eqref{e:main1} becomes
	\begin{align}
			 \partial_t w_0 + 2\mathcal{g}^{0\zeta}\partial_{\zeta} w_0-  \mathcal{g}^{\zeta\zeta}\partial_{\zeta} w_\zeta     = & \Bigl(\frac{8}{3}-\kappa\Bigr)\frac{  f_0 w_0}{1+f}     -\frac{4}{3t}  w_0 -\frac{4 f_0^2 w}{3 (1+f)^2}  +  \frac{2w}{3t^2} + \frac{\kappa f_0^2  w}{(1+f)^2}  +\frac{4f w}{3t^2}  \notag  \\
			 &  + \frac{2 \omega (1 - 	\iota^3 )  (1+f)}{ 3 t^2}   w     +  \frac{  (8+5\omega)  (1 - 	\iota^3 )  (1+f) }{  9t^2}     w_\zeta  + Q_1 ,   \label{e:weq1}\\
	\partial_{t}w_\zeta=& \partial_{\zeta}w_0,  \label{e:weq2} \\
	\partial_{t}w =&w_{0} . \label{e:weq3}
\end{align}
where
\begin{align*}
	Q_1:=	&     - \frac{2 f_0^2}{9(1+f)^2}  \nu w_\zeta   +   \frac{ (\omega+1)   (2+\omega) (1 - 	\iota^3 ) }{9 t^2}  	\frac{ (1 +f+w)^{\omega} }{(1+f)^{\omega}} w_\zeta^2      + \frac{f_0^2}{9(1+f)^2} \nu^2  w_\zeta  \notag  \\
	&     + \frac{ 2  (1 - 	\iota^3 )   (1+f)}{9 t^2}  \underbrace{\Bigl[   \Bigl(1+ \frac{w}{1+f}\Bigr)^{1+\omega}  -1 \Bigr]  }_{=\mathrm{O}(w)}  w_\zeta  +   \frac{2 \iota^3 f  }{3t^{2} }  w_\zeta 	\Psi   + \frac{  4f_0^2 \nu^2 w_\zeta^2}{27 (1+f)^2 (1+f+w) }   \notag  \\
	&     + \frac{8  f_0  \nu  w_\zeta (f_0+w_0) }{9(1+f)(1+f+w) }   +  \frac{  (8+5\omega)  (1 - 	\iota^3 )  (1+f) }{  9t^2}    \underbrace{\Bigl[\Bigl(1+ \frac{w  }{1+f} \Bigr)^{\omega +1}  -1\Bigr] }_{=\mathrm{O}(w)}w_\zeta \notag \\
	&
	+ \frac{2}{3} (1+ f+w)   \Bigl( \frac{f_0w}{(1+f)(1+f+w)}  -\frac{w_0}{1+ f+w} -\frac{f_0}{3(1+f)}\frac{  \nu  w_\zeta }{1+ f+w} -\frac{f_0}{(1+f)} \nu\Bigr)^2  \notag \\
	& + \underbrace{\frac{2 (1 - 	\iota^3 )  (1+f)^2}{ 3 t^2}  \Bigl[  \Bigl(1+ \frac{w }{1+f} \Bigr)^{\omega+2} \Bigl(1- \Bigl(1+\frac{ w }{1+f}\Bigr)^{-\omega}\Bigr)-\omega \frac{w}{1+f}\Bigr]  }_{=\mathrm{O}(w^2)}  \notag\\	
	& +\frac{2}{3t^2}w^2	+\frac{4}{3(1+w+f) } w^2_0 - \frac{ 8 f_0 w_0w}{3(1+f+w)(1+f)} + \frac{4 f_0^2 w^2}{3 (1+f+w)(1+f)^2 }  .
\end{align*}
\end{lemma}
\begin{proof}
	The proof is direct computations from \eqref{e:main1} by
	substituting these new variables \eqref{e:ww}--\eqref{e:wi} into \eqref{e:main1} and using the equation \eqref{e:feq1b} of $f$. Then with direct calculations from \eqref{e:w0} and \eqref{e:wi}, we obtain \eqref{e:weq2}--\eqref{e:weq3}.
\end{proof}

Next, as \cite{Liu2022b}, we introduce
\begin{equation}\label{e:u}
	u(t,\zeta)=   \frac{ 1 }{f(t)} w (t,\zeta)	,   \quad
	u_0(t,\zeta)=    \frac{1}{f_0(t)}  w_0(t,\zeta) \AND
	u_\zeta(t,\zeta)=  \frac{\cc}{1+f(t)} w_\zeta(t,\zeta).
\end{equation}
where $\cc$ is a constant to be determined.
Recall the notation conventions \eqref{e:udl} in \S\ref{s:AIN}. As a reminder, we use $\uu$ to denote $u$ in the compactified time coordinate $\tau$. We will frequently use this notation in the following. For instance, we shall denote:
\begin{equation}\label{e:udvar}
	\underline{u}(\tau,\zeta)=u(g^{-1}(-\tau),\zeta), \quad \uuo(\tau,\zeta)=  	u_0(g^{-1}(-\tau),\xi) \AND \underline{u_\zeta}(\tau,\zeta)=  u_\zeta(g^{-1}(-\tau),\zeta) .
\end{equation}

We define two useful quantities $\chi$ and $\xi$  (see Appendix \ref{t:refsol}, or \cite[\S$2.5$]{Liu2022b}, for their detailed properties) by
\begin{equation}\label{e:Gdef}
	\chi(t):=\frac{t^{\frac{2}{3}} f_0(t)}{(1+f(t))^{\frac{2}{3} } f(t) g^{\frac{2}{3A}}(t)} \overset{(\star)}{=} \frac{  g^{-\frac{4}{3A}}(t) t^{-\frac{2}{3}}}{B f(t) (1+f(t))^{-\frac{2}{3} }} >0 \AND 	\xi(t):=\frac{1}{ g(t) (1+f(t)) } .
\end{equation}
where $(\star)$ holds due to the relation \eqref{e:f0frl} given in the following (see \cite{Liu2022b} for details).
Then Lemma \ref{t:pref1} can be further transformed to the following pre-Fuchsian system.

\begin{lemma}\label{t:pref2}
		If $(\hrho,\nu)$ solves \eqref{e:main1}--\eqref{e:main2} and $(\underline{ u}_0,\underline{ u}_\zeta,\underline{ u})$ is defined by \eqref{e:udvar}.   In terms of the $\tau$-coordinate, 	\eqref{e:main1} becomes
 \begin{align}
	&   \partial_\tau \underline{u_0}   - \frac{  2 }{3AB\tau}     \underline{\chi} \underline{\nu} \partial_{\zeta}  \underline{u_0}   + \frac{1}{A\tau} 	\Bigl( \frac{ 1 }{ 36 \cc } +   \frac{ 1 }{ 36 \cc \uf}  + \cc\mathfrak{Z}_0 \Bigr) \partial_{\zeta}    \underline{u_\zeta}
	=  \frac{1}{A\tau}  \Bigl(   4\kappa-\frac{ 14 }{    3   } + \bigl(\kappa- \frac{4}{3} \bigr)B^{-1} \underline{\mathfrak{G}}    \Bigr)   \underline{u_0}  \notag  \\
	& \hspace{2cm} - \frac{  (8+5\omega)  (1 - 	\iota^3 )  }{  9   \cc A \tau }       \underline{u_\zeta } +  \frac{1}{A\tau}  \Bigl( 4       -    4 \kappa  + \bigl(\frac{4 }{3   }      -     \kappa \bigr) B^{-1} \underline{\mathfrak{G}}     -     \frac{2 \omega (1 - 	\iota^3 )  }{ 3   }  \Bigr)  \uu        \notag \\
	& \hspace{2cm}    - \frac{1}{ A   }    \underline{\xi} \Bigl(      4\kappa  -   \frac{ 14 }{3   }   +\bigl(\kappa     -   \frac{4   }{3   }  \bigr)  B^{-1}    \underline{\mathfrak{G}}\Bigr) \uu  + \frac{  (8+5\omega)  (1 - 	\iota^3 )  }{  9   \cc A  }  \underline{\xi} \Bigl(1+\frac{1}{\underline{f}}\Bigr)      \underline{u_\zeta }  +     \underline{ F_2 } , \label{e:stp6} \\
	& \frac{1}{ 4 +B^{-1} \underline{\mathfrak{G}} }	\Bigl( \frac{ 1 }{ 36 \cc^2 } +   \frac{ 1 }{ 36 \cc^2 \uf}  +  \mathfrak{Z}_0\Bigr)\partial_{\tau} \underline{u_\zeta}  +  \frac{1}{	A \tau   }	\Bigl( \frac{ 1 }{ 36 \cc } +   \frac{ 1 }{ 36 \cc \uf}  + \cc \mathfrak{Z}_0\Bigr)  \partial_{\zeta} \underline{u_0 } \notag \\
	& \hspace{2cm} =   \frac{1}{ A    \tau   } \Bigl( \frac{ 1 }{ 36 \cc^2 }   +  \mathfrak{Z}_0\Bigr) \underline{u_\zeta} - \frac{ 1 }{ 36 \cc^2 A  }  \underline{\xi} \Bigl(1+\frac{1}{\uf}\Bigr) \underline{u_\zeta} \label{e:uzetaeq} \\
	\intertext{and}
	&\partial_{\tau} \underline{u} =    -\frac{1 }{	A    \tau  }   (4 +B^{-1}\underline{\mathfrak{G}}) \underline{u_0 } +  \frac{1}{	A    \tau  } (4 +B^{-1}\underline{\mathfrak{G}})  \underline{u} \notag  \\
	& \hspace{2cm} +\frac{1 }{	A   }   \underline{\xi} \Bigl(1+\frac{1}{\uf}\Bigr)(4 +B^{-1}\underline{\mathfrak{G}}) \underline{u_0 } -  \frac{1}{	A    }  \underline{\xi}\Bigl(1+\frac{1}{\uf}\Bigr) (4 +B^{-1}\underline{\mathfrak{G}})  \underline{u}   , \label{e:s3stp1}
\end{align}
		where
\begin{align}
\mathfrak{Z}_0=& \mathfrak{Z}_0(\tau,\uu,\underline{\nu}):= \frac{ 1 }{ 36 \cc^2 } \Bigl(1+ \frac{ 1 }{    \underline{f} }  \Bigr) \underbrace{\Bigl[ \Bigl( 1+ \frac{     \underline{f}   \underline{u}}{ 1+   \underline{f}}  \Bigr)^{1+\omega}   -1\Bigr]}_{=\mathrm{O}(\uu)}      -   \frac{ (4  + B^{-1}    \underline{\mathfrak{G}})}{ 9  \cc^2 }    \underline{\nu}^2  \label{e:Z0}  \\
\intertext{and}
 F_2  :=   &      -\frac{2  \chi}{9 A  B  g}\cc^{-1}  \nu  u_\zeta +  \frac{ \chi  }{ 9 A B   g } \cc^{-1}  \nu^2  u_\zeta   +  \frac{ 5 (\omega+2)  (1 - 	\iota^3 )   (1+f)  }{9 \cc A f   g }  \underbrace{\Bigl[   \Bigl(1+ \frac{    fu }{1+f}\Bigr)^{1+\omega}  -1 \Bigr]  }_{=\mathrm{O}(u)}    u_\zeta  \notag \\
& +  \frac{ (\omega+1)   (2+\omega) (1 - 	\iota^3 ) }{9 A f   g }  	  \Bigl(1+  \frac{fu}{1+f} \Bigr)^{\omega}  \cc^{-2}(1+f) u_\zeta^2    \notag \\
& +     \frac{2 \iota^3   }{3 A   g }  \cc^{-1}  u_\zeta	\Psi   +  \frac{  4  \chi \cc^{-2}  \nu^2  u_\zeta^2}{27 A B  g  (1 +    \frac{f}{1+f} u) }  +   \frac{8  \chi  \cc^{-1}   \nu  u_\zeta   (1 +    u_0) }{9  A B     g (1 +   \frac{f}{1+f} u) }  \notag \\
&
+   \frac{  2 \chi }{ 3 AB  g }  \Bigl(1 +   \frac{ f u}{1+f} \Bigr)   \Bigl( \frac{   fu }{(1+f)(1+    \frac{fu}{1+f} )}  -\frac{   u_0}{1+   \frac{fu}{1+f} } -\frac{1}{3 }\frac{  \nu  \cc^{-1}  u_\zeta }{1 +    \frac{fu}{1+f} } -  \nu\Bigr)^2  \notag \\
& + \frac{2 (1 - 	\iota^3 )  (1+f) }{ 3   A f  g}  \underbrace{ \Bigl[  \Bigl(1+ \frac{    fu }{1+f} \Bigr)^{\omega+2} \Bigl(1- \Bigl(1+\frac{    fu }{1+f}\Bigr)^{-\omega}\Bigr)-\omega \frac{     fu }{1+f}\Bigr]  }_{=\mathrm{O}(u^2)}  \notag\\	
& +  \frac{2  }{ 3  A  (1+f )  g }  f u^2	+   \frac{4\chi }{3A B    g  (1+    \frac{fu}{1+f}) }
\Bigl(u_0-\frac{fu}{1+f}\Bigr)^2    .     \label{e:F2def.1}
	\end{align}
\end{lemma}

\begin{proof}
This proof relies on lengthy and tedious calculations which is similar to \cite[\S$3.2$]{Liu2022b}. We omit the details but point out a few crucial identities used in this proof.
Under this time transform \eqref{e:ttf}, using \eqref{e:dtg0}, we arrive at
\begin{equation}\label{e:ttf2}
	\del{\tau} \underline{ u}_\mu= -\underline{[g^\prime(t)]^{-1} \del{t} u_\mu}
	= \underline{ \frac{t^{\frac{2}{3}} (1+f )^{\frac{1}{3} } }{A  B   g^{\frac{2}{3A}+1}  f} \del{t} u_\mu } , \quad (\mu=0,\zeta)  .
\end{equation}
In the calculations, note we can use the ODE \eqref{e:feq1b} of $f$ to replace $f^{\prime\prime}$ and use the identity from Lemma \ref{t:f0fg} in Appendix \ref{t:ttf},
\begin{equation}\label{e:f0frl}
	f_0(t)=B^{-1} t^{-\frac{4}{3}} g^{-\frac{2}{3A}}(t)(1+f(t))^{\frac{4}{3} } >0
\end{equation}
to replace $f_0$ properly.
Then, with the help of Proposition \ref{t:limG} and the definition \eqref{e:xidef} of $\xi$, there is a function $\mathfrak{G} \in C^1([t_0,t_m))$, such that for $t\in [t_0,t_m)$, we have identities
\begin{equation}\label{e:chi4}
	\frac{	\chi(t)}{B}= 4  + \frac{ \mathfrak{G}(t) }{B} >0 \AND 	\frac{1}{fg}
	=\xi\Bigl(1+\frac{1}{f}\Bigr).
\end{equation}
Using the above identities, we are able to conclude this lemma by following similar computations in \cite[\S$3.2$]{Liu2022b}.
\end{proof}

The next lemma gives the equation of the velocity $\nu$ in terms of the compactified time $\tau$.
\begin{lemma}\label{t:pref3}
		If $(\hrho,\nu)$ solves \eqref{e:main1}--\eqref{e:main2}, then in terms of the $\tau$-coordinate, \eqref{e:main2} becomes
	\begin{align}\label{e:veq3}
	&     \del{\tau} \underline{ \nu }
	- \frac{1}{ 3 A  \tau }  (4  + B^{-1} \underline{ \mathfrak{G}}) \underline{ \nu  }	\del{\zeta} \underline{\nu} 	    =     \frac{1}{A \tau } \Bigl( 4\kappa- \frac{2   }{3   }  +\bigl(\kappa -  \frac{1 }{ 3    } \bigr)  B^{-1} \underline{ \mathfrak{G}}  \Bigr) \underline{  \nu}      + \frac{ 1}{3 A \tau }( 4  + B^{-1} \underline{ \mathfrak{G}})\underline{ \nu}^2   + \frac{2 \iota^3   }{ A \tau } \underline{ \Psi}  \notag \\
	&  \hspace{0.5cm}    - \frac{(2+\omega)  (1 - 	\iota^3 )  }{3 \cc A       }  \underline{ \xi}\Bigl(1+\frac{1}{\uf}\Bigr)   \Bigl(1+\frac{\uf \uu}{1+\uf}\Bigr)^{\omega}  \underline{ u_\zeta }    +  \frac{(2+\omega)  (1 - 	\iota^3 ) }{3 \cc A  \tau }     \Bigl(1+\frac{\uf\uu}{1+\uf}\Bigr)^{\omega}  \underline{ u_\zeta} \notag \\
	&\hspace{0.5cm}  +  \frac{ 2     (1 - 	\iota^3 )   (1+\uf)}{3  A   \tau  \uf  }   \underbrace{\Bigl[   \Bigl(1+\frac{\uf\uu}{1+\uf}\Bigr)^{1+\omega}  -1 -\frac{(1+\omega)\uf \uu}{1+\uf}\Bigr] }_{\mathrm{O}(\uu^2)}   + \frac{ 2  (1+\omega)   (1 - 	\iota^3 ) \uu }{3  A \tau  }      .
\end{align}
\end{lemma}
\begin{proof}
	Multiplying \eqref{e:main2} by $  \frac{t^{\frac{2}{3}} (1+f )^{\frac{1}{3} }}{A B g^{\frac{2}{3A}+1} f}  $, we obtain
	\begin{equation*}
  \frac{t^{\frac{2}{3}} (1+f )^{\frac{1}{3} }}{A B g^{\frac{2}{3A}+1} f}  \del{t} \nu
		+  \frac{t^{\frac{2}{3}} (1+f )^{\frac{1}{3} }}{A B g^{\frac{2}{3A}+1} f} \frac{f_0}{3 (1+f)} \nu    	\del{\zeta} \nu    	    =  \frac{t^{\frac{2}{3}} (1+f )^{\frac{1}{3} }}{A B g^{\frac{2}{3A}+1} f} G_1  .
	\end{equation*}
By using similar identities \eqref{e:ttf2} for $\nu$, \eqref{e:f0frl}  and the definition of $\chi(t)$ \eqref{e:Gdef}, we obtain
\begin{equation*}
  \del{\tau} \underline{ \nu}
	- \frac{ \underline{ \chi } }{ 3 A B \tau }  \underline{ \nu }	\del{\zeta} \underline{ \nu } 	    = \underline{  \frac{t^{\frac{2}{3}} (1+f )^{\frac{1}{3} }}{A B g^{\frac{2}{3A}+1} f} G_1}  .
\end{equation*}
Noting by \eqref{e:Psi0}, \eqref{e:ww} and \eqref{e:u}, we express $\Psi$ in term of $u$,
\begin{equation}\label{e:ps2}
	\Psi(t,\zeta)=\frac{1}{e^{3\zeta} }\int^\zeta_{-\infty} u(t,z)e^{3z} dz .
\end{equation}
Substituting $G_1$ by \eqref{e:G1} into the above equation with direct calculations
yields \eqref{e:veq3}, which completes the proof.
\end{proof}

\subsection{Reformulation of continuity equation and evolutional gravity term $\Psi$}\label{s:evo}
In $F_1$ of \eqref{e:stp6} and \eqref{e:veq3}, there exist singular terms like $\Psi/\tau$ in the remainders which cannot be easily eliminated. To obtain a complete Fuchsian system, we propose to \textit{treat $\Psi$ as a new Fuchsian field}. Accordingly, we need to construct a singular evolution equation for $\Psi$ based on the continuity equation \eqref{e:keyid5} and supplement it to the Fuchsian system. As a crucial bridge connecting the velocity, the density, and their derivatives, we first rewrite the continuity equation \eqref{e:keyid5} to prepare for this task.
\begin{lemma}\label{t:keyid3}
	The continuity equation \eqref{e:keyid5}, in terms of $(u_0,u_\zeta,u)$ becomes
	\begin{equation}\label{e:keyid2}
		\frac{ 3   fu}{ 1+f+ fu} -  \frac{ 3  u_0}{1 + \frac{fu}{1+f}} - \frac{  \cc^{-1}   \nu   u_\zeta }{1+ \frac{f u }{1+f} }  =  \del{\zeta}\nu	+3 \nu   .
	\end{equation}
Moreover, we also use the following two expressions of \eqref{e:keyid2} in the following,
\begin{equation}\label{e:keyid3a}
	\del{\zeta}\nu	
	= 	-3u_0+  3u -3 \nu      -   \frac{3u }{1+f}  + 3 \Bigl(u_0-\frac{fu}{1+f}\Bigr)\Bigl(1-  \frac{ 1}{1 + \frac{fu}{1+f}} \Bigr)- \frac{    \cc^{-1} \nu   u_\zeta }{1+  \frac{fu}{1+f}}
\end{equation}
and
\begin{equation}\label{e:keyid3b}
	\frac{1}{3} \del{\zeta} \Bigl[\bigl(1+\frac{f u}{1+f}\bigr)\nu\Bigr]+\Bigl(1+\frac{fu}{1+f}\Bigr)\nu   =\frac{fu}{1+f}-u_0.
\end{equation}
\end{lemma}
\begin{proof}
	By inserting \eqref{e:ww}--\eqref{e:wi} and \eqref{e:u} into \eqref{e:keyid5}, direct calculations yield \eqref{e:keyid2} and \eqref{e:keyid3a}. Let us prove \eqref{e:keyid3b}. Using \eqref{e:ww}--\eqref{e:wi} and \eqref{e:u}, we obtain
	\begin{equation}\label{e:uzuz}
		u_\zeta=\frac{\cc f}{1+f}\del{\zeta} u.
	\end{equation}
Inserting \eqref{e:uzuz} into \eqref{e:keyid2} implies \eqref{e:keyid3b}, which completes the proof.
\end{proof}

\begin{lemma}\label{t:pseq}
	Suppose $\Psi$ is defined by \eqref{e:Psi0} and further given by \eqref{e:ps2} in terms of $u$ and $u$ is given by \eqref{e:u}. Then $\Psi$ satisfies an equation,
	\begin{equation}\label{e:dtpsi6}
		\del{\tau} \underline{ \Psi  }
		=   \frac{ 1 }{	 3 A    \tau     } 	 ( 4 +B^{-1}  \underline{\mathfrak{G}})    \Bigl( 1+  \frac{\uf  \underline{u}}{1+\uf} \Bigr)  \underline{\nu}  -  \frac{ \underline{ \chi } \underline{\xi} }{	 3 A  B  } 	    \Bigl(1+\frac{1}{\uf}\Bigr)   \Bigl( 1+   \frac{\uf\uu}{1+\uf} \Bigr) \underline{ \nu}    -   \frac{ \underline{\chi} \underline{\xi} }{	A  B   }	  \Bigl(1+\frac{1}{\uf}\Bigr)   \underline{ \Psi } .
	\end{equation}
Moreover, we also have a useful equation
\begin{equation}\label{e:dipsi}
		\del{\zeta}  \Psi
		=u-3 \Psi   .
\end{equation}
\end{lemma}
\begin{proof}
	First, with the help of \eqref{e:u}, we have
	\begin{equation}\label{e:dtu00}
		\partial_{t} u =     \frac{f_0}{f} (u_0 -   u) .
	\end{equation}
	Differentiating $e^{3\zeta}\Psi$ with respect to $t$ and using \eqref{e:dtu00}, we obtain
	\begin{equation}\label{e:dteps}
		\del{t} (e^{3\zeta}\Psi)=\int^\zeta_{-\infty} \del{t} u(t,z) e^{3z} dz
		=  \int^\zeta_{-\infty}  \frac{f_0(t) }{f(t)}\bigl(u_0(t,z)-u(t,z)\bigr)   e^{3z} dz .
	\end{equation}
Inserting \eqref{e:keyid3b} into \eqref{e:dteps}, we obtain
\begin{align*}
		\del{t} (e^{3\zeta}\Psi)	= &-   \int^\zeta_{-\infty}   \frac{f_0  u }{f (1+f) }   e^{3z} dz -\frac{ f_0 }{3 f } \int^\zeta_{-\infty}  3 \Bigl(1+  \frac{fu}{1+f}\Bigr) \nu  e^{3z}  +  \del{z}\Bigl[ \Bigl(1+  \frac{fu}{1+f} \Bigr) \nu   \Bigr]  e^{3z} dz \notag  \\
		= &-   \int^\zeta_{-\infty}   \frac{f_0  u }{f (1+f) }   e^{3z} dz -\frac{ f_0 }{3 f }  \int^\zeta_{-\infty}   \del{z}\Bigl[ \Bigl(1+  \frac{fu}{1+f} \Bigr) \nu     e^{3z}\Bigr] dz \notag  \\
		= &-   \frac{f_0   }{f (1+f) }   \int^\zeta_{-\infty}  u (t,z) e^{3z} dz -\frac{ f_0 }{3 f }  \Bigl( 1+  \frac{fu}{1+f} \Bigr) \nu     e^{3 \zeta}\notag  \\
		= &-   \frac{f_0   }{f (1+f) }   e^{3\zeta} \Psi(t,\xi) -\frac{ f_0 }{3 f }  \Bigl( 1+  \frac{fu}{1+f} \Bigr) \nu     e^{3 \zeta}  .
\end{align*}
Then this implies
\begin{equation*}
	\del{t} \Psi
	=  -   \frac{f_0   }{f (1+f) } \Psi  -\frac{ f_0 }{3 f }  \Bigl( 1+  \frac{fu}{1+f} \Bigr) \nu    .
\end{equation*}

Transforming the time derivative $\del{t}$ to be $\del{\tau}$ which is similar to \eqref{e:ttf2}, we have
\begin{equation*}
	\frac{t^{\frac{2}{3}} (1+f)^{\frac{1}{3} } }{	A  B g^{\frac{2}{3A}+1}    f   }		\del{t} \Psi
	=  -   \frac{t^{\frac{2}{3}} (1+f)^{\frac{1}{3} } }{	A  B g^{\frac{2}{3A}+1}    f   }	 \frac{f_0   }{f (1+f) }  \Psi  -   \frac{t^{\frac{2}{3}} (1+f)^{\frac{1}{3} } }{	A  B g^{\frac{2}{3A}+1}    f   }	\frac{ f_0 }{3 f }  \Bigl( 1+  \frac{fu}{1+f} \Bigr) \nu,
\end{equation*}
which leads to, with the help of \eqref{e:ttf},  \eqref{e:f0frl}  and the definition \eqref{e:Gdef} of $\chi(t)$,
\begin{equation*}
	\del{\tau}  \underline{\Psi}
	=  \underline{-   \frac{\chi}{	A  B g   f   }	    \Psi   -   \frac{ \chi  }{	 3 A  B g    } 	  \frac{1+f}{f}  \Bigl( 1+  \frac{fu}{1+f} \Bigr) \nu    }  ,
\end{equation*}
which further yields \eqref{e:dtpsi6} by using \eqref{e:chi4}.  Then  differentiating  \eqref{e:ps2} with respect to $\zeta$ concludes \eqref{e:dipsi}. We finish the proof.
\end{proof}

\subsection{The Fuchsian system}
It turns out the equations in Lemma \ref{t:pref2}--\ref{t:pref3} and Lemma \ref{t:pseq} can not directly form a Fuchsian system given in Appendix
\ref{s:fuc}.
We first need to use identities \eqref{e:keyid2} and \eqref{e:dipsi} to rewrite \eqref{e:veq3} and \eqref{e:dtpsi6}, and this is given by the following lemma.
\begin{lemma}\label{t:vph2}
	If $(\hrho,\nu)$ solves \eqref{e:main1}--\eqref{e:main2}, then \eqref{e:veq3} and \eqref{e:dtpsi6} can be rewritten as
	\begin{align}\label{e:veq5}
	   \del{\tau} \underline{ \nu }
		=   	&     \frac{ 2  (1+\omega)   (1 - 	\iota^3 ) \uu }{3  A \tau  }  +    \frac{1}{A \tau }  \Bigl( 4\kappa- \frac{2   }{3   }  +\bigl(\kappa -  \frac{1 }{ 3    } \bigr)  B^{-1} \underline{ \mathfrak{G}}  \Bigr) \underline{  \nu}     + \frac{2 \iota^3   }{ A \tau } \underline{ \Psi} +  \frac{(2+\omega)  (1 - 	\iota^3 ) }{3 \cc A  \tau }      \underline{ u_\zeta} \notag \\
		&      - \frac{(2+\omega)  (1 - 	\iota^3 )  }{3 \cc A       }  \underline{ \xi}\Bigl(1+\frac{1}{\uf}\Bigr)   \Bigl(1+\frac{\uf \uu}{1+\uf}\Bigr)^{\omega}  \underline{ u_\zeta }    +  \frac{(2+\omega)  (1 - 	\iota^3 ) }{3 \cc A  \tau }    \Bigl[ \Bigl(1+\frac{\uf\uu}{1+\uf}\Bigr)^{\omega}  - 1 \Bigr]\underline{ u_\zeta} \notag \\
		&   +  \frac{ 2     (1 - 	\iota^3 )   (1+\uf)}{3  A   \tau  \uf  }   \underbrace{\Bigl[   \Bigl(1+\frac{\uf\uu}{1+\uf}\Bigr)^{1+\omega}  -1 -\frac{(1+\omega)fu}{1+f}\Bigr] }_{\mathrm{O}(\uu^2)}      + \frac{ 1}{3 A \tau }( 4  + B^{-1} \underline{ \mathfrak{G}})\underline{ \nu}^2  \notag \\
		&   + \frac{1}{ 3 A  \tau }  (4  + B^{-1} \underline{ \mathfrak{G}}) \underline{ \nu  }	\Bigl(\frac{ 3   fu}{ 1+f+ fu} -  \frac{ 3  u_0}{1 + \frac{fu}{1+f}} - \frac{  \cc^{-1}   \nu   u_\zeta }{1+ \frac{f u }{1+f} } -3 \nu \Bigr)
	\end{align}
	and
	\begin{align}\label{e:dtpsi7}
		\del{\tau} \underline{\Psi }
		-\frac{\alpha}{AB\tau}	\del{\zeta} \underline{ \Psi} = &  -\frac{\alpha}{AB\tau}\underline{u }+ \frac{3 \alpha}{AB\tau} \underline{\Psi} + \frac{ 1 }{	 3 A    \tau     } 	 ( 4 +B^{-1} \underline{\mathfrak{G}})    \Bigl( 1+ \frac{\uf \underline{u}}{1+\uf} \Bigr) \underline{\nu}  \notag  \\
		& -  \frac{ \underline{\chi}  }{	 3 A  B  } 	 \underline{ \xi} \Bigl(1+\frac{1}{\uf}\Bigr)  \Bigl( 1+  \frac{\uf \underline{u}}{1+\uf} \Bigr) \underline{\nu }   -   \frac{\underline{\chi}}{	A  B   }	\underline{\xi} \Bigl(1+\frac{1}{\uf}\Bigr)    \underline{\Psi }.
	\end{align}
\end{lemma}
\begin{proof}
	By direct calculations, \eqref{e:veq3}, \eqref{e:keyid2} and \eqref{e:keyid3a} yield \eqref{e:veq5}, while  \eqref{e:dtpsi6} and \eqref{e:dipsi} imply \eqref{e:dtpsi7}.
\end{proof}

Now we are in the position to give the complete Fuchsian system. Letting the Fuchsian field $\U:=(\uuo, \underline{u_\zeta},  \uu, \underline{\nu}, \underline{\Psi} )^T$ and selecting the parameters
\begin{equation}\label{e:para0}
\kappa:=\frac{7}{6}+ \lambda ,  \quad	\cc:= \frac{1}{5}, \quad q:=\lambda+\frac{1}{30} (3-8 \iota^3) \;
 \Bigl(>\frac{7}{150}\Bigr)  \AND  	 \alpha =\frac{(3 \iota^3+2)^2 B}{6 (10 \lambda+\iota^3+9)}  .
\end{equation}
for any constant $\lambda>0$ and $\iota^3\in (0, 1/5]$, with the help of $\omega=-8/5$ (recalling \eqref{e:S1}),
the equations \eqref{e:stp6}--\eqref{e:s3stp1} and \eqref{e:veq5}--\eqref{e:dtpsi7}  become
\begin{equation}\label{e:fuc}
	\B^0\partial_{\tau}\U+\B^\zeta\partial_{\zeta}\U=\frac{1}{ \tau}\mathfrak{B} \Pbb \U+ \mathcal{H} +  (-\tau)^{-\frac{1}{2}} \mathcal{F}  ,
\end{equation}
where $\Pbb=\mathbb{I}  $,
\begin{gather}
	\B^0:=\p{1 & 0 &0 &  0 & 0\\
		0 &  \frac{1}{ 4 +B^{-1} \underline{\mathfrak{G}} }	\bigl( \frac{ 25 }{ 36   } +   \frac{ 25 }{ 36   \uf}  +  \mathfrak{Z}_0 \bigr)  & 0 & 0 & 0 \\
		0 & 0 & \lambda+\frac{1}{30} (3-8 \iota^3) & 0 & 0\\
		0 & 0 & 0 & 1 & 0  \\
		0 & 0 & 0 & 0 & 1}, \label{e:B0}\\
	\B^\zeta := \frac{ 1 }{A \tau} \p{ -\frac{2  \underline{\chi}    }{3  B }  \underline{\nu}      &  \frac{ 5 }{ 36  } +   \frac{ 5 }{ 36   \uf}  + \frac{1}{5} \mathfrak{Z}_0  & 0 & 0   & 0\\
\frac{ 5 }{ 36  } +   \frac{ 5 }{ 36   \uf}  + \frac{1}{5} \mathfrak{Z}_0  & 0 &  0  & 0 & 0\\
		0 & 0 & 0 & 0 & 0 \\
		0  & 0 & 0 & 0  & 0 \\
		0 & 0 & 0 & 0 & -\frac{(3 \iota^3+2)^2 }{6 (10 \lambda+\iota^3+9)}	} ,  \\
	\mathfrak{B}:=
	\frac{1}{A}   \p{    4\lambda  + \mathfrak{Z}_1    &  \mathfrak{Z}_2      & \frac{2(3-8 \iota^3)  }{15}    -4  \lambda   +\mathfrak{Z}_3  &  \mathfrak{Z}_4    & 0 \\
	0 &    \frac{25}{ 36 	  }  +  \mathfrak{Z}_0         &    0 &      0   & 0 \\
 -\frac{2(3-8 \iota^3) }{15}	-4 \lambda   & 0 &   4 \lambda+\frac{2(3-8 \iota^3)    }{15}     & 0 & 0  \\
	0	    &    \frac{2 (1-\iota^3) }{3}    &       -\frac{2(1-\iota^3) }{5}  +\mathfrak{Z}_5    &    4\lambda+4    +\mathfrak{Z}_6  &  2\iota^3  \\
	0 &  0  & -\frac{(3 \iota^3 +2)^2}{6 (10 \lambda+\iota^3+9)}  &  \frac{4}{3}   	      +\mathfrak{Z}_7  & \frac{(3 \iota^3+2)^2}{2 (10 \lambda+\iota^3+9)}  },  \label{e:frakB}\\
\mathcal{F}:= \p{ - \frac{1}{A B }	\bigl(\lambda- \frac{1}{6} \bigr)     (-\tau)^{- \frac{1}{2}}\underline{\mathfrak{G}}   (\underline{u_0} - \uu) \\ 0 \\ -\frac{1}{4 A B} \bigl(-\frac{2(3-8 \iota^3) }{15}	-4 \lambda\bigr)   (-\tau)^{- \frac{1}{2}}\underline{\mathfrak{G}} ( \underline{u_0} -\uu )  \\ 	
	- \frac{1 }{A B } \bigl(\lambda +\frac{5}{6} \bigr)   	(-\tau)^{- \frac{1}{2}}\underline{\mathfrak{G}} \underline{\nu}  \\
-\frac{1}{3 A B } (-\tau)^{- \frac{1}{2}}\underline{\mathfrak{G}} \underline{\nu} } \label{e:calF}
\end{gather}
and $
\mathcal{H}=\mathcal{H}(\tau,\U):= \p{ H_1,H_2, H_3, H_4, H_5}^T$ such that
\begin{align}
	H_1  =H_1(\tau,\U)
	=&    - \frac{1}{ A   }    \underline{\xi} \Bigl[     4\lambda   +\Bigl(\lambda   -   \frac{1 }{6 }  \Bigr)  B^{-1}    \underline{\mathfrak{G}}\Bigr] \uu  ,  \label{e:H1} \\
	H_2=H_2(\tau,\U)= & - \frac{ 25 }{ 36 A  }  \underline{\xi} \Bigl(1+\frac{1}{\uf}\Bigr) \underline{u_\zeta} , \\
	H_3=H_3(\tau,\U)= &   \Bigl(\lambda+\frac{3-8 \iota^3}{30} \Bigr)  \frac{ 1}{	A   }      \underline{\xi} \Bigl(1+\frac{1}{\uf}\Bigr)(4 +B^{-1}   \underline{\mathfrak{G}})    (\underline{u_0 }  - \uu ) , \\
	H_4 =H_4(\tau,\U)
	=&       - \frac{ 2  (1 - 	\iota^3 )  }{3   A       }  \underline{ \xi}\Bigl(1+\frac{1}{\uf}\Bigr)   \Bigl(1+\frac{\uf \uu}{1+\uf}\Bigr)^{-\frac{8}{5}}  \underline{ u_\zeta }  , \\
	H_5=H_5(\tau,\U)= &    -  \frac{   \underline{\chi } \underline{\xi}}{	 3 A  B  } 	   \Bigl(1+\frac{1}{\uf}\Bigr)   \Bigl( 1+   \frac{\uf\uu}{1+\uf} \Bigr)   \underline{\nu}     -   \frac{  \underline{\chi}  \underline{\xi}}{	A  B   }	 \Bigl(1+\frac{1}{\uf}\Bigr)      \underline{ \Psi }   .   \label{e:H5}
\end{align}
and $\mathfrak{Z}_\ell:=\mathfrak{Z}_\ell(\tau,\mathcal{U})$ ($\ell=0,1,\cdots,7$) are analytic for $\mathcal{U}\in B_{\tilde{R}}(\Rbb^N)$ (denote $N:=5$) and continuous  for $\tau\in[-1,0]$ and satisfy $\mathfrak{Z}_\ell(\tau,0)=0$.

\subsection{Verifications of Fuchsian system}\label{s:stp4}
This section contributes to verifying the singular system \eqref{e:fuc} satisfies all the conditions \ref{c:2}--\ref{c:7} in Appendix \ref{s:fuc} which implies \eqref{e:fuc} is a Fuchsian system \eqref{e:model1}.

\subsubsection{Verifications of Conditions \ref{c:2} and \ref{c:6} in Appendix \ref{s:fuc}} \label{s:F135}
Since $\Pbb=\mathbb{I}$ leads to $\Pbb^\perp=0$ (in fact, all the conditions involving $\Pbb^\perp$ hold immediately), \ref{c:2} and \ref{c:6} are verified directly.

\subsubsection{Partial verification of Condition  \ref{c:4}  in Appendix \ref{s:fuc}} \label{s:F3}
Now let us turn to Condition \ref{c:4}. It is clear that $\B^0$ and $\B^\zeta$ are symmetric and $[\mathbb{I},\mathfrak{B}]=0$.
From the expressions of  $\B^\zeta$ and $\mathfrak{B}$, we know
$\B^\zeta \in   C^{0}([-1,0),C^{\infty}( B_R(\mathbb R^{N}), \mathbb M_{N\times N})$,  $\mathfrak{B}\in   C^{0}([-1,0],C^{\infty}( B_R(\mathbb R^{N}), \mathbb M_{N\times N})$.  Corresponding to \eqref{e:model1}, we identify $\B^\zeta_0=0$ and $\B^\zeta_2=\tau \B^\zeta$ and it is direct to check
\begin{equation*}
	\tau \B^\zeta \in C^{0}([-1,0],C^{\infty}( B_R(\mathbb R^{N}), \mathbb M_{N\times N})) .
\end{equation*}
Correspondingly, we can construct $\tilde{\B}^0(\tau)$ and $\tilde{\mathfrak{B}}(\tau)$ by dropping all the $\mathfrak{Z}_\ell$ ($\ell=0,\cdots 7$) terms, while construct $\tilde{\B}^\zeta_2(\tau)$  by dropping all the $\mathfrak{Z}_\ell$ ($\ell=0,\cdots 7$) terms and $\underline{\nu}$ involved term $-\frac{2\underline{\chi}}{3AB}  \underline{\nu}$.  We can see  $\tilde{\B}^0,\;\tilde{\B}^\zeta_2, \;\tilde{\mathfrak{B}} \in  C^0([-1,0], C^\infty(\mathbb M_{N\times N}))$.  It is direct, from the expression \eqref{e:B0} of $\B^0$, to verify $\B^0\in C^{0}([-1,0),C^{\infty}( B_R(\mathbb R^{N}), \mathbb M_{N\times N})$ instead of $\B^0\in C^{1}([-1,0),C^{\infty}( B_R(\mathbb R^{N}), \mathbb M_{N\times N})$ and we postpone this complete verification to \S\ref{s:F6}.
We now partially verified \ref{c:4}.

\subsubsection{Verification of Condition \ref{c:3} in Appendix \ref{s:fuc}}\label{s:F2}
Let us verify \ref{c:3}. By the definition \eqref{e:H1}--\eqref{e:H5} of $H_\ell$ ($\ell=1,\cdots,5$) and  Proposition \ref{t:limG} and \ref{t:fginv0}, we have $\mathcal{H}(\tau,0)=0$ and $\mathcal{H}\in C^0([-1,0],C^\infty(B_{\tilde{R}}(\Rbb^N),\mathbb{M}_{N\times N}))$, which verifies \ref{c:3}. In order to verify $\mathcal{F}\in C^0([-1,0],C^\infty(B_{\tilde{R}}(\Rbb^N),\mathbb{M}_{N\times N}))$ (recall \eqref{e:calF}), we need preparations:

\begin{lemma}\label{t:Geq2}
	Suppose $\mathfrak{G}$ is given by \eqref{e:chi4}, then $	\del{\tau} \underline{\mathfrak{G}}$ is
	\begin{equation}\label{e:dtgeq}
		\del{\tau} \underline{\mathfrak{G}}= \frac{  \underline{\mathfrak{G}}  \underline{\chi} }{3  A  B   \tau   }      - \frac{ \underline{\chi}^2 }{A  B   \underline{g  f}    	 }-   \frac{2 \underline{\chi}^{\frac{3}{2}}  }{ 3A    B^{\frac{1}{2}}   \underline{g } \underline{f}^{\frac{1}{2}}   }
	\end{equation}
for $\tau\in(-1,0)$.
\end{lemma}
\begin{proof}
	First by the definition of $\chi$
	\eqref{e:Gdef}, we express
	\begin{equation}\label{e:gbA}
		g^{\frac{2}{3A}}   = B^{-\frac{1}{2}} t^{-\frac{1}{3} }f^{-\frac{1}{2}}  (1+f )^{\frac{1}{3} }	\chi^{-\frac{1}{2}} .
	\end{equation}
Multiplying \eqref{e:dtchi} in Lemma \ref{t:dtchi} of Appendix \ref{s:dtchi} on the both sides (letting $\cc=\frac{4}{3}$ and $\ca=\frac{4}{3}$) by $\frac{t^{\frac{2}{3}} (1+f )^{\frac{1}{3} } }{A  B   g^{\frac{2}{3A}+1}  f}$ and using \eqref{e:gbA} to replace $g^{\frac{2}{3A}}$, we obtain
\begin{equation*}
	\frac{t^{\frac{2}{3}} (1+f )^{\frac{1}{3} } }{A  B   g^{\frac{2}{3A}+1}  f}	\del{t}\mathfrak{G}
	= -\frac{  \mathfrak{G}  \chi }{3  A  B   g   }      - \frac{ \chi^2 }{A  B   g  f    	 }-   \frac{2 \chi^{\frac{3}{2}}  }{ 3A    B^{\frac{1}{2}}   g  f^{\frac{1}{2}}   } .
\end{equation*}
Similar to \eqref{e:ttf2}, we arrive at
\begin{equation*}
		\del{\tau} \underline{\mathfrak{G}}=  \frac{  \underline{\mathfrak{G}}  \underline{\chi} }{3  A  B   \tau   }      - \frac{ \underline{\chi}^2 }{A  B   \underline{g  f}    	 }-   \frac{2 \underline{\chi}^{\frac{3}{2}}  }{ 3A    B^{\frac{1}{2}}   \underline{g } \underline{f}^{\frac{1}{2}}   } .
\end{equation*}
It completes the proof.
\end{proof}

\begin{lemma}\label{t:Gest2}
Suppose $\mathfrak{G}$ is given by \eqref{e:chi4}.
Then $\underline{\mathfrak{G}} (\tau )$ has an estimate
\begin{equation*}
|\underline{\mathfrak{G}}(\tau) | \lesssim (-\tau)^{\frac{1}{2}}
\end{equation*}
for $\tau\in[-1,0)$ and the function $(-\tau)^{-\frac{1}{2}}\underline{\mathfrak{G}}(\tau)$ can be continuously extended to $\tau\in[-1,0]$.
\end{lemma}
\begin{proof}
	Multiplying $2 \underline{\mathfrak{G}}$ on the both sides of \eqref{e:dtgeq} yields
	\begin{equation*}	  \del{\tau} \underline{\mathfrak{G}}^2 = 2 \underline{\mathfrak{G}}\del{\tau} \underline{\mathfrak{G}}=  \frac{  2 \underline{\mathfrak{G}}^2  \underline{\chi} }{3  A  B   \tau   }      - \frac{ 2 \underline{\mathfrak{G}} \underline{\chi}^2 }{A  B   \underline{g  f}    	 }-  \frac{ 4 \underline{\mathfrak{G}}  \underline{\chi}^{\frac{3}{2}}  }{ 3A    B^{\frac{1}{2}}   \underline{g } \underline{f}^{\frac{1}{2}}   } .
	\end{equation*}
By Proposition \ref{t:limG}, \ref{t:fginv0} and Corollary \ref{s:gf1/2}, noting $\tau\in[-1,0)$, there is a time $\tau_1\in(-1,0)$, such that for $\tau\in (\tau_1,0]$,  $\underline{\chi}\in(\frac{7}{2}B,\frac{9}{2}B)$, $|\underline{\mathfrak{G}}| < \frac{1}{2}B$,  $|1/(\underline{g}\underline{f})| < 1$ and $|1/(\underline{g}\underline{f}^{\frac{1}{2}})| < 1$. Furthermore,
there is a constant $C_0>0$, such that for $\tau\in(\tau_1,0)$, noting $A\in(0,2)$, then
	\begin{equation*}	
		\del{\tau} \underline{\mathfrak{G}}^2
		\leq   \frac{ 7 }{ 6 \tau} \underline{\mathfrak{G}}^2   + C_0 .
\end{equation*}
By the Gronwall's inequality (by multiplying $ (-\tau)^{-\frac{ 7 }{ 6}}$), we derive
\begin{equation*}
	\del{\tau} \bigl(  (-\tau)^{-\frac{ 7 }{ 6} } \underline{\mathfrak{G}}^2\bigr) \leq C_0 (-\tau)^{-\frac{ 7 }{ 6}} ,
\end{equation*}
which yields
\begin{equation*}
	\underline{\mathfrak{G}}^2 (\tau) \leq (-\tau_1)^{-\frac{7}{6}} \underline{\mathfrak{G}}^2(\tau_1) (-\tau)^{\frac{ 7 }{ 6} } +6C_0(-\tau) - 6C_0(-\tau_1)^{-\frac{1}{ 6}}(-\tau)^{\frac{ 7 }{ 6}}  \lesssim -\tau .
\end{equation*}
for $\tau\in (\tau_1,0)$. For $\tau\in [-1,\tau_1]$, due to the continuity of $\underline{\mathfrak{G}}$ yielding $\underline{\mathfrak{G}}$ is bounded, there is a constant $C>0$, such that $|\underline{\mathfrak{G}} (\tau)| \leq C(-\tau_1)^{	\frac{1}{2}} \leq C(-\tau)^{	\frac{1}{2}} $ for $\tau\in [-1,\tau_1]$.
This implies $
	|\underline{\mathfrak{G}}(\tau) | \lesssim (-\tau)^{\frac{1}{2}}  $
for $\tau\in[-1,0)$.
\end{proof}

By Lemma \ref{t:Gest2}, we obtain $
|(-\tau)^{-\frac{1}{2}} \underline{\mathfrak{G}}(\tau) | \lesssim 1 $, by continuously extend $(-\tau)^{-\frac{1}{2}} \underline{\mathfrak{G}}(\tau)$ to $\tau\in[-1,0]$, that is, let $(-\tau)^{-\frac{1}{2}} \underline{\mathfrak{G}}(\tau)|_{\tau=0}:=\lim_{\tau\rightarrow 0} \bigl[(-\tau)^{-\frac{1}{2}} \underline{\mathfrak{G}}(\tau)\bigr]$, then we derive $(-\tau)^{-\frac{1}{2}} \underline{\mathfrak{G}}(\tau)\in C^0([-1,0])  $.  By observing the definition of $\mathcal{F}$ \eqref{e:calF}, we conclude
\begin{equation*}
	\mathcal{F}\in C^0([-1,0],C^\infty(B_{\tilde{R}}(\Rbb^N),\mathbb{M}_{N\times N}))
\end{equation*}
and directly verifies \eqref{e:Fcd}.

\subsubsection{Verification of Condition \ref{c:5} in Appendix \ref{s:fuc}}\label{s:F4}
Now let us verify Condition \ref{c:5}. In order to do so, we have to verify
that there exists constants $\kappa, \,\bar{\gamma}_{2}$ and $\bar{\gamma}_{1}$ such that
\begin{equation}\label{e:pstv}
	\bar{\gamma}_{1}\mathbb{I} \leq  \B^{0}\leq \frac{1}{\kappa} \mathfrak{B} \leq \bar{\gamma}_{2}\mathbb{I} \quad \text{i.e., }\quad \bar{\gamma}_{1} \eta^T \mathbb{I} \eta\leq  \eta^T \B^{0} \eta\leq \frac{1}{\kappa} \eta^T\mathfrak{B}\eta \leq \bar{\gamma}_{2}\eta^T\mathbb{I} \eta
\end{equation}
for all $ (\tau,\U) \in[-1,0] \times B_R(\Rbb^{N}) $ and $\eta :=(\eta_1,\eta_2,\eta_3,\eta_4,\eta_5)^T \in  \Rbb^{N}$.

Firstly, by \eqref{e:frakB}, we note
\begin{align}\label{e:eBe1}
	A \eta^T\mathfrak{B}\eta= & (4\lambda +\mathfrak{Z}_1)\eta_1^2  +\Bigl(\frac{25}{36}+\mathfrak{Z}_0\Bigr) \eta_2^2 +\Bigl(4\lambda+\frac{2(3-8\iota^3)}{15}\Bigr) \eta_3^2+ (4\lambda+4+\mathfrak{Z}_6) \eta_4^2 \notag  \\
	&  +\frac{(3 \iota^3+2)^2}{2 (10 \lambda+\iota^3+9)} \eta_5^2  +\mathfrak{Z}_2\eta_1\eta_2 +(-8\lambda +\mathfrak{Z}_3) \eta_1\eta_3 + \mathfrak{Z}_4 \eta_1\eta_4  +\frac{2}{3}  (1-\iota^3)  \eta_2 \eta_4 \notag  \\
	& + \Bigl(-\frac{2(1-\iota^3) }{5} +\mathfrak{Z}_5 \Bigr) \eta_3 \eta_4  - \frac{(3 \iota^3+2)^2}{6 (10 \lambda+\iota^3+9)}  \eta_5 \eta_3 +\Bigl(2\iota^3 +\frac{4}{3} +\mathfrak{Z}_7\Bigr) \eta_5\eta_4   .
\end{align}
The following lemma verifies \ref{c:5}.
\begin{lemma}\label{t:lbd}
Suppose $\lambda>0$ and $\iota^3\in(0,1/5]$.
There are constants $\tilde{R}>0$ and $\gamma_1>0$, $\gamma_2>0$ given by
	\begin{align}
		\gamma_1:= & \frac{1}{2}\min\Bigl\{\frac{8\lambda}{5(3+800\lambda)},\; \frac{1}{1500},\; \frac{   1 }{27 (13 \lambda+12) (10 \lambda+ \iota^3+9)^2} \Bigr\} , \label{e:gm1}  \\
		\intertext{and}
		\gamma_2:= & \max\{8\lambda +1+\gamma_1,\;4 \lambda+6+\gamma_1\}  ,\label{e:gm2}
	\end{align}
such that
\begin{equation*}
	\frac{\gamma_1}{A}	\eta^T \mathds{1} \eta<\eta^T\mathfrak{B}\eta < \frac{\gamma_2}{A} \eta^T\mathds{1} \eta .
\end{equation*}
for all $(\tau,\U) \in [-1,0]\times B_{\tilde{R}}(\Rbb^N)$.
\end{lemma}
\begin{proof}
	Since $\mathfrak{Z}_\ell(\tau,0)=0$, by its continuity, there is a constant $\tilde{R}>0$, such that if $\U\in B_{\tilde{R}}(\Rbb^N)$, then \begin{equation}\label{e:sumz}
		\sum_{\ell=0}^7 |\mathfrak{Z}_\ell(\tau,\U)| < \gamma_1  .
	\end{equation}
By taking
\begin{equation*}
	\delta_1:=\frac{3-13 \iota^3}{800 \lambda}+1, \quad \delta_3:=\frac{12}{35} (13 \lambda+12) \AND \delta_5:=\frac{18 (3-13 \iota^3) (10 \lambda+\iota^3+9)}{25 (3 \iota^3+2)^2},
\end{equation*}
and using the Cauchy's inequality with $\delta_\ell>0$ ($\ell=1,3,5$), we have estimates
\begin{gather}
2|\eta_4\eta_1| \leq \eta_4^2+\eta_1^2, \quad 	2|\eta_1\eta_3 |\leq \frac{1}{\delta_1} \eta_1^2+ \delta_1 \eta_3^2, \quad 2|\eta_2\eta_4| \leq \eta_2^2+\eta_4^2, \quad 2|\eta_3\eta_4| \leq \eta_3^2+\eta_4^2,  \label{e:eta1}\\
	2|\eta_2\eta_1| \leq \eta_2^2+\eta_1^2, \quad  2|\eta_3\eta_5| \leq \frac{1}{\delta_5}\eta_5^2+\delta_5 \eta_3^2, \quad   2|\eta_5\eta_4| \leq \frac{1}{\delta_3}\eta_5^2+\delta_3\eta_4^2  , \label{e:eta2}  \\
	\intertext{and}
	|\mathfrak{Z}_\ell \eta_a\eta_b| \leq \frac{|\mathfrak{Z}_\ell|}{2}(\eta_a^2+\eta_b^2).  \label{e:eta3}
\end{gather}
Using \eqref{e:eta1}--\eqref{e:eta3}, we estimate \eqref{e:eBe1},
\begin{align}\label{e:eBe3}
	A \eta^T\mathfrak{B}\eta
\geq  &  \biggl(\frac{4 \lambda (3-13 \iota^3)}{800 \lambda+3-13 \iota^3}  +\mathfrak{Z}_1-\frac{1}{2}\sum_{\ell=2}^4 |\mathfrak{Z}_\ell|  \biggr) \eta_1^2  +  \biggl(\frac{\iota^3}{3}+\frac{13}{36}+\mathfrak{Z}_0-\frac{1}{2} |\mathfrak{Z}_2|\biggr) \eta_2^2     \notag  \\
	& + \biggl(\frac{4\bigl(9 \lambda (3-13 \iota^3)+(19-94  \iota^3)\bigr)}{105}  +\mathfrak{Z}_6-\frac{1}{2}\sum_{\ell=4,5,7}|\mathfrak{Z}_\ell|  \biggr)\eta_4^2  \notag \\
	&  + \biggl(\frac{3-13 \iota^3}{600} -\frac{1}{2} \sum_{\ell=3,5} |\mathfrak{Z}_\ell| \biggr)\eta_3^2    +\Bigl(\mathtt{q}-\frac{1}{2}|\mathfrak{Z}_7| \Bigr)\eta_5^2
\end{align}
where
\begin{align}\label{e:q1exp}
	\mathtt{q} :
	=&	\frac{(3 \iota^3+2)^2}{2 (10 \lambda+\iota^3+9)}- \frac{25 (3 \iota^3+2)^4}{216 (3-13 \iota^3) (10 \lambda+\iota^3+9)^2}-\frac{35 (3 \iota^3+2)}{36 (13 \lambda+12) } \notag  \\
	= & - \frac{(3 \iota^3+2) \bigl[ 120 \lambda^2 \mathfrak{q}_1(\iota^3)+\lambda \mathfrak{q}_2(\iota^3)+6 \mathfrak{q}_3(\iota^3) \bigr]}{216 (13 \lambda+12) (3-13 \iota^3) (10 \lambda+ \iota^3+9)^2} >0
\end{align}
and
\begin{align*}
	\mathfrak{q}_1(\iota^3):=&
	(13 \iota^3-3) (351 \iota^3+59)   ,  \\
	\mathfrak{q}_2(\iota^3):= &  63531 (\iota^3)^3+985062 (\iota^3)^2-40392 (\iota^3)-37576   , \\
	\mathfrak{q}_3(\iota^3):= & 9319 (\iota^3)^3+74103 (\iota^3)^2-1413 (\iota^3)-2759 .
\end{align*}
Using the derivatives of $\mathfrak{q}_\ell$ ($\ell=1,2,3$), we can have the monotonicity of $\mathfrak{q}_\ell$. Further we can estimate
\begin{equation}\label{e:hest}
	\mathfrak{q}_1(\iota^3)\leq -\frac{1292}{25}<0 ,\quad 	\mathfrak{q}_2(\iota^3) \leq -\frac{717959}{125}<0 ,\AND 	\mathfrak{q}_3(\iota^3) \leq  -\frac{366}{125}<0
\end{equation}
for $\iota^3 \in(0,1/5]$.
Noting
\begin{equation}\label{e:qaest}
  \frac{3 \iota^3+2}{3-13 \iota^3}  >\frac{2}{3} \AND 775200 \lambda^2+717959 \lambda+2196> 1500
\end{equation}
for $ \iota^3 \in(0,1/5]$ and $ \lambda>0$.
Then, by inserting the estimate \eqref{e:hest} into \eqref{e:q1exp}, with the help of \eqref{e:qaest}, we estimate
\begin{equation}\label{e:qest3}
\mathtt{q}
	\geq   \frac{\left(775200 \lambda^2+717959 \lambda+2196\right) (3 \iota^3+2)}{27000 (13 \lambda+12) (3-13 \iota^3) (10 \lambda+\iota^3+9)^2}
	>   \frac{   1 }{27 (13 \lambda+12) (10 \lambda+ \iota^3+9)^2}   .
\end{equation}

Then inserting the estimate \eqref{e:qest3} into \eqref{e:eBe3}, we arrive at
\begin{align*}
	A \eta^T\mathfrak{B}\eta\geq & \Bigl(\frac{8\lambda}{5(3+800\lambda)} -\sum_{\ell=0}^7 |\mathfrak{Z}_\ell|\Bigr) \eta_1^2+\Bigl(\frac{\iota^3}{3} +\frac{13}{36}-\sum_{\ell=0}^7 |\mathfrak{Z}_\ell|\Bigr) \eta_2^2+\Bigl(\frac{1}{1500}-\sum_{\ell=0}^7 |\mathfrak{Z}_\ell| \Bigr)\eta_3^2 \notag  \\
	& +\Bigl(\frac{4}{525}-\sum_{\ell=0}^7 |\mathfrak{Z}_\ell| \Bigr) \eta_4^2+ \Bigl( \frac{   1 }{27 (13 \lambda+12) (10 \lambda+ \iota^3+9)^2} -\sum_{\ell=0}^7 |\mathfrak{Z}_\ell| \Bigr) \eta_5^2  \notag  \\
	\overset{\eqref{e:gm1} \& \eqref{e:sumz}}{>} & (2\gamma_1-\gamma_1)\sum^5_{k=1} \eta_k^2>\gamma_1	\eta^T \mathds{1} \eta  .
\end{align*}

On the other hand, we estimate the upper bound of $\eta^T\mathfrak{B}\eta$ by using the Cauchy's inequality and taking $\delta_1=\delta_3=\delta_5=1$ in \eqref{e:eta1}--\eqref{e:eta2},
\begin{align*}
	A \eta^T\mathfrak{B}\eta
\leq & \Bigl(8\lambda+\sum_{\ell=0}^7|\mathfrak{Z}_\ell|\Bigr) \eta_1^2 +\Bigl(\frac{25}{36}+\frac{1-\iota^3}{3} +\sum_{\ell=0}^7|\mathfrak{Z}_\ell|  \Bigr)  \eta_2^2  \notag  \\
& +\Bigl(8\lambda+\frac{2(3-8\iota^3)}{15}+ \frac{1-\iota^3}{5}    + \frac{(3 \iota^3+2)^2}{12 (10 \lambda+\iota^3+9)} +\sum_{\ell=0}^7|\mathfrak{Z}_\ell| \Bigr) \eta_3^2 \notag  \\
	& + \Bigl[(4\lambda+4)    +\frac{8(1-\iota^3)}{15}  + \Bigl(\iota^3 +\frac{2}{3} \Bigr)  +\sum_{\ell=0}^7|\mathfrak{Z}_\ell|
	\Bigr] \eta_4^2 \notag  \\
	&  + \Bigl[\frac{7 (3 \iota^3+2)^2}{12 (10 \lambda+\iota^3+9)}      + \Bigl(\iota^3 +\frac{2}{3} \Bigr) +\sum_{\ell=0}^7|\mathfrak{Z}_\ell| \Bigr]\eta_5^2\notag  \\
	\overset{\eqref{e:sumz}}{<} & (8\lambda+\gamma_1) \eta_1^2 + \Bigl(\frac{37}{36}+\gamma_1\Bigr) \eta_2^2 +\Bigl(8\lambda +1+\gamma_1 \Bigr) \eta_3^2  + \Bigl(4 \lambda+6+\gamma_1 \Bigr)\eta_4^2  + \Bigl(\frac{3}{2}+\gamma_1 \Bigr) \eta_5^2 \notag  \\
	\overset{\eqref{e:gm2}}{<} & \gamma_2 \eta^T\mathds{1} \eta .
\end{align*}
We complete the proof.
\end{proof}

\begin{lemma}\label{t:upbd}
Suppose $\lambda>0$ and $\iota^3\in(0,1/5]$. There are constants $\tilde{R}>0$ and $\hat{\gamma}_1>0$, $\hat{\gamma}_2>0$ given by
\begin{align*}
	\hat{\gamma}_1:=& \min\Bigl\{\frac{7}{150} , \; \frac{1}{ 4 +B^{-1}   \max_{\tau\in[-1,0]}\underline{\mathfrak{G} }  } 	\Bigl( \frac{ 25 }{ 36   }   - \gamma_1 \Bigr) \Bigr\}  >0 \\
	\intertext{and}
	\hat{\gamma}_2:=& \max\Bigl\{1,\lambda+\frac{1}{10} , \; \frac{1}{4+ B^{-1} \min_{\tau\in[-1,0]}\underline{\mathfrak{G} }  } 	\Bigl( \frac{ 25 }{ 36   } +\frac{ 25 }{ 36  \beta } + \gamma_1 \Bigr) \Bigr\} >0 ,
\end{align*}
such that
\begin{equation*}
\hat{\gamma}_1 \eta^T \mathbb{I} \eta<	\eta^T\B^0 \eta <\hat{\gamma}_2 \eta^T \mathbb{I} \eta
\end{equation*}
for all $(\tau,\U) \in [-1,0]\times B_{\tilde{R}}(\Rbb^N)$.
\end{lemma}
\begin{proof}
	Firstly, there is a constant $\tilde{R}>0$, such that if $\U\in B_{\tilde{R}}(\Rbb^N)$, then \eqref{e:sumz} holds and it  yields $-\gamma_1<-|\mathfrak{Z}_\ell|<\mathfrak{Z}_\ell<|\mathfrak{Z}_\ell|<\gamma_1$.
	By the expression  \eqref{e:B0} of $\B^0$,  we obtain
	\begin{align*}
	\frac{7}{150}<\lambda+\frac{1}{30} (3-8 \iota^3)< 	\lambda+\frac{1}{10},  \quad \text{for} \quad \iota^3\in(0,1/5],
	\end{align*}
and
	\begin{align*}
		  \frac{1}{ 4 +B^{-1} \max_{\tau\in[-1,0]}\underline{\mathfrak{G}}  } 	\Bigl( \frac{ 25 }{ 36   }   - \gamma_1 \Bigr)&< \frac{1}{ 4 +B^{-1} \underline{\mathfrak{G}} }	\Bigl( \frac{ 25 }{ 36   } +   \frac{ 25 }{ 36   \uf}  +  \mathfrak{Z}_0 \Bigr) \notag \\
		  & \overset{(\star)}{<}\frac{1}{ 4 + B^{-1} \min_{\tau\in[-1,0]}\underline{\mathfrak{G} }   }	\Bigl( \frac{ 25 }{ 36   }  \frac{1+\beta}{\beta}+  \gamma_1 \Bigr)  .
	\end{align*}
We point out in $(\star)$, $4 + B^{-1} \min_{\tau\in[-1,0]}\underline{\mathfrak{G} }  >0$,
since, by \eqref{e:chi4}, we have $\underline{\mathfrak{G} }(\tau) >-4B $ for $\tau\in[-1,0]$. Then we conclude this lemma.
\end{proof}

Using above Lemma \ref{t:lbd} and \ref{t:upbd} and taking
\begin{equation}\label{e:kpa}
		 \kappa:=\frac{\gamma_1}{A \hat{\gamma}_2}, \quad \bar{\gamma}_1:=\hat{\gamma}_1 \AND \bar{\gamma}_2:= \frac{\gamma_2\hat{\gamma}_2}{\gamma_1} ,
\end{equation}
we conclude \eqref{e:pstv} and thus have verified Condition \ref{c:5}.

\subsubsection{Complete verifications of Conditions \ref{c:4} and \ref{c:7} in Appendix \ref{s:fuc}}\label{s:F6}
Since $\Pbb=\mathbb{I}$ and $\Pbb^\perp=0$, we only need  to verify there are constants $\theta$ and $\beta_\ell>0$ ($\ell=0,1$), such that it verifies \eqref{e:PhP1} which reduces to
\begin{equation}\label{e:divB1}
	\mathrm{div} \B (\tau,\U,\mathcal{W}) =   \;\mathcal{O}\bigl(\theta   +|\tau|^{-\frac{1}{2}}\beta_0+|\tau|^{-1}\beta_1 \bigr) ,
\end{equation}
where\footnote{To facilitate readers' understanding, we have labeled the orders of the terms below. However, their proofs will be presented in the subsequent expositions. }
\begin{align}\label{e:divB2}
	\mathrm{div} \B(\tau,\U,\mathcal{W}):=& \underbrace{\del{\tau}  \B^0(\tau, \U)}_{(d)\;(-\tau)^{-\frac{1}{2}} -\text{term}}  +	\del{\U}  \B^0(\tau, \U)  \cdot (\B^0(\tau, \U) )^{-1} \Bigl[\underbrace{-	\B^\zeta (\tau, \U)  \cdot \mathcal{W} }_{(a)\;\tau^{-1}-\text{term}} +\underbrace{ \frac{1}{ \tau}	 \mathfrak{B} (\tau, \U)   \U}_{(b)\; \tau^{-1}-\text{term}}  \notag  \\
	&   +  \mathcal{H} (\tau, \U) +  \underbrace{(-\tau)^{-\frac{1}{2}} 	  \mathcal{F}(\tau, \U)}_{(e)\;(-\tau)^{-\frac{1}{2}} -\text{term}}  \Bigr] + \underbrace{	\del{\U}  \B^\zeta(\tau, \U ) \cdot \mathcal{W} }_{(c)\;\tau^{-1}-\text{term}}
\end{align}
for $(\tau, \U, \mathcal{W}) \in [-1,0)\times B_{\tilde{R}}(\Rbb^N) \times  B_{\tilde{R}}(\Rbb^N)$.
Requirements \eqref{e:PhP2}--\eqref{e:PhP4} hold constantly since $\Pbb^\perp=0$ results all the left hands of \eqref{e:PhP2}--\eqref{e:PhP4}  vanish.

Let us firstly recall \eqref{e:Z0}. That is
\begin{equation*}
	\mathfrak{Z}_0=  \mathfrak{Z}_0(\tau,\uu,\underline{\nu}) = \frac{ 25 }{ 36   } \Bigl(1+ \frac{ 1 }{    \underline{f} }  \Bigr)  \Bigl[ \Bigl( 1+ \frac{     \underline{f}   \underline{u}}{ 1+   \underline{f}}  \Bigr)^{-\frac{3}{5}}   -1\Bigr]  -   \frac{ 25 (4  + B^{-1}    \underline{\mathfrak{G}})}{ 9   }    \underline{\nu}^2  .
\end{equation*}
Similar to \eqref{e:ttf2}, with the help of  \eqref{e:f0frl} and \eqref{e:Gdef} (yielding $ f_0 =	B^{-\frac{1}{2}}     t^{-1}  f^{\frac{1}{2}}   	\chi^{\frac{1}{2}}(1+f) $), we calculate
\begin{align}\label{e:dtauf}
	\del{\tau} \uf= \underline{\frac{t^{\frac{2}{3}} (1+f )^{\frac{1}{3} } }{A  B   g^{\frac{2}{3A}+1}  f}  f_0}=  \underline{\frac{1  }{A  B    g    }        	\chi         	 (1+f)  }  .
\end{align}
Let us now calculate the only non-vanishing term in $\del{\tau}\B^0$ by \eqref{e:chi4}, \eqref{e:dtgeq} and \eqref{e:dtauf},
\begin{align}\label{e:dtB0}
	& \del{\tau} \biggl( \frac{1}{ 4 +B^{-1} \underline{\mathfrak{G}} }	\Bigl( \frac{ 25 }{ 36   } +   \frac{ 25 }{ 36   \uf}  +  \mathfrak{Z}_0 \Bigr) \biggr) \notag  \\
	= &\biggl[-\biggl(\frac{40 B^2 \uu  }{9 \uf  (\uf +1)  } + \frac{25 B^2   }{9 \uf  (\uf +1)  } \Bigl(1+\frac{1}{f}\Bigr)  + \frac{10 B \uu \underline{\mathfrak{G}}  }{9 \uf (\uf+1)  }    + \frac{25 B \underline{\mathfrak{G}}    }{36 \uf (\uf+1)  } + \frac{25 B \underline{\mathfrak{G}}  }{36 \uf^2 (\uf+1)  } \biggr) \del{\tau} \uf \notag\\
	&  -\biggl(\frac{1}{18    }+\frac{ \uu   }{36   } +\frac{  \uu \uf    }{36  }+\frac{  \uf   }{36  }  +\frac{1}{36 \uf} \biggr) \frac{25 B \del{\tau}\underline{\mathfrak{G}} }{ (\uf+1) }\biggr](4 B+\underline{\mathfrak{G}})^{-2} \left(\frac{ \uf \uu }{\uf +1}+1\right)^{-\frac{8}{5}}  \notag   \\
	= &\biggl[-\biggl(\frac{40 B  \uu  }{9    } + \frac{25 B    }{9      }  + \frac{25 B   }{9 \uf    } + \frac{10   \uu \underline{\mathfrak{G}}  }{9    }    + \frac{25   \underline{\mathfrak{G}}    }{36     } + \frac{25   \underline{\mathfrak{G}}  }{36 \uf    } \biggr) \frac{1  }{A }  \underline{\xi}\Bigl(1+\frac{1}{\uf}\Bigr)      	\underline{\chi}         	    \notag\\
	&  -25 B \biggl(\frac{1}{18 (\uf+1)   }+\frac{ \uu   }{36   } +\frac{  \uf   }{36 (\uf+1) }  +\frac{1}{36 \uf (\uf+1)} \biggr)  \biggl(\frac{  \underline{\mathfrak{G}}  \underline{\chi} }{3  A  B   \tau   }  - \frac{ \underline{\chi}^2 }{A  B   \underline{g  f}    	 }      \notag \\
	&   -   \frac{2 \underline{\chi}^{\frac{3}{2}}  }{ 3A    B^{\frac{1}{2}}   \underline{g } \underline{f}^{\frac{1}{2}}   }\biggr) \biggr](4 B+\underline{\mathfrak{G}})^{-2} \left(\frac{ \uf \uu }{\uf +1}+1\right)^{-\frac{8}{5}}  \notag  \\
	\overset{(\star)}{=} & \left(\frac{ \uf \uu }{\uf +1}+1\right)^{-\frac{8}{5}}  \bigl[\mathcal{S}_1(\tau) +\mathcal{S}_2(\tau) \uu+|\tau|^{-\frac{1}{2}} \bigl(\mathcal{S}_3(\tau)+\mathcal{S}_4(\tau) \uu \bigr)\bigr]
\end{align}
where $\mathcal{S}_\ell$ ($\ell=1,\cdots,4$) are continuous in $\tau\in[-1,0]$,  and $\lim_{\tau\rightarrow 0}\mathcal{S}_\ell(\tau)=0$ for $\ell=1,2$. In the last step $(\star)$, we have
used Theorem \ref{t:mainthm0}, Proposition \ref{t:limG}, \ref{t:fginv0},  Corollary \ref{s:gf1/2} and Lemma \ref{t:Gest2} (which implies $(-\tau)^{-\frac{1}{2}}|\underline{\mathfrak{G}}| \lesssim 1$) and $\U\in B_{\tilde{R}}(\Rbb^N)$. By \eqref{e:dtB0},  we obtain that there exists constants $\tilde{\theta}$ and $\tilde{\beta}_{0}$, such that
\begin{equation}\label{e:F6.1}
	\del{\tau}\B^0 =   \;\mathcal{O}\bigl(\tilde{\theta} +|\tau|^{-\frac{1}{2}}\tilde{\beta}_{0} \bigr) .
\end{equation}
In addition, this \eqref{e:dtB0} also verifies $\B^0\in C^{1}([-1,0),C^{\infty}( B_R(\mathbb R^{N}), \mathbb M_{N\times N})$ which finishes Condition \ref{c:4}.

Next let us calculate the key term  $\del{\U}\mathfrak{Z}_0$, i.e.,
\begin{align}\label{e:duz0}
	\del{\uu}\mathfrak{Z}_0=-\frac{5}{12}  \Bigl(\frac{ \uf \uu}{\uf+1}+1\Bigr)^{-\frac{8}{5}}\AND  \del{\underline{\nu}}\mathfrak{Z}_0= - \frac{50}{9}    \left(\frac{\underline{\mathfrak{G}}}{B}+4\right) \underline{\nu} .
\end{align}
We can calculate the non-vanishing term of $\del{\U}\B^0$ is
$( 4 +B^{-1} \underline{\mathfrak{G}} )^{-1} \del{\U}  \mathfrak{Z}_0$
which, by \eqref{e:duz0}, implies there exist constants $\hat{\theta}$ and $\tilde{R}$, such that
\begin{equation}\label{e:F6.2}
	\del{\U}\B^0=\;\mathcal{O}\bigl(\hat{\theta}   \bigr)
\end{equation}
for $\U\in B_{\tilde{R}}(\Rbb^N)$.

Next, we turn to $\del{\U}  \B^\zeta$, and the non-vanishing terms of $\del{\U}\B^\zeta$ are
\begin{equation*}
	\del{\underline{\nu} }\biggl(-\frac{ 1 }{A \tau}  \frac{2  \underline{\chi}    }{3  B }  \underline{\nu}  \biggr) =-\frac{ 1 }{A \tau}  \frac{2  \underline{\chi}    }{3  B } \AND \frac{1}{A\tau} \del{\U}\biggl( \frac{ 5 }{ 36  } +   \frac{ 5 }{ 36   \uf}  + \frac{1}{5} \mathfrak{Z}_0\biggr)=\frac{1}{5A\tau} \del{\U}    \mathfrak{Z}_0 .
\end{equation*}
With the help of \eqref{e:duz0}, there exist constants $\hat{\beta}_{1}$  and $\tilde{R}$, such that
\begin{equation}\label{e:F6.3}
	\del{\U}\B^\zeta =   \;\mathcal{O}\bigl( |\tau|^{-1}\hat{\beta}_{1} \bigr)
\end{equation}
for $\U\in B_{\tilde{R}}(\Rbb^N)$.
We also note there are constants $\hat{\beta}_0$ and $\tilde{R}$, such that
\begin{equation}\label{e:F6.5}
	(-\tau)^{-\frac{1}{2}} 	  \mathcal{F}  =   \;\mathcal{O}\bigl( |\tau|^{-\frac{1}{2}}\hat{\beta}_0  \bigr)
\end{equation}
for $\U\in B_{\tilde{R}}(\Rbb^N)$.

With the help of \eqref{e:F6.1}, \eqref{e:F6.2}, \eqref{e:F6.3} and \eqref{e:F6.5}, directly examining \eqref{e:divB2} terms by terms, we can conclude there are constants $
	\theta>0, \;   \beta_0>0 $ and $ \beta_1 >0$,
such that it verifies \eqref{e:divB1} for $(\tau, \U, \mathcal{W}) \in [-1,0)\times B_{\tilde{R}}(\Rbb^N) \times  B_{\tilde{R}}(\Rbb^N)$.  Moreover, by properly shrinking the ball $B_{\tilde{R}}(\Rbb^N)\ni \U, \; \mathcal{W}$, we can further take $\beta_1$ in \eqref{e:divB1} small enough such that
\begin{equation}\label{e:bt1}
	\beta_1 \in \biggl(0,\frac{2\gamma_1\hat{\gamma}_1}{A \hat{\gamma}_2} \biggr) ,
\end{equation}
since the $\tau^{-1}$-terms $(a)$, $(b)$ and $(c)$ in \eqref{e:divB2} all include a factor $\U$ or $\mathcal{W}$.
This finishes the verification of Condition \ref{c:7}.

Now after verifying all the conditions \ref{c:2}--\ref{c:7}, we conclude \eqref{e:fuc} is a Fuchsian system given in Appendix \ref{s:fuc}. Therefore, we can use Theorem \ref{t:fuc} to obtain the global estimates.  By \S\ref{s:F3}, we have known $\tilde{\B}^0(\tau),\;\tilde{\B}^\zeta_2(\tau), \;\tilde{\mathfrak{B}}(\tau) $ are all only $\tau$ dependent, hence in this case,  we find the parameter $\mathtt{b}$ defined in Theorem \ref{t:fuc} satisfying $\mathtt{b}=0$.

\subsection{The  uniqueness and periodicity}\label{s:stp3}
As the system \eqref{e:fuc} is a symmetric hyperbolic system, the standard theory of symmetric hyperbolic systems (see, for example, \cite[\S$16$]{Taylor2010} and \cite[Theorem $5.1$]{Racke2015}) ensures the existence and uniqueness of the solution $\U$ for $(\tau,\zeta) \in[-1,\tau_1] \times \Rbb$ ($\tau_1<0)$. Moreover, since the system \eqref{e:main1}--\eqref{e:main2} is equivalent to the system \eqref{e:fuc}, this entails that if $(\hrho(t,\zeta),\nu(t,\zeta))$ satisfies \eqref{e:main1}--\eqref{e:main2} with the data \eqref{e:data3}, then the solution is unique. In other words, once we prove the existence of $\U$ for some $\tau$, $\U$ is unique within the interval, and the corresponding solution $(\hrho(t,\zeta),\nu(t,\zeta))$ is unique in terms of $t$. Additionally, using Proposition \ref{t:prdsl}, we conclude that the solution $(\hrho(t,\zeta),\nu(t,\zeta))$ is a periodic solution satisfying \eqref{e:prdsl}. By Lemma \ref{t:Psprd}, we know that $\Psi(t,\zeta)$ is also a periodic function. Thus, instead of considering the periodic solution of the system \eqref{e:fuc} on $(t,\xi)\in[1,t_m)\times \Rbb$, we consider the equations \eqref{e:main1}--\eqref{e:main2} with data \eqref{e:data3} for $(t,\xi)\in[1,t_m)\times \Tbb$, which enables us to directly apply Theorem \ref{t:fuc}.

\subsection{The  global existence of the solution $\U$}\label{s:stp5}
After verifying that  \eqref{e:fuc} is a Fuchsian system for $(\tau,\zeta) \in [-1,0) \times \Tbb$ and $\U \in B_{\tilde{R}}(\Rbb^N)$, we are able to
directly use Theorem \ref{t:fuc}, that is, let $s\in\Zbb_{>\frac{7}{2}}$, $\U|_{\tau=-1} \in H^s(\Tbb)$ and, by \eqref{e:kpa} and \eqref{e:bt1},
\begin{equation*}
\kappa :=\frac{\gamma_1}{A\hat{\gamma}_2} >\frac{\beta_1}{2\hat{\gamma}_1}=\frac{\beta_1}{2\bar{\gamma}_1} ,   \quad \text{(verifying \eqref{e:kpa2})},
\end{equation*}
then there exist small constants $\sigma, \sigma_0>0$ and $\sigma<\sigma_0$, such that if
\begin{equation*}
	\|\U|_{\tau=-1}\|_{H^s} \leq \sigma  ,
\end{equation*}
then there exists a unique solution
\begin{equation}\label{e:sltdel}
	\U\in C^0([-1,0),H^s(\Tbb) ) \cap C^1([-1,0),H^{s-1}(\Tbb) ) \cap \Li ([-1,0),H^s(\Tbb))
\end{equation}
of the initial value problem for the Fuchsian system \eqref{e:fuc}. Moreover, for $-1 \leq \tau <0$, the solution $\U$ satisfies the energy estimate
\begin{equation}\label{e:ineq1ab}
	\|\U(\tau)\|_{H^s} -\int^\tau_{-1}\frac{1}{\tau} \|\U(z)\|^2_{H^s} dz \leq C(\sigma_0,\sigma_0^{-1}) \|\U|_{\tau=-1}\|_{H^s} < C_1 \sigma  .
\end{equation}

\subsection{Proofs of Theorem \ref{t:mainthm1}} \label{s:stp7}
Before proceeding the main proof, we first estimate $	\|\Psi(t)\|_{H^s(\Tbb)} $ by $\|u(t)\|_{H^{s-1}(\Tbb)} $.
\begin{lemma}\label{t:psest}
	Suppose $\Psi$ is given by \eqref{e:ps2}, then there is a constant $C >0$, such that
	\begin{equation*}
		\|\Psi(t)\|_{H^s(\Tbb)}  \leq C_1  \|u(t)\|_{H^{s-1}(\Tbb)}   .
	\end{equation*}	
\end{lemma}
\begin{proof}
	Recalling the definition of Sobolev norms (see \S\ref{s:funsp}), we have
	\begin{equation*}
		\|\Psi(t)\|_{H^s(\Tbb)}^2 =   \sum_{0\leq k \leq s} \int_{\mathbb{T}} (\partial_\zeta^k \Psi(t,\zeta),\partial_\zeta^k \Psi(t,\zeta))\, d\zeta=\sum_{0\leq k\leq s} \|\partial_\zeta^k \Psi(t,\zeta)\|_{L^2(\Tbb)}^2.
	\end{equation*}
	By using \eqref{e:dipsi} and  the  induction, we can prove the higher order derivatives of $\Psi$,
	\begin{equation}\label{e:psest0}
		\partial_{\zeta}^k \Psi=(-3)^k\Psi +\sum^k_{\ell=1}(-3)^{k-\ell} \partial_{\zeta}^{\ell-1} u , \quad \text{for} \quad k=1,\cdots s.
	\end{equation}
	Then by \eqref{e:psest0} and the triangle inequality,
	\begin{equation}\label{e:psest1}
		\|\Psi(t)\|_{H^s(\Tbb)}^2 \leq   C\|\Psi(t)  \|_{L^2(\Tbb)}^2 +C \|u(t)\|_{H^{s-1}(\Tbb)}^2 .
	\end{equation}
	We need to estimate
	\begin{equation}\label{e:pssob}
		\|\Psi (t) \|_{L^2(\Tbb)}^2= \int_{\mathbb{T}} |\Psi(t,\zeta)|^2  d\zeta .
	\end{equation}
	By \eqref{e:ps2}, note
	\begin{equation}\label{e:psest}
		|\Psi(t,\zeta)|  \leq  \frac{1}{e^{3\zeta} }\int^\zeta_{-\infty} \bigl|u(t,z)\bigr|e^{3z} dz\leq  \|u(t)\|_{\Li(\Tbb)}\frac{1}{e^{3\zeta} }\int^\zeta_{-\infty}  e^{3z} dz=\frac{1}{3} \|u(t)\|_{\Li(\Tbb)} .
	\end{equation}
	Substituting the estimate \eqref{e:psest} into \eqref{e:pssob}, with the help of the Sobolev embedding theorem, yields
	\begin{equation}\label{e:psest2}
		\|\Psi (t) \|_{L^2(\Tbb)}^2 \leq C\|u(t)\|_{\Li(\Tbb)}^2\leq C\|u(t)\|_{H^{s-1}(\Tbb)}^2 .
	\end{equation}
	Inserting the estimate \eqref{e:psest2} into \eqref{e:psest1}, we obtain there is a constant $C >0$, such that
	\begin{equation*}
		\|\Psi(t)\|_{H^s(\Tbb)}^2 \leq  C \|u(t)\|_{H^{s-1}(\Tbb)}^2  ,
	\end{equation*}
	which concludes this lemma.
\end{proof}

Now, we are in a position to prove Theorem \ref{t:mainthm1}.
\begin{proof}[Proof of Theorem \ref{t:mainthm1}]
	Using \eqref{e:ww}--\eqref{e:wi},  \eqref{e:u} (noting now $\cc=1/5$ by \eqref{e:para0}), \eqref{e:vv} and \eqref{e:ps2} we arrive at the transforms
	\begin{gather*}
		\hrho(t,\zeta)=f(t)+f(t)u(t,\zeta), \quad 	\del{t}\hrho(t,\zeta)=f_0(t)+f_0(t)u_0(t,\zeta) \label{e:trf1} \intertext{and} \partial_{\zeta}\hrho(t,\zeta)= 5(1+f(t)) u_\zeta(t,\zeta) , \quad \hv=\frac{f_0 R }{3(1+f)}\nu  .  \label{e:trf2}
	\end{gather*}
By \eqref{e:xvartrf} and \eqref{e:logcd}, we obtain the transform
\begin{equation*}
	|\bx|=t^{\frac{2}{3}} (1+f)^{-\frac{1}{3}}\exp\zeta.
\end{equation*}
which consists with \eqref{e:xztsf}.

Further, using \eqref{e:xvartrf},  \eqref{e:undf}, \eqref{e:direl2},  \eqref{e:vv} and \eqref{e:dzdr}, we arrive at
\begin{gather}
\varrho(t,x^i) =f(t)+f(t)u(t,\zeta),  \label{e:etrf1}  \\
	\del{t}\varrho(t,x^i)+\Bigl(\frac{2}{3t}-\frac{f_0(t)}{3(1+f(t))}\Bigr)x^i\del{i}\varrho(t,x^i)= f_0(t)+f_0(t)u_0(t,\zeta), \label{e:etrf2}\\
x^i	\partial_{i}\varrho(t, x^i) = 5(1+f(t)) u_\zeta(t,\zeta),  \label{e:etrf3} \\
	v^i(t,x^i) =\Bigl(\frac{2}{3t}-\frac{ (1-\nu(t,\zeta)) f_0(t) }{3(1+f(t))}  \Bigr)x^i .   \label{e:etrf4}
\end{gather}
When $t=1$, with the help of \eqref{e:data0}, \eqref{e:feq2b} and \eqref{e:udvar}, the relations \eqref{e:etrf1} and \eqref{e:etrf4} become
\begin{gather}
\varrho|_{t=1}(x^i) =	\beta \mathcal{d}(|\bx|)=\beta (1+u(1,\zeta)) \; \Leftrightarrow  \;   \uu(-1,\zeta)=\mathcal{d}\circ\mathcal{y}_1- 1   , \label{e:dt1}\\
v^i|_{t=1}(x^i) =\Bigl(\frac{2}{3}-  (1-\nu(t,\zeta))   \gamma     \Bigr)x^i = \frac{2}{3}x^i +\gamma \mathcal{v}(|\bx|) x^i \; \Leftrightarrow  \;   \underline{\nu}(-1,\zeta) =\mathcal{v}\circ\mathcal{y}_1 +1. \label{e:dt2}
\end{gather}
Using \eqref{e:keyid2} and \eqref{e:uzuz}, we have
\begin{gather}
	\underline{u_\zeta}(-1,\zeta)=\frac{ \beta }{5(1+\beta)} \del{\zeta} \uu(-1,\zeta)  ,   \label{e:dt3}\\
	\underline{u_0}(-1,\zeta)=\frac{\beta \uu (-1,\zeta)}{1+\beta} -  \frac{5}{3}   \underline{\nu}  (-1,\zeta) \underline{u_\zeta } (-1,\zeta) - \Bigl[\frac{1}{3}\del{\zeta}\underline{\nu}(-1,\zeta)	+  \underline{\nu}(-1,\zeta) \Bigr]\Bigl(1 + \frac{\beta \uu(-1,\zeta)}{1+\beta}\Bigr)   .   \label{e:dt4}
\end{gather}

According to the data satisfying \eqref{e:mtdata}, by using \eqref{e:dt1} and \eqref{e:dt2}, we obtain estimates
\begin{equation}\label{e:dtest1}
	\|\uu(-1)\|_{H^{s}(\Tbb)} +\|\underline{\nu}(-1)\|_{H^{s}(\Tbb)} \leq	\|\uu(-1)\|_{H^{s+1}(\Tbb)} +\|\underline{\nu}(-1)\|_{H^{s+1}(\Tbb)} \leq   \sigma_\star\sigma  .
\end{equation}
Using \eqref{e:dt3}, we derive
\begin{equation}\label{e:dtest2}
	\|\underline{u_\zeta}(-1)\|_{H^s(\Tbb)}  \leq  \| \del{\zeta} \uu(-1) \|_{H^s(\Tbb)}  \leq C \|   \uu(-1) \|_{H^{s+1}(\Tbb)} \overset{\eqref{e:dtest1}}{\leq} C\sigma_\star \sigma  .
\end{equation}
Similarly, using \eqref{e:dt4}, with the help of \eqref{e:dtest1} and \eqref{e:dtest2}, we estimate
\begin{align}\label{e:dtest3}
		\|\underline{u_0}(-1)\|_{H^s(\Tbb)} \leq & C\Bigl(\|   \uu(-1) \|_{H^{s}(\Tbb)}+\|   \underline{\nu}(-1) \|_{H^{s+1}(\Tbb)} \bigl(\|  \underline{u_\zeta} (-1) \|_{H^{s}(\Tbb)} + \|  \underline{u} (-1)  \|_{H^{s}(\Tbb)} +1 \bigr)\Bigr)  \notag  \\
		\leq & C \sigma_\star\sigma.
\end{align}
By Lemma \ref{t:psest} and \eqref{e:dtest1}, we conclude
\begin{equation}\label{e:dtest5}
	\|\underline{\Psi}(-1)\|_{H^s(\Tbb)}  \leq C \|\uu(-1)\|_{H^{s-1}(\Tbb)} \leq C \sigma_\star\sigma.
\end{equation}
Gathering \eqref{e:dtest1}--\eqref{e:dtest5} together, by taking $\sigma_\star$ small enough and recalling $\U=(\uuo, \underline{u_\zeta},  \uu, \underline{\nu}, \underline{\Psi} )^T$,  yields
\begin{equation*}
	\|\U|_{\tau=-1}\|_{H^s(\Tbb)} \leq \sigma  .
\end{equation*}
Further, we are able to use the result in \S\ref{s:stp5} and conclude there exists a global unique solution $\U$ such that \eqref{e:sltdel} and the estimate \eqref{e:ineq1ab} hold for $\tau\in[-1,0)$, which, in turn by the Sobolev embedding theorems, leads to
\begin{equation}\label{e:Uest3}
	\|\U(\tau)\|_{W^{3,\infty}(\Tbb)} \leq C	\|\U(\tau)\|_{H^s(\Tbb)} \leq C \sigma  \; \Leftrightarrow \; 	\|(u_0(t), u_\zeta(t), u(t),\nu(t))\|_{W^{3,\infty}(\Tbb)} \leq C\sigma
\end{equation}
for $t\in[1,t_m)$. By \eqref{e:etrf1} and \eqref{e:etrf4}, since $u(t,\zeta)$ has a uniform estimate $-C\sigma<u(t,\zeta)<C\sigma$ for any $\zeta\in \Tbb$ (i.e., for any $|\bx|\in(0,\infty)$ by \eqref{e:xztsf} and the periodicity), we  obtain \eqref{e:mainest1} and \eqref{e:mainest2}.

On the other hand, by \eqref{e:etrf2} and \eqref{e:etrf3}, we reach
\begin{equation}\label{e:dtrhof}
	\frac{\del{t}\varrho(t,x^i)}{f_0(t)}-1= \Bigl(\underline{u_0}(\tau,\zeta)+\frac{5}{3} \underline{u_\zeta}(\tau,\zeta) \Bigr)-\frac{10 (1+f(t))}{3t f_0(t)} \underline{u_\zeta}(\tau,\zeta)  .
\end{equation}
In order to estimate \eqref{e:dtrhof}, we have to estimate $(1+f)/(tf_0)$. In fact, recalling Remark \ref{r:hyper}, we have point out that by the definition \eqref{e:Gdef} of $\chi(t)$ and  \eqref{e:gbA}, we obtain
\begin{equation}\label{e:limf2}
	\frac{f_0^2}{(1+f)^2} =\frac{f\chi}{Bt^2}   \;  \Rightarrow   \;   \frac{1+f}{t f_0} =\sqrt{\frac{B}{\chi f }} \;  \Rightarrow   \;  	\lim_{t\rightarrow t_m } \frac{1+f}{t f_0}  =0 .
\end{equation}

Therefore, by \eqref{e:Uest3},  \eqref{e:dtrhof} and \eqref{e:limf2}, we can estimate, for any $t\in[1,t_m)$,
\begin{equation*}
	\biggl\|\frac{\del{t}\varrho(t,x^i)}{f_0(t)}-1 \biggr\|_{W^{3,\infty}(\Rbb^3)}  \leq C\sigma  ,
\end{equation*}
which results in \eqref{e:mainest1b}
for $(t,x^i)\in[1,t_m) \times \Rbb^3$.

By \eqref{e:etrf3} and \eqref{e:Uest3}, we reach the estimate
\begin{equation*}
	\|x^i	\partial_{i}\varrho(t, x^i) \|_{W^{3,\infty}(\Rbb^3)} \leq C(1+f(t))\sigma,
\end{equation*}
which implies \eqref{e:mainest1c}.

In the end, since $
	\lim_{t\rightarrow t_m} f(t)=+\infty$ and $\lim_{t\rightarrow t_m} f_0(t)=+\infty $ from Theorem \ref{t:mainthm0}, we conclude \eqref{e:blup0}, and $\gamma>1/3$, the blowups occur at a finite time $t_m$ (see Theorem \ref{t:mainthm0}.$(5)$).
We then complete the proof of this theorem.
\end{proof}


\appendix


\section{Miscellaneous}
\subsection{Physical interpretations of $\mathcal{D}^i$ and $\mathcal{S}$}\label{s:DS} 
To clarify the physical significance of $\mathcal{D}^i$ and $\mathcal{S}$, we introduce the following physical quantities to rewrite \eqref{e:D1} and \eqref{e:S1} respectively. Let $\mathcal{T}$ denote the fixed temperature distribution, which satisfies the following:
\begin{equation*}
	\mathfrak{T} \propto |\bx|^2, \quad \text{i.e., } \quad 	\mathfrak{T} = \mathcal{k}(t) \delta_{kl}x^kx^l, \quad \text{for some time function} \; \mathcal{k}(t).
\end{equation*}
The \textit{entropy flux} and \textit{thermodynamic force} are given, respectively, by
\begin{equation*}
	\mathfrak{J}^i:= \Bigl(\frac{2}{3}+\omega\Bigr)\mathfrak{T} \rho \check{v}^i \AND 	\mathfrak{F}_i:=\del{i}    \ln \bigl(\rho \mathfrak{T}^{\frac{1}{ \frac{2}{3}+\omega} }  \bigr).
\end{equation*}
and we denote the \textit{relative velocity} due to the \textit{inhomogeneous densities} for the system \eqref{e:EP1}--\eqref{e:data0b},
\begin{equation*}
	\check{v}^i:=v^i-\Bigl(\frac{2}{3t}-\frac{f_0}{3(1+f)}\Bigr) x^i.
\end{equation*}

As perturbations of the homogeneous density give rise to solutions of the system \eqref{e:EP1}--\eqref{e:data0b} that reduce to the Newtonian universe \ref{c:Ntsl} and the homogeneous blowup solution \ref{c:hmsl}, perturbations of the inhomogeneous density cause deviations of the velocity from $\rv^i$ \eqref{e:exsol1} and $v_r^i$ \eqref{e:sl1}.

In terms of these quantities, $\mathcal{D}^i$ and $\mathcal{S}$ can be viewed as effects of the inhomogeneous density. Then with straightforward calculations, readers can verify that
\begin{align*}
	\mathcal{D}^i(t,x^j,\rho,v^k,s,\phi)= & 		\underbrace{ -\frac{\kappa f_0}{1+f} 	\check{v}^i}_{\text{damping due to inhomogeneous densities}}, 
	\\
	\mathcal{S}(t,x^j,\rho,v^k,s,\phi)	= & 	\underbrace{-\frac{1}{\rho}  \del{i}   \Bigl( \frac{\mathfrak{J}^i}{\mathfrak{T}} \Bigr) }_{\text{flux of entropy}}+  \underbrace{ \frac{\mathfrak{J}^i \mathfrak{F}_i }{\mathfrak{T} \rho}  }_{\text{entropy productions}}   .  
\end{align*} 
Thus, $\mathcal{D}^i$ can be interpreted as a \textit{damping term} caused by inhomogeneous densities, while $\mathcal{S}$ characterizes the \textit{growth of entropy} that arises from the inhomogeneous densities. Moreover, there is a temperature-dependent matter flow associated with this growth.

\subsection{Spherical coordinates}\label{s:SScd}
Let us introduce the \textit{spherical coordinates} by
\begin{equation*}
	X^1=R \sin \theta \cos \varphi, \quad X^2= R \sin \theta\sin \varphi \AND X^3=R \cos\theta ,
\end{equation*}
and the spherical coordinate is related to the Cartesian coordinate  by
\begin{align}\label{e:Spcrd2}
	R=\sqrt{\delta_{ij}X^iX^j}, \quad  \theta = \arccos\Bigl(\frac{X^3}{R}\Bigr) \AND \varphi= \arctan \Bigl(\frac{X^2}{X^1}\Bigr) .
\end{align}
The Jacobian matrix of this transform is
\begin{align}\label{e:Jac}
	\p{\frac{\partial R}{\partial X^1} & \frac{\partial R}{\partial X^2} &\frac{\partial R}{\partial X^3} \\
		\frac{\partial \theta}{\partial X^1} & \frac{\partial \theta}{\partial X^2} &\frac{\partial \theta}{\partial X^3} \\
		\frac{\partial \varphi}{\partial X^1} & \frac{\partial \varphi}{\partial X^2} &\frac{\partial \varphi}{\partial X^3} }= \p{\sin\theta \cos \varphi & \sin\theta\sin\varphi & \cos\theta \\
		\frac{1}{R} \cos\theta\cos\varphi & \frac{1}{R} \cos\theta\sin\varphi & -\frac{1}{R} \sin\theta\\
		-\frac{1}{R}\frac{\sin\varphi}{\sin\theta}& \frac{1}{R}\frac{\cos\varphi}{\sin\theta} & 0}  .
\end{align}
Then the derivatives become
\begin{equation*}\label{e:Di}
	D_i=\frac{\partial R}{\partial X^i} \del{R}+\frac{\partial \theta}{\partial X^i} \del{\theta}+\frac{\partial \varphi}{\partial X^i} \del{\varphi}=\frac{\delta_{ij} X^j}{R}\del{R}+\frac{\partial \theta}{\partial X^i} \del{\theta}+\frac{\partial \varphi}{\partial X^i} \del{\varphi}  .
\end{equation*}
That is, using the Jacobian matrix given in \eqref{e:Jac}, they are
\begin{align*}
	D_1 = & \frac{\delta_{1j} X^j}{R}\del{R}+\frac{\partial \theta}{\partial X^1} \del{\theta}+\frac{\partial \varphi}{\partial X^1} \del{\varphi}=\frac{\delta_{1j} X^j}{R}\del{R}+\frac{\cos \theta \cos\varphi}{R} \del{\theta}-\frac{\sin\varphi}{R \sin \theta} \del{\varphi}  ; \\
	D_2 =& \frac{\delta_{2j} X^j}{R}\del{R}+\frac{\partial \theta}{\partial X^2} \del{\theta}+\frac{\partial \varphi}{\partial X^2} \del{\varphi} =\frac{\delta_{2j} X^j}{R}\del{R}+\frac{\cos\theta\sin\varphi}{R} \del{\theta}+\frac{\cos\varphi}{R \sin\theta} \del{\varphi} ; \\
	D_3 = & \frac{\delta_{3j} X^j}{R}\del{R}+\frac{\partial \theta}{\partial X^3} \del{\theta}+\frac{\partial \varphi}{\partial X^3} \del{\varphi} =\frac{\delta_{3j} X^j}{R}\del{R} - \frac{\sin\theta}{R} \del{\theta} .
\end{align*}

\section{Tools of analyses}
\subsection{Analysis of a nonlinear ODE}\label{s:ODE}
One of the main tools we use is a type of ODEs developed in our previous article \cite[\S$2$]{Liu2022b}. For convenience, we provide the results without proofs in this appendix. Detailed proofs can be found in \cite[\S$2$]{Liu2022b}. Specifically, we consider solutions $f(t)$ to the following ODE:
\begin{gather}
	f^{\prime\prime}(t)+\frac{\ca}{t}  f^\prime(t)-\frac{\cb}{t^2} f(t)(1+  f(t))-\frac{\cc (  f^\prime(t))^2}{1+f(t)}=   0 , \label{e:feq0}\\
	f(t_0)= \mf>0 \AND
	f^\prime(t_0)=   \mf_0>0.  \label{e:feq1}
\end{gather}
where $\mf,\mf_0>0$ are positive constants and
\begin{equation}\label{e:abcdk}
	\ca>1, \quad \cb>0 \AND 1<\cc < 3/2 .
\end{equation}

\begin{remark}
The equation  \eqref{e:feq1b}--\eqref{e:feq2b} is  \eqref{e:feq0}--\eqref{e:feq1} when
\begin{equation*}
	\ca=\frac{4}{3}, \quad \cb=\frac{2}{3}, \quad \cc=\frac{4}{3} , \quad t_0=1 \AND \beta_0=3(1+\mf)\mg.
\end{equation*}
\end{remark}

From now on, to simplify the notations, we denote
\begin{equation*}
	\triangle:=\sqrt{(1-\ca)^2+4\cb}>-\ba, \; \ba=1-\ca<0, \; \bc=1-\cc<0
\end{equation*}
and introduce constants $\mathtt{A}$,  $\mathtt{B}$, $\mathtt{C}$, $\mathtt{D}$ and $\mathtt{E}$ depending on the initial data  $\mf$ and $\mf_0$ of \eqref{e:feq0}--\eqref{e:feq1} and parameters $\ca$, $\cb$ and $\cc$,
\begin{align*}
	\mathtt{A}:=&\frac{t_0^{   -\frac{\ba-\triangle}{2}  }}{\triangle}\biggl(  \frac{t_0   \mf_0}{(1+\mf)^2} - \frac{\ba+\triangle}{2}  \frac{\mf  } {1+\mf} \biggr), \\
	\mathtt{B}:= &  \frac{t_0^{-\frac{\ba+\triangle}{2} }}{\triangle} \Bigl( \frac{\ba-\triangle}{2}  \frac{\mf }{1+\mf}  -\frac{t_0 \mf_0}{(1+\mf)^2} \Bigr)<0, \\
	\mathtt{C}:= & \frac{2 } {2+\ba+\triangle} \Bigl( \ln( 1+\mf) +\frac{\ba+\triangle}{2\cb} \frac{t_0\mf_0}{1+\mf}\Bigr)t_0^{-\frac{\ba+\triangle}{2 }} >0,  \\
	\mathtt{D}:= &
	\frac{ \ba+\triangle   }{2+\ba+\triangle}  \Bigl(  \ln( 1+\mf)  - \frac{1 } { \cb} \frac{t_0\mf_0}{1+\mf}\Bigr) t_0,   \\
	\mathtt{E}:=  &
	\frac{\bc  \mf_0 t_0^{1-\ba}  }{  \ba  (1+\mf) } >0.
\end{align*}

We define the following two critical times $t_\star$ and $t^\star$.
\begin{definition}\label{t:tdef}
	Suppose $\mathtt{A},\;\mathtt{B}, \;\mathtt{E},\;\ba$ and $\triangle$ are defined above, then
	\begin{enumerate}
		\item
		Let $\mathcal{R}:=\{t_r>t_0 \;|\;\mathtt{A} t_r^{\frac{\ba-\triangle}{2} } + \mathtt{B} t_r^{\frac{\ba+\triangle}{2} } + 1 =0\}$ and define $t_\star:=\min \mathcal{R}$.
		\item If $t_0^{\ba}> \mathtt{E}^{-1}$, we define $t^\star :=    (t_0^{\ba}- \mathtt{E}^{-1} )^{1/\ba}\in(0,\infty)$, i.e.,  $t=t^\star$ solves $1-\mathtt{E}  t_0^{\ba} +  \mathtt{E}   t^{\ba}=0$.
	\end{enumerate}
\end{definition}

We are now in a position to state the main theorem regarding ODE \eqref{e:feq0}--\eqref{e:feq1}, with the proof available in \cite[\S$2$]{Liu2022b}. 
\begin{theorem}\label{t:mainthm0}
	Suppose constants $\ca$, $\cb$ and $\cc$ are defined by  \eqref{e:abcdk}, $t_\star$ and $t^\star$ are defined above and the initial data $\mf, \mf_0>0$, then
	\begin{enumerate}
		\item $t_\star \in[0,\infty)$ exists and $t_\star>t_0$;
		\item there is a constant $t_m\in [t_\star,\infty]$, such that there is a unique solution $f\in C^2([t_0,t_m))$ to the equation \eqref{e:feq0}--\eqref{e:feq1}, and
		\begin{equation*}
			\lim_{t\rightarrow t_m} f(t)=+\infty \AND \lim_{t\rightarrow t_m} f_0(t)=+\infty .
		\end{equation*}
		\item  $f$ satisfies upper and lower bound estimates,
		\begin{align*}
			1+f(t)>&\exp \bigl( \mathtt{C} t^{\frac{\ba+\triangle}{2} }  +\mathtt{D}  t^{-1}\bigr)      &&\text{for}\quad t\in(t_0,t_m);
			\\
			1+f(t) < & \bigl(\mathtt{A} t^{\frac{\ba-\triangle}{2} } + \mathtt{B} t^{\frac{\ba+\triangle}{2} } + 1 \bigr)^{-1}    && \text{for}\quad t\in(t_0,t_\star).
		\end{align*}
	\end{enumerate}		
	Furthermore, if the initial data satisfies
	$\mf_0 >  \ba(1+\mf) /(\bc t_0 )$,
	then
	\begin{enumerate}
		\setcounter{enumi}{3}
		\item 	$t_\star$  and $t^\star$  exist and finite, and  $t_0<t_\star<t^\star<\infty$;
		\item there is a finite time $t_m\in [t_\star,t^\star)$, such that there is a solution $f\in C^2([t_0,t_m))$ to the equation \eqref{e:feq0} with the initial data \eqref{e:feq1},   and 	\begin{equation*}
			\lim_{t\rightarrow t_m} f(t)=+\infty \AND \lim_{t\rightarrow t_m} f_0(t)=+\infty .
		\end{equation*}
		\item the solution $f$ has improved lower bound estimates, for $t\in(t_0,t_m)$,
		\begin{equation*}\label{e:ipvest}
			(1+\mf)  \bigl(1-\mathtt{E}  t_0^{\ba} +  \mathtt{E}   t^{\ba} \bigr)^{1/\bc}  < 1+f(t) .
		\end{equation*}
	\end{enumerate}
\end{theorem}

\subsection{The analysis of the reference solutions} \label{t:refsol}
\subsubsection{The time transform function $g(t)$ and rough estimates of $f$ and  $\del{t}f(t)$}\label{t:ttf}
Let
\begin{align}\label{e:gdef0a}
	g(t):=\exp\Bigl(-A\int^t_{t_0} \frac{f(s)(f(s)+1)}{s^2 f_0(s)} ds \Bigr)>0
\end{align}
where $A\in(0,2\cb/(3-2\cc))$ is a constant.
The following lemma gives an alternative representation of $g(t)$ only involving $f$ without $f_0$, fundamental properties of the function $g(t)$ and expresses $f_0$ in terms of $f$ and $g$.

\begin{lemma}\label{t:f0fg}
	Suppose $f\in C^2([t_0,t_1))$ ($t_1>t_0$) solves the equation \eqref{e:feq0}--\eqref{e:feq1}, $g(t)$ is defined by \eqref{e:gdef}, and  denote $f_0(t):=\del{t} f(t)$, then
	\begin{enumerate}	
		\item  $f_0$ can be expressed by
		\begin{equation}
			f_0(t)=B^{-1} t^{-\ca} g^{-\frac{\cb}{A}}(t)(1+f(t))^{\cc} >0  \label{e:f0aa}
		\end{equation}
		for $t\in [t_0,t_1)$ where $B:= (1+\mf)^\cc/( t_0^{ \ca} \mf_0)>0 $ is a constant depending on the data;
		\item If the data $\mf>0$, then $f(t)>0$ for $t \in[t_0,t_1)$;
		\item $g(t)$ can be represented by
		\begin{equation}\label{e:gdef2}
			g (t) =\Bigl(1+ \cb B \int^t_{t_0} s^{\ca-2} f(s)(1+f(s))^{1-\cc}  ds \Bigr)^{-\frac{A}{\cb}}\in(0,1],
		\end{equation}
		for $t\in[t_0,t_1)$,  and $g(t_0)=1$; 	
		\item
		$g(t)$ is strictly decreasing and  invertible in $[t_0,t_1)$.
	\end{enumerate}
\end{lemma}

\begin{remark}\label{t:dtg1}
	There are two useful identities from \cite[eqs. $(2.7)$ and $(2.8)$]{Liu2022b}. We list them in this remark.
	\begin{align}
		\del{t}g(t) = &-A  B g^{\frac{\cb}{A}+1}(t)  t^{\ca-2} f(t) (1+f(t))^{1-\cc} , \label{e:dtg0} \\
		\del{t} g^{-\frac{\cb}{A}}(t) = & -\frac{\cb}{A}g^{-\frac{\cb}{A}-1}\del{t} g = \cb B t^{\ca-2} f(t)  (1+f(t))^{1-\cc}.  \label{e:dtgf}
	\end{align}
\end{remark}

\subsubsection{Estimates of two crucial quantities $\chi(t)$ and $\xi(t)$}
In this section, we  estimate two important quantities $\chi(t)$ and $\xi(t)$ which are  frequently used (see \cite{Liu2022b} for details).  The first quantity is defined by
\begin{equation}\label{e:Gdef0}
	\chi(t):=\frac{t^{2-\ca} f_0(t)}{(1+f(t))^{2-\cc} f(t) g^{\frac{\cb}{A}}(t)} \overset{\eqref{e:f0aa}}{=} \frac{  g^{-\frac{2\cb}{A}}(t) t^{2(1-\ca)}}{B f(t) (1+f(t))^{2(1-\cc)}}  .
\end{equation}

\begin{lemma}\label{t:gmap}
	Suppose $g(t)$ is defined by \eqref{e:gdef0a} and  $\cc\in(1,3/2)$,  and  $f\in C^2([t_0,t_m))$ (where $[t_0,t_m)$ is the maximal interval of existence of $f$ given by Theorem \ref{t:mainthm0}) solves ODE \eqref{e:feq0}--\eqref{e:feq1}. Then $
	\lim_{t\rightarrow t_m}g(t)=0$.
\end{lemma}
\begin{remark}
	Due to this lemma, it is convenient to continuously extend $g(t)$ from $[t_0,t_m)$ to $[t_0,t_m]$ by letting $g(t_m):=\lim_{t\rightarrow t_m}g(t)=0$, then $g^{-1} (0)=t_m$.
\end{remark}

\begin{proposition}\label{t:limG}
	Suppose $\cc\in(1,3/2)$, $\cb>0$, $\ca>1$,  $\chi$ is defined by \eqref{e:Gdef0} and  $f\in C^2([t_0,t_m))$ (where $[t_0,t_m)$ is the maximal interval of existence of $f$ given by Theorem \ref{t:mainthm0}) solves ODE \eqref{e:feq0}--\eqref{e:feq1}.
	Then there is a function $\mathfrak{G} \in C^1([t_0,t_m))$, such that for $t\in [t_0,t_m)$,
	\begin{equation}\label{e:limG}
		\chi(t)=\frac{2\cb B}{3-2\cc}+\mathfrak{G}(t)
	\end{equation}
	where $\lim_{t\rightarrow t_m}\mathfrak{G}(t)=0$.
	Moreover, there is a constant $C_\chi>0$ such that $0<\chi(t) \leq C_\chi$ in $[t_0,t_m)$, and there are continuous extensions of $\chi$ and $\mathfrak{G}$ such that $\chi\in C^0([t_0,t_m])$ and $\mathfrak{G}\in C^0([t_0,t_m])$ by letting $\chi(t_m):=2\cb B/(3-2\cc)$ and $\mathfrak{G}(t_m):=0$.
\end{proposition}

The second crucial quantity is
\begin{equation}\label{e:xidef}
	\xi(t):=1/[g(t) (1+f(t)) ],
\end{equation}
the next proposition prove $\xi$ is bounded and its limit vanishes as $t$ tends to $t_m$.
\begin{proposition}\label{t:fginv0}
	Suppose $f\in C^2([t_0,t_m))$ (where $[t_0,t_m)$ is the maximal interval of existence of $f$ given by Theorem \ref{t:mainthm0}) solves ODE \eqref{e:feq0}--\eqref{e:feq1}, $g(t)$ is defined by \eqref{e:gdef} and $\xi(t)$ is given by \eqref{e:xidef}, then $\xi\in C^1([t_0,t_m))$ and
	\begin{equation}\label{e:fginv}
		\lim_{t\rightarrow t_m} \xi(t)= 0 .
	\end{equation}
	Moreover, there is a constant $C_\star>0$, such that $0<\xi(t)  \leq C_\star$ for every $t\in[t_0,t_m)$, and there is a continuous extension of $\xi$ such that  $\xi\in C^0([t_0,t_m])$ by letting $\xi(t_m):=0$.
\end{proposition}

\subsubsection{The generalization of Proposition \ref{t:fginv0} and the expression of $\del{t}\chi$}\label{s:dtchi}
The derivative $\del{t}\chi$ is required in \S\ref{s:stp4}. Thus we calculate it in the following lemma.
\begin{lemma}\label{t:dtchi}
	Suppose $\chi(t)$ is defined by \eqref{e:Gdef0} and $\mathfrak{G}$ is given by \eqref{e:limG}, then
	\begin{equation}\label{e:dtchi}
		\del{t}\chi= \del{t}\mathfrak{G} = -\frac{(3-2\cc)  \mathfrak{G} f^{\frac{1}{2}} \chi^{\frac{1}{2}}}{B^{\frac{1}{2}} t}   - \frac{\chi^{\frac{3}{2}}}{B^{\frac{1}{2} }t f^{\frac{1}{2}}}   +2(1-\ca) \frac{\chi }{t} .
	\end{equation}
\end{lemma}
\begin{proof}
	Differentiating \eqref{e:Gdef0} with respective $t$, using \eqref{e:f0aa}	and \eqref{e:dtg0} to replace $f_0$ and $\del{t}g(t)$, respectively, yield
	\begin{align*}
		\del{t}\chi
		\overset{\eqref{e:f0aa} \& \eqref{e:dtg0}}{=}   &  -\frac{2 (1-\cc) t^{2 -3\ca } (f+1)^{3\cc-3}  g^{-\frac{3 \cb}{A}}}{B^2 f }   -\frac{t^{2 -3a } (f +1)^{3\cc-2}    g^{-\frac{3 \cb}{A}}}{B^2 f^2} \notag  \\
		& + 2 b   g^{-\frac{\cb}{A}}  t^{-\ca}  (1+f )^{\cc-1}     +\frac{2 (1-\ca) t^{1-2\ca} (f +1)^{-2 (1-\cc)} g^{-\frac{2 \cb}{A}}}{B f }
		\notag  \\
		\overset{ \eqref{e:Gdef0}}{=}   &  -\frac{(3-2\cc) f^{\frac{1}{2}} \chi^{\frac{3}{2}}}{B^{\frac{1}{2}} t}   - \frac{\chi^{\frac{3}{2}}}{B^{\frac{1}{2} }t f^{\frac{1}{2}}}   + 2 b t^{-1}  B^{\frac{1}{2}} f^{\frac{1}{2}} \chi^{\frac{1}{2}} +2(1-\ca) \frac{\chi}{t}	 \notag  \\
		\overset{\eqref{e:limG}}{=} &  -\frac{(3-2\cc)  f^{\frac{1}{2}} \chi^{\frac{1}{2}}}{B^{\frac{1}{2}} t} \Bigl(\frac{2\cb B}{3-2\cc}+\mathfrak{G}\Bigr)  - \frac{\chi^{\frac{3}{2}}}{B^{\frac{1}{2} }t f^{\frac{1}{2}}}   + 2 b t^{-1}  B^{\frac{1}{2}} f^{\frac{1}{2}} \chi^{\frac{1}{2}} +2(1-\ca) \frac{\chi}{t}	 \notag  \\
		= &   -\frac{(3-2\cc)  \mathfrak{G} f^{\frac{1}{2}} \chi^{\frac{1}{2}}}{B^{\frac{1}{2}} t}   - \frac{\chi^{\frac{3}{2}}}{B^{\frac{1}{2} }t f^{\frac{1}{2}}}   +2(1-\ca) \frac{\chi}{t} .
	\end{align*}
We proved this lemma.
\end{proof}

The following Proposition generalizes Proposition \ref{t:fginv0} by considering
the following crucial quantity
\begin{equation}\label{e:xidef2}
	\eta_\theta(t):=1/[g^\theta (t) (1+f(t)) ],
\end{equation}
the next proposition prove $\eta_\theta$ is bounded and its limit vanishes as $t$ tends to $t_m$.
\begin{proposition}\label{t:fginv2}
	Suppose $\theta\geq 1$ is a  constant and $A\theta <  2\cb/(3-2\cc)$ and $f\in C^2([t_0,t_m))$ (where $[t_0,t_m)$ is the maximal interval of existence of $f$ given by Theorem \ref{t:mainthm0}) solves ODE \eqref{e:feq0}--\eqref{e:feq1}, $g(t)$ is defined by \eqref{e:gdef0a} and $\eta_\theta(t)$ is given by \eqref{e:xidef2}, then $\eta_\theta\in C^1([t_0,t_m))$ and
	\begin{equation}\label{e:fginv2b}
		\lim_{t\rightarrow t_m} \eta_\theta(t)= 0 .
	\end{equation}
	Moreover, there is a constant $C_\star>0$, such that $0<\eta_\theta(t)  \leq C_\star$ for every $t\in[t_0,t_m)$, and there is a continuous extension of $\eta_\theta$ such that  $\eta_\theta\in C^0([t_0,t_m])$ by letting $\eta_\theta(t_m):=0$.
\end{proposition}
\begin{proof}
	By the definition  \eqref{e:Gdef0} of $\chi(t)$ and Proposition \ref{t:limG} (i.e., $0<\chi\leq C_\chi$),
	\begin{equation}\label{e:gest2}
		0< g^{-\frac{2\cb}{A}}(t)  \leq C_\chi B f(t) (1+f(t))^{2(1-\cc)} t^{2(\ca-1)}.
	\end{equation}
	then for $\theta>1$
	\begin{equation}\label{e:gest2}
		0< g^{-\theta\frac{2\cb}{A}}(t)  \leq C_\chi^\theta B^\theta f^\theta (1+f)^{2\theta(1-\cc)} t^{2\theta(\ca-1)}.
	\end{equation}
	Since $1+f>f$, using \eqref{e:gest2}, we arrive at
	\begin{align*}
		\eta^{\frac{2\cb}{A}} =g^{-\theta \frac{2\cb}{A}}  (1+f)^{-\frac{2\cb}{A}} <   \frac{ (1+f )^{1-\frac{2\cb}{A}} }{g^{\theta\frac{2\cb}{A}} f}\leq B^\theta C_\chi^\theta  t^{2\theta(\ca-1)}(1+f)^{(2\theta+1)-2\theta\cc-\frac{2\cb}{A}} f^{\theta-1}.
	\end{align*}
	That is,
	\begin{equation}\label{e:gfest}
		0<\eta<   B^{\frac{A}{2\cb}\theta}  C_\chi^{\frac{A}{2\cb}\theta} t^{A\theta (\ca-1)/\cb}(1+f)^{\frac{A}{2\cb}(2\theta+1-2\theta\cc-\frac{2\cb}{A})} f^{\frac{A}{2\cb}(\theta-1)}<   B^{\frac{A}{2\cb}\theta}  C_\chi^{\frac{A}{2\cb}\theta} t^{A\theta (\ca-1)/\cb}(1+f)^{\frac{A}{2\cb}(3\theta-2\theta\cc-\frac{2\cb}{A})}  .
	\end{equation}
	Since $\cc\in(1,3/2)$ and $A\theta<2\cb/(3-2\cc)$, we have $3\theta-2\theta\cc-\frac{2\cb}{A}<0$.

	Then $(1)$ if $t_m<\infty$, directly using \eqref{e:gfest} and Theorem \ref{t:mainthm0} (i.e., $\lim_{t\rightarrow t_m} f(t)=+\infty$), we conclude \eqref{e:fginv2b}. $(2)$ If $t_m=\infty$, the right hand of \eqref{e:gfest} can be estimated, due to $3\theta-2\theta\cc-\frac{2\cb}{A}<0$ and Theorem \ref{t:mainthm0}.$(3)$, by
	\begin{equation*}
		t^{A\theta(\ca-1)/\cb}(1+f)^{\frac{A}{2\cb}((3-2\cc)\theta-\frac{2\cb}{A})} < t^{A\theta(\ca-1)/\cb} \exp\Bigl[\frac{A}{2\cb}\Bigl((3-2\cc)\theta -\frac{2\cb}{A}\Bigr) \bigl( \mathtt{C} t^{\frac{\ba+\triangle}{2} }  +\mathtt{D}  t^{-1}\bigr)\Bigr].
	\end{equation*}
	Then, using the fact $\lim_{x\rightarrow \infty}(x^a/e^x)=0$,
	\begin{align*}
		& \lim_{t \rightarrow \infty} \biggl(t^{A\theta(\ca-1)/\cb} \exp\Bigl[\frac{A}{2\cb}\Bigl((3-2\cc) \theta -\frac{2\cb}{A}\Bigr) \bigl( \mathtt{C} t^{\frac{\ba+\triangle}{2} }  +\mathtt{D}  t^{-1}\bigr)\Bigr]\biggr) \notag  \\
		&\hspace{3cm}=\lim_{t \rightarrow \infty} \Biggl(\frac{t^{A\theta(\ca-1)/\cb}}{\exp\bigl[-\frac{A}{2\cb}\bigl((3-2\cc)\theta-\frac{2\cb}{A}\bigr)  \mathtt{C} t^{\frac{\ba+\triangle}{2} }   \bigr]} \Biggr)=0.
	\end{align*} 
	This implies that the right-hand side of \eqref{e:gfest} approaches $0$ as $t \rightarrow t_m$, ultimately leading to \eqref{e:fginv2b}. Next, since $\eta_\theta(t)$ is a continuous function on $[t_0,t_m)$, we can apply a similar proof as presented in Proposition \ref{t:fginv0} to conclude the boundedness of $\eta_\theta(t)$ and therefore complete the proof. 
\end{proof}

\begin{corollary}\label{s:gf1/2}
	Suppose $\ca=4/3$, $\cb=2/3$,  $\cc=4/3$ and $0<
	A<2$, then
	\begin{equation*}
	\lim_{t\rightarrow t_m} \biggl(\frac{1}{g f^{\frac{1}{2}}}\biggr)= 0 .
	\end{equation*}
\end{corollary}
\begin{proof}
 Taking $\theta=2$ and by using Proposition \ref{t:fginv2} with parameters $\ca=4/3$, $\cb=2/3$,  $\cc=4/3$, and $0<A<2$, we obtain
	\begin{equation*}
		\lim_{t\rightarrow t_m} \biggl(\frac{1}{g^2(1+f)}\biggr) =0 \Rightarrow  	\lim_{t\rightarrow t_m}  \biggl(\frac{1}{g f^{\frac{1}{2}}}\biggr)  =	\lim_{t\rightarrow t_m} \biggl(\frac{1}{g^2(1+f)}\Bigl(1+\frac{1}{f}\Bigr)\biggr)^{\frac{1}{2}} =0 ,
	\end{equation*}
which completes the proof.
\end{proof}

\subsection{Cauchy problems of Fuchsian systems}\label{s:fuc}

In this Appendix, we introduce the Cauchy problem for the Fuchsian system, which is a variation of the one originally established in \cite[Appendix B]{Oliynyk2016a} and significantly developed in \cite{Beyer2020,Beyer2020b}. While the proof has been omitted here, readers can find detailed proofs in \cite{Beyer2020} and explore other generalizations and applications in, for example, \cite{Liu2018,Liu2018b,Liu2018a,Liu2022,Liu2022a,Liu2022b,Beyer2021}.

Consider the following symmetric hyperbolic system.
\begin{align}
	B^{\mu}(t,x,u)\partial_{\mu}u =&\frac{1}{t}\textbf{B}(t,x,u)\textbf{P}u+H(t,x,u)+|t|^{-\frac{1}{2}} F(t,x, u) \quad &&\text{in}\;[T_{0},T_{1})\times\mathbb{T}^{n},  \label{e:model1}\\
	u =&u_{0} &&\text{in}\;\{T_{0}\}\times\mathbb{T}^{n},\label{e:model2}
\end{align}
where $T_{0}<T_{1}\leq0$. For given $R>0$, we require the following \textbf{Conditions}\footnote{The notations used in this Appendix, such as $\mathrm{O}(\cdot)$, $\mathcal{O}(\cdot)$, and $B_R(\Rbb^N)$, are defined in \cite[\S$2.4$]{Beyer2020}. }:
\begin{enumerate}[leftmargin=*,label={(F\arabic*)}]
	\item \label{c:2} $\textbf{P}$ is a constant, symmetric projection operator, i.e., $\textbf{P}^{2}=\textbf{P}$, $\textbf{P}^{T}=\textbf{P}$ and $\partial_\mu \textbf{P}=0$.
	
	\item \label{c:3} $u=u(t,x)$ and $H(t,x,u)$ are $\mathbb R^{N}$-valued maps, $H, \; F\in C^{0}([T_{0},0],C^{\infty}(\Tbb^n\times B_R(\mathbb R^{N}),\Rbb^N))$ and satisfy $H(t,x,0)=0$ and there are constants $\lambda_1,\;\lambda_2 \geq 0$, such that
	\begin{equation}\label{e:Fcd}
		\mathbf{P} F(t,x,u)=\mathcal{O}(\lambda_1 u) \AND \mathbf{P}^\perp F(t,x,u)=\mathcal{O}(\lambda_2 \mathbf{P} u)
	\end{equation}
for all $ (t,x,u) \in[T_{0},0]\times \Tbb^n\times B_R(\Rbb^N) $ 	where
$
\textbf{P}^{\bot}=\mathds{1} - \textbf{P}
$
is the complementary projection operator.
	
	\item \label{c:4} $B^{\mu}=B^{\mu}(t,x,u)$ and $\textbf{B}=\textbf{B}(t,x,u)$ are $\mathbb M_{N\times N}$-valued maps, and  $B^{i}\in   C^{0}([T_{0},0),C^{\infty}(\Tbb^n\times  B_R(\mathbb R^{N}), \mathbb M_{N\times N})$,  $\textbf{B}\in   C^{0}([T_{0},0],C^{\infty}(\Tbb^n\times  B_R(\mathbb R^{N}), \mathbb M_{N\times N})$, $B^{0}\in C^{1}([T_{0},0),C^{\infty}(\Tbb^n\times  B_R(\mathbb R^{N}), \mathbb M_{N\times N})$ and they satisfy
	\begin{equation*}\label{e:comBP}
		(B^{\mu})^{T}=B^{\mu},\quad [\textbf{P}, \textbf{B}]=\textbf{PB}-\textbf{BP}=0.
	\end{equation*}
	and $	B^i$ can be expanded as
	\begin{equation*}
		B^i(t,x,u)=B_0^i(t,x,u)+\frac{1}{t} B_2^i(t,x,u)
	\end{equation*}
	where $B_0^i, B_2^i\in C^{0}([T_{0},0],C^{\infty}(\Tbb^n\times  B_R(\mathbb R^{N}), \mathbb M_{N\times N}))$.
	
	Suppose there is $\tilde{B}^0, \tilde{\mathbf{B}} \in C^0([T_0,0], C^\infty(\Tbb^n, \mathbb M_{N\times N}))$, such that
	\begin{equation*}
		[\mathbf{P}, \tilde{\mathbf{B}}] =0, \quad
		B^0(t,x,u)-\tilde{B}^0(t,x)= \mathrm{O}(u) , \quad
		\mathbf{B}(t,x,u)-\tilde{\mathbf{B}} (t,x)= \mathrm{O}(u)
	\end{equation*}
	for all $ (t,x,u) \in[T_{0},0]\times \Tbb^n\times B_R(\Rbb^N) $.
		
	Moreover, there is $\tilde{B}_2^i \in C^0([T_0,0], C^\infty(\Tbb^n, \mathbb M_{N\times N}))$, such that
	\begin{gather*}
		\mathbf{P} B_2^i(t,x,u) \mathbf{P}^\perp =   \mathrm{O}(\mathbf{P} u), \quad
		\mathbf{P}^\perp B_2^i(t,x,u) \mathbf{P} =   \mathrm{O}(\mathbf{P} u),  \\
		\mathbf{P}^\perp B_2^i(t,x,u) \mathbf{P}^\perp =   \mathrm{O}(\mathbf{P} u\otimes \mathbf{P} u),\quad
		\mathbf{P}  (B_2^i(t,x,u)-\tilde{B}_2^i(t,x)) \mathbf{P}  =   \mathrm{O}( u),
	\end{gather*}
	for all $ (t,x,u) \in[T_{0},0]\times \Tbb^n\times B_R(\Rbb^N) $.
	
	\item \label{c:5}
	There exists constants $\kappa,\,\gamma_{1},\,\gamma_{2}$ such that
	\begin{equation*}
		\frac{1}{\gamma_{1}}\mathds{1}\leq B^{0}\leq \frac{1}{\kappa} \textbf{B} \leq\gamma_{2}\mathds{1} \label{e:Bineq}
	\end{equation*}
	for all $ (t,x,u) \in[T_{0},0]\times \Tbb^n\times B_R(\Rbb^N) $.
	
	\item \label{c:6} For all $(t,x,u)\in[T_{0},0]\times\Tbb^n \times B_R(\mathbb R^{N})$,
	\begin{equation*}
		\textbf{P}^{\bot}B^{0}(t,\textbf{P}^\perp u)\textbf{P}=\textbf{P}B^{0}(t,\textbf{P}^\perp u)\textbf{P}^{\bot}=0.
	\end{equation*}
	\item \label{c:7}
	There exist
	constants $\theta$ and $\beta_{\ell}\geq 0$, $\ell=0,\cdots,7$, such that
	\begin{align*}
		\mathrm{div} B(t,x, u,w):= &  \del{t} B^0 (t,x,u)+D_u B^0(t,x,u) \cdot(B^0(t,x, u))^{-1}\Bigl[-B^i(t,x,u)\cdot w_i  \notag  \\
		& +\frac{1}{t} \mathbf{B}(t,x,u)\mathbf{P} u+H(t,x,y) +|t|^{-\frac{1}{2}}F(t,x,u)\Bigr]+\del{i}B^i(t,x,u) \notag  \\
		& +D_u B^i(t,x,u) \cdot w_i  ,
	\end{align*}
where $w=(w_i)$ and $(t,x,u, w) \in [T_0,0)\times \Tbb^n \times B_R(\Rbb^N) \times B_R(\mathbb{M}_{N \times n})$, satisfies
	\begin{align}
 \mathbf{P} 	\mathrm{div} B \mathbf{P}  = & \;\mathcal{O}\bigl(\theta   +|t|^{-\frac{1}{2}}\beta_0+|t|^{-1}\beta_1 \bigr),  \label{e:PhP1}\\
 \mathbf{P} 	\mathrm{div} B \mathbf{P}^\perp  = &\; \mathcal{O}\biggl(\theta  +|t|^{-\frac{1}{2}} \beta_2 +\frac{|t|^{-1}\beta_3}{R}\mathbf{P}u  \biggr),  \label{e:PhP2}\\
\mathbf{P}^\perp 	\mathrm{div} B \mathbf{P}  = & \; \mathcal{O}\biggl(\theta  +|t|^{-\frac{1}{2}} \beta_4  +\frac{|t|^{-1}\beta_5}{R}\mathbf{P}u   \biggr) \label{e:PhP3}
		\intertext{and}
\mathbf{P}^\perp  	\mathrm{div} B \mathbf{P}^\perp  =  & \; \mathcal{O}\biggl(\theta  +\frac{|t|^{-\frac{1}{2}}\beta_6}{R }\mathbf{P}u   +\frac{|t|^{-1}\beta_7}{R^2}\mathbf{P}u \otimes \mathbf{P}u \biggr)  . \label{e:PhP4}
	\end{align}	
\end{enumerate}

Now let us present the global existence theorem for the Fuchsian system (see \cite{Beyer2020} for detailed proofs).
\begin{theorem}\label{t:fuc}
	Suppose that $k\in\Zbb_{> \frac{n}{2}+3}$, $u_{0}\in H^{k}(\mathbb T^{n})$ and conditions \ref{c:2}--\ref{c:7} are fulfilled, and the constants $\kappa, \gamma_1,  \beta_\ell$ ($\ell=0,1,\cdots,7$) from the conditions \ref{c:2}--\ref{c:7} satisfy
	\begin{equation}\label{e:kpa2}
		\kappa>\frac{1}{2} \gamma_1 \max\Bigl\{\sum^3_{\ell=0} \beta_{2\ell+1}, \beta_1+2k(k+1)  \mathtt{b} \Bigr\}
	\end{equation}
where
\begin{equation*}
	\mathtt{b}:= \sup_{T_0\leq t<0} \bigl(\||\mathbf{P} \tilde{\mathbf{B}} D (\tilde{\mathbf{B}}^{-1} \tilde{B}^0)(\tilde{B}^0)^{-1} \mathbf{P} \tilde{B}_2^i \mathbf{P}|_{\mathrm{op}}\|_{\Li}+\||\mathbf{P}\tilde{\mathbf{B}} D (\tilde{\mathbf{B}}^{-1} \tilde{B}_2^i)\mathbf{P}|_{\mathrm{op}}\|_{\Li}\bigr).
\end{equation*}
	Then there exist constants $\delta_0, \delta>0$ satisfying $\delta<\delta_0$, such that if
	\begin{equation*}
		\|u_0\|_{H^k} \leq \delta,
	\end{equation*}
	then there exists a unique solution
	\begin{equation*}
		u\in C^0([T_0,0),H^k(\Tbb^n) ) \cap C^1([T_0,0),H^{k-1}(\Tbb^n) ) \cap \Li ([T_0,0),H^k(\Tbb^n))
	\end{equation*}
of the initial value problem \eqref{e:model1}--\eqref{e:model2} such that $\Pbp u(0):=\lim_{t\nearrow 0}\Pbp u(t)$ exists in $H^{s-1}(\Tbb^n)$.

\noindent Moreover, for $T_0\leq t<0$, the solution $u$ satisfies the energy estimate
\begin{equation*}\label{e:ineq1}
	\|u(t)\|_{H^k}^2  - \int^t_{T_0}\frac{1}{\tau}\|\mathbf{P} u(\tau)\|^2_{H^k}d\tau \leq C(\delta_0,\delta_0^{-1}) \|u_0\|^2_{H^k} .
\end{equation*}
\end{theorem}

\section*{Acknowledgement}
		C.L. is supported by the Fundamental Research Funds for the Central Universities, HUST: $5003011036$, $5003011047$, and the National Natural Science Foundation of China (NSFC) under the
		Grant No. $11971503$.

\bigskip

\textbf{Data Availability} Data sharing is not applicable to this article as no datasets were
generated or analyzed during the current study.

\bigskip

\textbf{Declarations}

\bigskip

\textbf{Conflict of interest} The authors declare that they have no conflict of interest.

\bibliographystyle{amsplain}
\bibliography{Reference_Chao}

\end{document}